\documentclass[11pt]{article}
\usepackage{amsmath, amssymb}

\usepackage{hyperref}
\usepackage{comment}
\usepackage{authblk}

\usepackage{tikz}

\usetikzlibrary{automata, positioning, arrows,matrix, calc}
\usetikzlibrary{decorations.markings,arrows.meta,bending}
\tikzset{
        ->,  
        >=stealth',
        node distance=3cm, 
        every state/.style={thick, fill=gray!10}, 
        initial text=$ $,
        }

\tikzstyle{blocknf} = [minimum height=2em,minimum width=2em]
\tikzstyle{block} = [draw,rectangle,thick,minimum height=2em,minimum width=2em]
\tikzstyle{connector} = [->,thick]

\newtheorem{theorem}{Theorem}[section]
\newtheorem{proposition}[theorem]{Proposition}
\newtheorem{lemma}[theorem]{Lemma}
\newtheorem{corollary}[theorem]{Corollary}
\newtheorem{DefAux}[theorem]{Definition}
\newenvironment{definition}{\begin{DefAux} \rm}{\end{DefAux}}
\newtheorem{ExAux}[theorem]{Example}
\newenvironment{example}{\begin{ExAux} \rm}{\end{ExAux}}
\newtheorem{NotAux}[theorem]{Notation}

\newtheorem{ProbAux}[theorem]{Problem}

\newcommand{\tombstone}{\smash{\qquad\qquad \mbox{$\Box$}\medskip}}
\newenvironment{proof}{\smallskip\noindent{\bf Proof.\ }}{\tombstone 

\medskip}

\newcommand{\rowmult}{\mathbin{\begin{tikzpicture}
  \def\pLength{3.5pt}   
  \def\startWidth{0.0pt}
  \def\endWidth{.6pt}  
  \foreach \angle in {0, 45, 135, 180, 225, 315} {
    \begin{scope}[rotate=\angle]
      \filldraw[draw=black, fill=black, ultra thin] (0, \startWidth) 
        -- (\pLength-\endWidth, \endWidth) 
        arc (120:-120:\endWidth) 
        -- (0, -\startWidth) 
        -- cycle;
    \end{scope}
  }
\end{tikzpicture}}}

\newcommand{\colmult}{\mathbin{\begin{tikzpicture}
  \def\pLength{3.5pt}   
  \def\startWidth{0.0pt}
  \def\endWidth{.6pt}  
  \foreach \angle in { 45, 90, 135,  225, 270, 315} {
    \begin{scope}[rotate=\angle]
      \filldraw[draw=black, fill=black, ultra thin] (0, \startWidth) 
        -- (\pLength-\endWidth, \endWidth) 
        arc (120:-120:\endWidth) 
        -- (0, -\startWidth) 
        -- cycle;
    \end{scope}
  }
\end{tikzpicture} }}

\newcommand{\mmmult}{\mathbin{\begin{tikzpicture}
  \def\pLength{3.5pt}   
  \def\startWidth{0.0pt}
  \def\endWidth{.6pt}  
  \foreach \angle in {0, 45, 90, 135, 180, 225, 270, 315} {
    \begin{scope}[rotate=\angle]
      \filldraw[draw=black, fill=black, ultra thin] (0, \startWidth) 
        -- (\pLength-\endWidth, \endWidth) 
        arc (120:-120:\endWidth) 
        -- (0, -\startWidth) 
        -- cycle;
    \end{scope}
  }
\end{tikzpicture} }}

\newcommand{\Idc}{I_c}
\newcommand{\Idr}{I_r}
\newcommand{\ic}{{\mathsf{i}_{c}}}
\newcommand{\ir}{{\mathsf{i}_{r}}}
\newcommand{\ltacti}{\overset{\kern-2pt\scriptscriptstyle{\circ}}{\curvearrowright}}
\newcommand{\rtacti}{\overset{\kern2pt\scriptscriptstyle{\circ}}{\curvearrowleft}}
\newcommand{\ltact}{\mathrel{\curvearrowright}}
\newcommand{\rtact}{\mathrel{\curvearrowleft}}
\newcommand{\transp}{T}
\newcommand{\take}[1]{{\langle #1 \rangle}}
\newcommand{\un}[1]{\mathfrak{U}(#1)}
\newcommand{\ut}[1]{\mathfrak{UT}(#1)}

\newcommand{\Z}{\mathbb Z}
\newcommand{\Aut}{\mathop{\rm Aut}\nolimits}

\newcommand{\Sym}{\mathop{\rm Sym}\nolimits}
\newcommand{\Tr}{\mathop{\rm Tr}\nolimits}

\newcommand{\A}{\mathcal A}

\newcommand{\dz}{\mathfrak d}
\newcommand{\iz}{\mathfrak i}

\newcommand{\gridpoints}[2]
{\foreach \y in {1,..., #1} \foreach \x in {1,...,#2}
   \fill[black] (\x,\y) circle(5pt) ; }     

\newcommand*\circled[1]{\tikz[baseline=(char.base)]{
            \node[shape=circle,draw,inner sep=1pt] (char) {#1};}}

\title{Combinatorial structures connecting Latin squares and bireversible automata}

\author[1]{Brian Curtin}
\affil[1]{Department of Mathematics and Statistics\\
        University of South Florida\\
        4202 E Fowler Ave\\
        Tampa, FL 33620-5700\\
        \href{mailto:bcurtin@usf.edu}{bcurtin@usf.edu}, \href{mailto:savchuk@usf.edu}{savchuk@usf.edu}}

\author[1]{Dmytro Savchuk}

\begin{document}

\maketitle

\begin{abstract} 
This paper explores the theory of letter transducers, Mealy automata, and bireversible automata from a combinatorial perspective analogous to the theory of Latin squares.
We view the sets of transitions of letter transducers as analogs of orthogonal arrays, and discuss two other combinatorial encodings of Mealy automata analogous to orthogonal pairs of Latin squares and to $(k,n)$-nets. 
We characterize various classes of automata (Mealy, reversible, invertible, bireversible) in terms of these combinatorial structures. 
In particular, we represent the inversion and dualization of transducers as parastrophisms. Further, similarly to the notion of the isotopisms of the quasigroups associated to Latin squares, we develop the notion of isotopisms of letter transducers generalizing transducer symmetry and preserving the class of bireversible automata.
\end{abstract}

\tableofcontents


\section{Introduction}
\label{sec:introduction}

We develop connections between certain families of automata and Latin square-like combinatorial structures.
As we specialize from letter transducers to  Mealy automata to bireversible Mealy automata, the analogy between the corresponding combinatorial structures  and orthogonal Latin squares/orthogonal arrays/$(k,n)$-nets grows stronger.  
Our primary goal is to build a bridge between automata theory and combinatorial design theory in the hope that both areas will benefit from the connections.  With this in mind, we discuss background material for the benefit of those who may not be familiar with one side or the other of this connection.  We also take care in adapting several concepts from the theory of Latin squares to new combinatorial structures relevant to automata.  Of particular note is that every pair of orthogonal Latin squares encodes a bireversible automaton.  

\subsection{Discussion of automata in connection with group theory}

Letter transducers and Mealy automata have been central objects of study in computer science and formal language theory since the middle of the last century. In this paper, we take our motivation from the algebraic side of these objects, namely, from groups and semigroups generated by automata (specifically, by Mealy automata). These groups were formally introduced in the beginning of the 1960s~\cite{horejs:automata}, but it took a while to recognize their importance, utility, and, at the same time, complexity. The first understanding of this class presented in the 1980s when these groups provided counterexamples to several long-standing conjectures in group theory~\cite{grigorch:burnside,grigorch:milnor,grigorch_lsz:atiyah,grigorch_z:basilica}. Later, it became clear that this class of groups had connections to other areas of mathematics, such as holomorphic dynamics~\cite{nekrash:self-similar,bartholdi_n:rabbit}, combinatorics~\cite{grigorchuk-s:hanoi-cr}, analysis on graphs~\cite{grigorch_s:hanoi_spectrum}, computer science~\cite{cain:automaton_semigroups,miasnikov_s:cayley_automatic11,miasnikov_s:automatic_graph}, cryptography~\cite{myasnikov_u:random_subgroups08,myasnikov_su:non_commutative_crypto_book11,garzon_z:crypto,petrides:cryptoanalysis_grigorchuk} and coding theory~\cite{cull_n:hanoi_codes99,grigorchuk-s:hanoi-cr}, and $p$-adic analysis and dynamics~\cite{anashin:automata12,ahmed_s:polynomial_ergodicity,grigorch_s:mealy_moore}. 

Groups generated by automata are especially interesting from an algorithmic point of view: Many basic questions do not have general answers for the whole class but admit partial solutions in some cases. For example, the conjugacy problem is known to be generally undecidable for automaton groups~\cite{sunic_v:conjugacy}, but there are solutions of this problem in multiple specific cases (e.g.,~\cite{modi_su:linear_time_algorithm_conjugacy_grigorchuk21} offers a linear time solution to the conjugacy problem in the Grigorchuk group, and~\cite{bondarenko_bsz:conjugacy} offers a solution to this problem in the group of all bounded automata in the sense of Sidki). Several partial algorithms (finding the order of an element, fast solution of a word problem in contracting groups, etc.) are based on the so-called contraction property which uses the length shortening of group elements representing the states of automata as the length of the input word increases. This property is also responsible for most of the extraordinary examples of groups in this class, such as Burnside groups and groups of intermediate growth. Not all automaton groups possess this property and, perhaps, the most visible classes of this type are those that are generated by reversible and bireversible automata introduced in~\cite{macedonska-n-s:comm} to describe the abstract commensurator group of the free nonabelian groups. Many partial algorithms implemented in GAP packages~\cite{muntyan_s:automgrp,bartholdi:fr} often fail for groups generated by such automata. 

Bireversible automata and related classes represent the main topic of this study. We informally introduce these objects now and give a formal treatment of this topic in Section~\ref{sec:automata}. An automaton $\A$ is called \emph{invertible} if for each state the induced map on letters is invertible. In this case, the \emph{inverse} automaton $\iz\A$ is defined by ``swapping the input and output'' of each transition.
For every $n$-state $m$-letter letter transducer $\A$ one can define an $m$-state $n$-letter \emph{dual} automaton $\dz\A$ by ``swapping the roles'' of letters of the alphabet and the states of an automaton. That is, there is a transition $q\stackrel{x/y}{\longrightarrow}w$ in $\A$ for states $q,w$ of $\A$ and letters $x,y$ from the alphabet if and only if there is a transition $x\stackrel{q/w}{\longrightarrow}y$ in $\dz\A$. An automaton $\A$ is called \emph{reversible} if its dual $\dz\A$ is invertible, and it is called \emph{bireversible} if $\A$, $\dz A$ and $\dz(\iz\A)$ are all invertible. 

The first examples of bireversible automata go back to the work of Aleshin~\cite{aleshin:free}, in which he constructed two bireversible automata with 3 and 5 states, respectively, and claimed that two states of these automata generate the free group of rank 2. Although the original proof contained gaps, it is remarkable that Aleshin was able to pick the only possible 3-state 2-letter automaton among 194 nonsymmetric ones (see~\cite{bondarenko_gkmnss:full_clas32_short}) generating a nonabelian free group.  The statement about the freeness was finally established in~\cite{vorobets:aleshin} 25 years later. This 3-state automaton, now called the Aleshin automaton, initiated several generalizations~\cite{vorobets:series_free,steinberg_vv:series_free,savchuk_v:free_prods} constructing series of bireversible automata generating free groups of higher ranks and free products of groups of order two.

While the proof of freeness in the above-mentioned papers was based on certain transitivity properties of the groups generated by dual automata, another source of bireversible automata generating free groups came from topology in a pioneering work of Glasner and Mozes~\cite{gl_mo:compl}, who were the first to analyze automaton groups via the fundamental groups of the associated $\mathcal{VH}$ square complexes, which are built by gluing square cells that correspond to the transitions in the automaton. This topological approach allowed the authors of~\cite{gl_mo:compl} to construct explicit examples of free groups and groups with Kazhdan property (T) generated by bireversible automata. It was also shown in that paper that an automaton $\mathcal A$ is bireversible if and only if the associated 1-vertex $\mathcal{VH}$ square complex $\mathcal S_\A$ is complete (the link at the vertex of $\mathcal S_\A$ is a complete bipartite graph), and if and only if the universal cover of $\mathcal S_\A$ is a product of two trees (which is equivalent to the universal cover being CAT(0)). We elaborate more on results connecting automaton groups and $\mathcal{VH}$ square complexes in Subsection~\ref{subsec:square-complex}.

Bireversible automata also often generate lamplighter-type groups of the form $A\wr\Z$ for various abelian groups $A$~\cite{bondarenko_dr:lamplighter,ahmed_s:lamplighter,skipper_s:lamplighter20,francoeur:bireversible_lamplighters23,nowak_op:reversible_lamplighter_groups25}.

\subsection{Combinatorial encodings}

In this paper, we examine combinatorial structures that encode 
    letter transducers, 
    Mealy automata, and 
    bireversible automata. 
Let $\A=(Q,X, E)$ be a letter transducer (see Definition~\ref{def:letter-transducer}) with fixed orderings  
   $(q_1, q_2, \ldots, q_m)$ of the states $Q$ and 
   $(x_1, x_2, \ldots, x_n)$ of the letters $X$.
For any positive integer $d$, abbreviate $[d]=\{1, 2, \ldots, d\}$.
Encode the transitions $E$ in the following two structures.
\begin{itemize}
\item \textit{Interleaved Array:} Form the  collection $\mathcal{Y}_\A$ of 4-tuples 
      $\{ (i,j,k,\ell)\in[m]\times[n]\times[m]\times[n] : (q_i, x_j, q_k, x_\ell)\in E\}$.
      We refer to $\mathcal{Y}_\A$ as the \emph{interleaved array} associated with $\A$.
\item \textit{Proper Rete:} Form four ordered partitions 
      $\mathcal{P}_\A=(\mathcal{P}^1_\A, \mathcal{P}^2_\A, \mathcal{P}^3_\A, \mathcal{P}^4_\A)$ 
      of $\mathcal{Y}_\A$ with one (possibly empty) cell containing the elements $y\in \mathcal{Y}_\A$ that take each particular value in the corresponding component of $y$: Cell
      $P^s_h =\{y \in \mathcal{Y}_\A : y_s = h\}$.
      We refer to $\mathcal{P}_\A$ as the \emph{proper rete} associated with $\A$.
\end{itemize}

We characterize the 
    Mealy,
    invertible, 
    reversible, and
    coreversible
properties (Definitions~\ref{def:mealy} and~\ref{def:classes}) for letter transducers $\A$ in terms of the combinatorial properties of 
   $\mathcal{Y}_\A$ and 
   $\mathcal{P}_\A$.
To describe these results, we give some provisional definitions, 
with formal definitions appearing in the body of the paper.
We say that  $\mathcal{Y}_\A$ has property $\mathfrak{T}$ whenever $|\mathcal{Y}_\A|=mn$.
For distinct $s$, $t\in\{1,2,3,4\}$, say that $\mathcal{Y}_\A$ has property $\un{st}$ whenever 
    for all $y$, $z\in \mathcal{Y}_\A$, 
       if $y_s=z_s$ and $y_t=z_t$, then $y=z$.
We say that $\mathcal{Y}_\A$ has property $\ut{st}$ whenever it has properties $\un{st}$ and $\mathfrak{T}$.       
Say partitions $\mathcal{P}$  and $\mathcal{Q}$ of some common set are \emph{orthogonal} when
    $|{P}_i\cap{Q}_j|=1$ for all cells $P_i$ and $Q_j$ of the respective partitions. In this case, write $\mathcal{P}\perp\mathcal{Q}$.
We show the following (Theorem~\ref{thm:classes-U(st)}).

\begin{theorem}
\label{thm:intro-transducerproperties}
Let $\A$ be a letter transducer.  Then the following hold.
\begin{enumerate}
\item $\A$ is Mealy                           if and only if 
      $\mathcal{Y}_\A$ has property $\ut{12}$ if and only if 
      $\mathcal{P}^1_\A\perp \mathcal{P}^2_\A$.
\item $\A$ is invertible                      if and only if 
      $\mathcal{Y}_\A$ has property $\ut{41}$ if and only if 
      $\mathcal{P}^1_\A\perp \mathcal{P}^4_\A$.
\item $\A$ is reversible                       if and only if 
      $\mathcal{Y}_\A$ has property $\ut{23}$  if and only if
      $\mathcal{P}^2_\A\perp \mathcal{P}^3_\A$.
\item $\A$ is coreversible                     if and only if 
      $\mathcal{Y}_\A$ has property $\ut{34}$  if and only if
      $\mathcal{P}^3_\A\perp \mathcal{P}^4_\A$.
\end{enumerate} 
\end{theorem}

We turn our attention to Mealy automata.  
The special cases of interleaved arrays and proper rete given by Theorem~\ref{thm:intro-transducerproperties}
are respectively referred to as \emph{grid arrays} and \emph{grid rete} because the first two components of the four-tuples are the points of a grid.  This leads to a third combinatorial encoding for Mealy automata.
Let $\A=(Q, X, \pi,\lambda)$ be a Mealy automaton with fixed orderings of 
   the states $Q$ and 
   the letters $X$ ordered as above.
Encode  
   the transition function $\pi\colon Q\times X\rightarrow Q$ and 
   the output function $\lambda\colon Q\times X\rightarrow X$ as follows.
\begin{itemize}
\item \textit{Matched Matrices:} Define a  pair of $m\times n$ matrices $C_\A$ and $R_\A$,
    where the $(i,j)$-entries are respectively $k$ and $\ell$ with 
     $\pi(q_i, x_j)=q_k$ and $\lambda(q_i, x_j)=x_\ell$. 
    We refer to $(C_\A, R_\A)$ as the \emph{matched matrices} associated with $\A$.
\end{itemize}

To describe the various families of Mealy automata in terms of matched matrices, we need some additional provisional definitions.
Say that $C_\A$ is \emph{column-Latin} 
     whenever every element of $[m]$ appears exactly once in each column of $C_\A$,
that $R_\A$ is \emph{row-Latin} 
     whenever every element of $[n]$ appears exactly once in each row of $R_\A$, and
that $(C_\A,R_\A)$ is \emph{orthogonal} 
    whenever every possible ordered pair of entries appears exactly once 
    when $C_\A$ and $R_\A$ are superimposed. 
Say that the matched pair of matrices $(C_\A, R_\A)$ is a \emph{cooperative pair} when 
     $C_\A$ is column-Latin, 
     $R_\A$ is row-Latin, and 
     $(C_\A, R_\A)$ is orthogonal.
We show the following 
    (Theorems~\ref{thm:array-matrix-automaton-cond} and~\ref{thm:automatafamily-birevconditions}).

\begin{theorem}\label{thm:intro-mealyconditions}
Let $\A$ be a Mealy automaton.  Then the following hold.
\begin{enumerate}
\item $\A$ is invertible                       if and only if 
      $R_\A$ is row-Latin.
\item $\A$ is reversible                       if and only if 
      $C_\A$ is column-Latin.
\item $\A$ is coreversible                     if and only if 
      $C_\A$ and $R_\A$ are orthogonal.
\item $\A$ is bireversible                     if and only if 
      $(C_\A, R_\A)$ is a cooperative pair.
\end{enumerate} 
\end{theorem}
  
If $\A$ is bireversible, we refer to the associated interleaved array and proper rete as 
a \emph{semi-orthogonal array} and a \emph{reticulation}, respectively.

\subsection{Latin square analogs}

The combinatorial encodings of Mealy automata as 
    grid arrays, 
    grid rete, and 
    matched matrices 
discussed above have analogs in the theory of Latin squares. 

\begin{definition}
\label{def:latin-square}
A {\em Latin square} of order $n$ is an $n\times n$ matrix with entries from $\{1, \ldots, n\}$ in which every possible entry appears exactly once in each row and each column. 
\end{definition}

Our key observation is the following.

\begin{proposition}
\label{prop:OLS-coop}
Any pair of orthogonal Latin squares forms a cooperative pair and therefore encodes a bireversible automaton.  
\end{proposition}

\begin{example}
\label{ex:LatinSquare}
Euler~\cite{Euler:magicsquares} found that the smallest pair of orthogonal Latin squares (up to isotopy) is the following:
\[ L_1=\left[ \begin{array}{ccc} 
     1 & 2 & 3\\
     2 & 3 & 1\\
     3 & 1 & 2
\end{array}\right],\qquad 
L_2=\left[ \begin{array}{ccc} 
     1 & 3 & 2\\
     2 & 1 & 3\\
     3 & 2 & 1
\end{array}\right].
\]
We revisit this pair of orthogonal Latin squares in Example~\ref{ex:latin_squares_correspondence}. 
\end{example}

Orthogonal Latin squares provide a prototype for the present discussion of cooperative pairs 
(and ultimately bireversible automata).  
In this case, the associated 
     semi-orthogonal arrays, 
     reticulations, and
     cooperative pairs     
are strong analogs of 
    orthogonal arrays, 
    $(k,n)$-nets (Bruck nets), and
    Latin squares,
respectively. 
We shall adapt the notions of parastrophy and isotopy to cooperative pairs and 
discuss their meaning for letter transducers and Mealy automata.  
Before further discussing this topic, 
we briefly review a few classical results concerning Latin squares. 

Latin squares lie at a nexus of combinatorics,  geometry, algebra, and statistics.
General references for Latin squares include~\cite{Denes+Keedwell-LSa_SecondEdition_2015, Evans:orthLSbasedongroups, Laywin+Mullen:discretemath,vanLintWilson:coursecomb}.
Euler~\cite{Euler:magicsquares} appears to have given the first analytic description of Latin squares and orthogonal Latin squares in 1782~\cite{Richardson:introducedLatinSq}.  
A discussion of earlier instances and implicit appearances of Latin squares can be found in~\cite{Andersen:Ch11LatinSquares}. 
By the early 1900s, Latin squares were a firmly established topic in combinatorics as regular structures and sets of discordant permutations, such as in~\cite{Cayley:oLS,MacMahon:CA}.  
Many early results discussed enumeration/construction and nonexistence. 
Euler conjectured that there is no pair of orthogonal Latin squares for orders $n \equiv 2 \pmod 4$~\cite{Euler:magicsquares}.
Tarry~\cite{Tarry1900} proved this to be true for $n=6$.  
It was not until 1959 that Bose, Shrikhande, and Parker~\cite{boseshrikhandeparker1960} proved that 
orthogonal Latin squares exist for every order \(n\neq 2,6\).
MacNeish gave a direct product construction~\cite{MacNeish1922} for families of mutually orthogonal Latin squares. 

There are strong connections between Latin squares and finite geometry.
A $(k,n)$-net consists of 
    $n^2$-many points and 
    $nk$-many lines 
such that  
    the lines comprise $k$-many parallel classes, 
    each containing $n$-many lines, 
    every line contains $n$-many points, and 
    every point is incident with $k$-many lines, one from each parallel class. 
 A set of $k-2$ mutually orthogonal Latin squares of order $n$ is equivalent to a $(k,n)$-net~\cite{Bruck:FNI, Bruck:FNII}.   
This equivalence motivates our notion of a rete (the Latin term for net).
We note that $(k,n)$-nets can be viewed as generalizations of finite affine planes, and
a complete set of $n-1$ many mutually orthogonal Latin squares of order $n$ is equivalent to a finite projective plane~\cite{moore:tactical} 
(see~\cite{Denes+Keedwell-LSa_SecondEdition_2015, Evans:orthLSbasedongroups, Laywin+Mullen:discretemath}, among many other texts).
The reference~\cite{Andersen:Ch11LatinSquares} discusses some background concerning connections to finite geometry.
We note that $(k,n)$-nets are the discrete analogs/extensions of the three-webs introduced and studied 
by Blaschke and the members of his school~\cite{BlaschkeBol:GdGTFdD} in the 1920s and 1930s 
as a means of studying local invariants of three foliations of curves in the plane in differential geometry.  
The reference~\cite{AKIVIS20001} discusses the differential geometry of webs, including a brief discussion of their history. 


Latin squares can be studied from an algebraic perspective.
A Latin square can be viewed as the Cayley table of a quasigroup. 
A quasigroup consists of a set $Q$ together with a binary operation $\circ$ such that 
    for all $a$, $b\in Q$ there exist unique $x$, $y\in Q$ such that $a\circ x=b$ and $y\circ a=b$. 
These conditions assert that each element of $Q$ appears exactly once 
in each row and in each column of the Cayley table.  
Quasigroups were introduced by Moufang~\cite{Moufang:ZScA} in the 1930s in connection with three-webs. 
See~\cite{Pflugfelder:HNoLT} for a discussion of the history of quasigroups.  
The texts~\cite{Scherbacov:Eqta,SmithRomanowsak:PMA, Smith:IQGR} discuss quasigroups as algebraic objects.  

In addition to their Cayley tables, groups are used to construct mutually orthogonal Latin squares in other ways. 
We note two such uses.
Mann~\cite{Mann1942} introduced what is now called Mann's automorphism method to 
construct orthogonal Latin squares based on orthomorphisms of groups.
Jungnickel~\cite{Jungnickel1980} used difference matrices with entries from abelian groups to construct families of orthogonal Latin squares (see also~\cite{Boykett2009}~\cite[Section 8.2.2]{Evans:orthLSbasedongroups}).
A construction of a complete set of mutually orthogonal Latin squares from finite fields was given in connection with geometric considerations~\cite{moore:tactical}, see also~\cite{Denes+Keedwell-LSa_SecondEdition_2015}.
Another classical construction was given by Bose in 1938.

\begin{theorem}[\cite{Bose1938}]
Let \(q=p^m\) be a prime power and let \(\mathbb F_q\) be the finite field of order \(q\). 
For each \(a\in \mathbb F_q^\times\), define $L_a(x,y)=ax+y$. 
Then each \(L_a\) is a Latin square of order \(q\), and \(L_a\) and \(L_b\) are orthogonal whenever \(a\neq b\).
\end{theorem}

Latin squares also play a role in the design of experiments (statistics).  
In the 1920s Fischer encouraged the use of Latin squares to design agricultural experiments~\cite{Fisher:tDoE}
(precursors are noted in~\cite{Andersen:Ch11LatinSquares}). 
Latin squares allow the experimenter to control for two nuisance variables by isolating treatment effects and removing systematic bias. 
A set of mutually orthogonal Latin squares allows one to control for more nuisance variables.  
Mutually orthogonal Latin squares are subsumed by the theory of orthogonal arrays.
An orthogonal array of type $(M,n,s,t)$  is 
    an $M\times n$ array with 
    entries from a set $S$ of size $s$ such that 
    every $M\times t$ subarray contains each $t$-tuple of elements of $S$ exactly $(M/s^t)$-many times as a row.  
Orthogonal arrays were introduced by Rao in the late 1940s~\cite{Rao:hypercubesofstrength, Rao:factorialexperiments}. 
See~\cite{HedayatSloneSufken:OA,LinStufken:ortharray} for a general discussion of orthogonal arrays.
Orthogonal arrays are often viewed as a collection of row vectors.  
In this sense, interleaved arrays are analogous to orthogonal arrays. 

Quasigroups have been studied as reversible automata with three mutually interacting state spaces~\cite{GVarmiyaPlotkin:hqua, Smith:QHSRA} (see also~\cite{SmithRomanowsak:PMA}).  
Cellular automata have been used to construct mutually orthogonal Latin squares~\cite{Mariotetal:molsca, Mariotetal:EOLSBCA}, where~\cite{Manzonietal:CDCAsurvey} surveys this topic. 
These prior connections between Latin squares and various automata do not appear to be directly related to the results of the present paper. 
Here we introduce combinatorial structures to study finite automata rather than adapt automata to study Latin squares or equivalent structures.

The theory of Latin squares is a rich and active area of study.  
The above discussion has not touched on the many recent developments.  
We note that Latin squares have not only  been studied in their own right, but their regularity has also been used to construct other regular combinatorial structures, such as 
    association schemes, 
    block designs, 
    error correcting codes,  
    tournaments, and
    cryptographic schemes
\cite{bailey:ASDEAC,Denes+Keedwell-LSa_SecondEdition_2015, Denes-Keedwell-LSnewdev,Laywin+Mullen:discretemath,vanLintWilson:coursecomb}.
Here we develop combinatorial encodings of finite automata in the hope that techniques developed in connection with Latin squares and associated objects can be adapted to the benefit of finite automata and vise versa.  
For example, it may be interesting to examine the groups generated by the bireversible automata encoded by various families of orthogonal Latin squares. 
Relating the quasigroups associated with such Latin squares to the group generated by the automaton might be interesting.
Another interesting direction to explore is to adapt techniques of random generation of Latin squares to the case of bireversible automata. 
It would also be interesting to investigate how isotopisms of Mealy automata impact the structure of the associated automaton (semi)group. 
We hope that the proposed connections may produce new enumeration methods for bireversible automata.

\subsection{Parastrophy and isotopy}

In the present work, we adapt the notions of  parastrophy and isotopy from the theory of Latin squares to automata and connect them to known symmetries of letter transducers and automata.  
The interleaved array provides a unifying framework for these topics. 
Parastrophy and isotopy are transformations of interleaved arrays (and the equivalent structures).

Parastrophy plays a significant role in Latin squares~\cite[Section 1.4]{Denes+Keedwell-LSa_SecondEdition_2015}~\cite[Chapter 17]{vanLintWilson:coursecomb}.  
In the theory of quasigroups, parastrophy corresponds to the quasigroups with operations derived from the original by left division, right division, and opposite multiplication.  
Parastrophisms are often used to define equivalence classes to simplify the enumeration/listing of Latin squares~\cite[Chapter 4]{Denes+Keedwell-LSa_SecondEdition_2015}. 
In Section~\ref{sec:parastrophism}, we introduce parastrophy for interleaved arrays and relate parastrophisms of interleaved arrays to duals and inverses of the corresponding finite automata.
A \emph{parastrophism} is any of the eight permutations (in disjoint cycle notation) 
    $()$, $(2,4)$, $(1,3)$, $(1,3)(2,4)$, $(1,2)(3,4)$, $(1,2,3,4)$, $(1,4,3,2)$, $(1,4)(2,3)$ 
acting on elements of an interleaved array by permuting the components.  

Isotopy is an important concept in Latin squares~\cite[Section 1.1]{Evans:orthLSbasedongroups}~\cite[Section 1.3]{Denes+Keedwell-LSa_SecondEdition_2015}~\cite[Chapter 17]{vanLintWilson:coursecomb}.  
Isotopisms are used to define equivalence classes to simplify the enumeration and analysis of Latin squares. 
In Section~\ref{sec:isotopism}, we introduce isotopism for interleaved arrays.
An \emph{isotopism} is 
    a four-tuple of permutations $\sigma=(\sigma_1,\sigma_2,\sigma_3,\sigma_4)$ 
    with $\sigma_1,\sigma_3\in\Sym(m)$ and $\sigma_2,\sigma_4\in\Sym(n)$ 
    acting on elements of $\mathcal{Y}_\A$ by permuting the entries in component $i$ by $\sigma_i$.
From the point of view of automata theory, the value of interleaved array representation of an automaton lies in the unifying nature of multiple automata operations as natural symmetries of the corresponding interleaved array. 
Namely, we prove the following theorem (Theorem~\ref{thm:birev-all-parastrophes} and~\ref{thm:isotopism_and_automata}).

\begin{theorem}
\label{thm:intro:paraiso}
Let $\A=(Q=[m],X=[n],E)$ be a letter transducer, and let $\iota_n$ and $\iota_m$ denote the trivial permutations of $X$ and $Q$, respectively. Let $\mathcal Y_\A$ be the interleaved array corresponding to $\A$.  Then
\begin{itemize}
    \item[(i)] Inversion of $\A$ corresponds to the parastrophism of $\mathcal{Y}_\A$ induced by permutation $(2,4)$.
    \item[(ii)] Dualization of $\A$ corresponds to the parastrophism of $\mathcal{Y}_\A$ induced by permutation $(1,2)(3,4)$.
    \item[(iii)] Permuting the set of states $Q$ of $\A$ by permutation $\rho\in\Sym(Q)$  corresponds to the isotopism $\sigma=(\rho, \iota_n, \rho, \iota_n)$ of $\mathcal{Y}_\A$.
    \item[(iv)] Permuting the letters of $X$ by permutation $\kappa\in\Sym(X)$ corresponds to the isotopism $\sigma=(\iota_m, \kappa,  \iota_m, \kappa)$ of $\mathcal{Y}_\A$.
\end{itemize}
\end{theorem}

We note that while the isotopisms from operations (iii) and (iv) in Theorem~\ref{thm:intro:paraiso} preserve the structure of a (semi)group generated by automaton $\A$, 
generally an isotopism can change the isomorphism class of the corresponding (semi)group as shown in Example~\ref{ex:bellaterra_aleshin}. 
This is analogous to the fact that the isotopisms of a Latin square may change the isomorphism type of a quasigroup with the multiplication table defined by this Latin square~\cite[Section 1.3]{Denes+Keedwell-LSa_SecondEdition_2015}.

\noindent\textbf{Acknowledgments.} The authors thank Ievgen Bondarenko for providing valuable input on the connection between various types of automata and $\mathcal{VH}$-square complexes and suggesting relevant references. The second author was partially supported by NSF grant DMS-2342254.

\section{Classes of automata}
\label{sec:automata}

We start this section by introducing the notions and terminology of automata-transducers, Mealy automata, and their actions on regular rooted trees. For automata-transducers we will mainly use the terminology equivalent to the standard one described in~\cite{sakarovitch:elements_of_automata_theory_book09}, and for Mealy automata and their actions on rooted trees we refer the reader to~\cite{gns00:automata}. We will not introduce the transducers in the most general terms; instead, we will restrict our attention only to the case relevant to this paper. In particular, we will assume that the input and output alphabets always coincide, and we will assume that each state is terminal, so that we will omit discussion about the set of terminal states. With these conventions, we start with the most general definition of finite-state transducer needed for our purposes.

\subsection{Finite state transducers, letter transducers, and Mealy automata}
 
Let $X$ be a finite set with cardinality $d\geq 2$ whose elements will be called \emph{letters}. 
While in many references it is customary to take $X=\{0,1,\ldots ,d-1\}$ and endow it with the structure of a ring $\Z/d\Z$, throughout this paper we will mostly consider $X=[d]$. Let $X^*$ denote the set of finite words over $X$. This set can be equipped with the structure of a rooted $d$-ary tree by declaring that $v$ is adjacent to $vx$ for every $v\in X^*$ and $x\in X$. The empty word $\varepsilon\in X^*$ corresponds to the root of the tree, and for each positive integer $n$ the set $X^n$ corresponds to the $n$-th level of the tree. All types of initial transducers defined below transform words from $X^*$ to other words from $X^*$. 

\begin{definition}
    A \emph{finite-state transducer (FST)} is a triple $\A=(Q,X,E)$ such that:
    \begin{itemize}
        \item $Q$ is a finite set of states,
        \item $X$ is a finite alphabet,
        \item $E\subset Q\times X^*\times Q\times X^*$ is the set of transitions (edges).
    \end{itemize}
    If a state $q\in Q$ is selected, then the 4-tuple $\A_q=(Q,X,E,q)$ is called the \emph{initial FST}. 
\end{definition}

A finite transducer can be represented by a transition diagram whose nodes are the states and for every transition $(q_1,w_1,q_2,w_2)\in E$ there is an oriented edge 
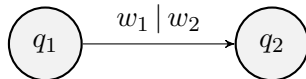
\begin{center}
    \begin{tikzpicture}
    \node[state] (q1) {$q_1$};
    \node[state, right of=q1] (q2) {$q_2$};
    \draw (q1) edge[above] node{$w_1\,|\,w_2$} (q2);
    \end{tikzpicture}
\end{center}
Then each initial FST $\A_q$ defines a \emph{rational relation} $\rho_q\subset X^*\times X^*$ (sometimes called a \emph{transduction relation}), in which $(w_1,w_2)\in\rho_q$ if and only if there is a path $p$ in the transition diagram of $\A$ initiating from $q$ such that $w_1$ and $w_2$ are obtained as concatenations of the first components and the second components of the edge labels along the path $p$, respectively. The formal definition of this relation can be found in~\cite[Section~1.3.1]{roche_s:sinite-state_language_processing97}.

The notion of FST allows for several useful specializations. The most general one that we will need in this paper is the letter (or letter-to-letter) transducer.  

\begin{definition}
\label{def:letter-transducer}
A \emph{letter transducer} is an FST $(Q,X,E)$ such that $E\subset Q\times X\times Q\times X$.   
\end{definition}

We note that there are different terms in the literature corresponding to letter transducers. For example, in~\cite{roche_s:sinite-state_language_processing97} these machines are called $\varepsilon$-free letter transducers, reserving the term ``letter transducers'' for machines that can read/output empty string as well.

In general, the rational relation does not define a transformation of $X^*$, but with a natural extra condition on $E$ this property holds. This represents the case of main interest of group and semigroup theory, as it allows one to define (semi)groups generated by transformations of $X^*$ induced by initial transducers. To avoid unnecessary notational clutter, we restrict the definitions below to the relevant case of letter transducers.

\begin{definition}
\label{def:mealy}
Let $\A=(Q,X,E)$ be a letter transducer. Then $\A$ is called 
    \begin{itemize}
        \item \emph{deterministic} if for all $q\in Q$ and $x\in X$ there is at most one element of $E$ of the form $(q,x,\cdot,\cdot)$;
        \item \emph{complete} if for all $q\in Q$ and $x\in X$ there is at least one element of $E$ of the form $(q,x,\cdot,\cdot)$;
        \item \emph{Mealy automaton} if it is both deterministic and complete. 
    \end{itemize}
\end{definition}

We give below another equivalent definition of Mealy automata more common in group and semigroup theory.

\begin{definition}
\label{def:mealy2}
A \emph{Mealy automaton} (or simply \emph{automaton}) is a tuple $(Q,X,\pi,\lambda)$, where $Q$ is a finite set (the set of \emph{states}), $X$ is a finite alphabet, $\pi\colon Q\times X\to Q$ is the \emph{transition function}, 
and $\lambda\colon Q\times X\to X$ is the \emph{output function}. Selecting a state $q\in Q$ produces an \emph{initial automaton} $\A_q$.
\end{definition}

The equivalence of the two definitions of Mealy automata is straightforward and achieved through the relation that $(q,x,r,y)\in E$ if and only if $\pi(q,x)=r$ and $\lambda(q,x)=y$. 
We will adopt the notation from~\cite{bondarenko_gkmnss:full_clas32_short} and refer to Mealy automata with $m$ states over an $n$-letter alphabet as $(m,n)$-automata. Note that in some cases in the literature, the condition of finiteness of the set of states is not included in the definition of a Mealy automaton. For example, to each infinite automaton group one can associate an infinite automaton.
In this paper to make notation more straightforward and compatible with the standard definition of finite state transducers, we will make the assumption that all Mealy automata have finite number of states.

The transition diagrams of Mealy automata are usually referred to as \emph{Moore diagrams}. The Moore diagram of an automaton $\A=(Q,X,\pi,\lambda)$ is a directed graph in which the vertices are the states of $Q$ and the edges have the form

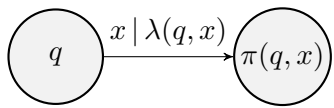
\begin{center}
    \begin{tikzpicture}
    \node[state] (q1) {\phantom{a}\hspace{1.6mm}$q$\hspace{1.6mm}\phantom{a}};
    \node[state, right of=q1] (q2) {\!\!$\pi(q,x)$\!\!};
    \draw (q1) edge[above] node{$x\,|\,\lambda(q,x)$} (q2);
    \end{tikzpicture}
\end{center}

\noindent for $q\in Q$ and $x\in X$. 
For example, the upper left corner of Figure~\ref{fig:latin_correspondence} shows the Moore diagram of the automaton studied in~\cite{bondarenko_dr:lamplighter}: 
This automaton generates the lamplighter group $(\Z/3\Z)\wr\Z$ as explained below.

The rational relation $\rho_q$ for an initial Mealy automaton $\A_q$ defines a transformation of $X^*$ as for each $w_1\in X^*$ there is exactly one $w_2\in X^*$ such that $(w_1,w_2)\in\rho_q$. Moreover, this transformation is a graph endomorphism of $X^*$, viewed as a regular rooted tree. This transformation will also be denoted by $\A_q$ and is defined as follows. Given a word
$v=x_1x_2x_3\ldots x_n\in X^*$, it scans its first letter $x_1$ and outputs $\lambda(q,x_1)$. The rest of the word is handled similarly by the initial automaton $\A_{\pi(q,x_1)}$. This allows us to extend the functions $\pi$ and $\lambda$ to $\pi\colon Q\times X^*\to Q$ and $\lambda\colon  Q\times X^*\to X^*$ via the equations
\[\begin{array}{l}
\pi(q,x_1x_2\ldots x_n)=\pi(\pi(q,x_1),x_2x_3\ldots x_n),\\
\lambda(q,x_1x_2\ldots x_n)=\lambda(q,x_1)\lambda(\pi(q,x_1),x_2x_3\ldots x_n).\\
\end{array}
\]

\subsection{Invertible, reversible, coreversible, and bireversible automata}

The class of Mealy automata as a subclass within the class of letter transducers has three additional natural counterparts.

\begin{definition}
\label{def:classes}
Let $\A=(Q,X,E)$ be a letter transducer. Then $\A$ is called 
    \begin{itemize}
        \item \emph{invertible} if for all $q\in Q$ and $y\in X$ there is exactly one element of $E$ of the form $(q,\cdot,\cdot,y)$;
        \item \emph{reversible} if for all $x\in X$ and $r\in Q$ there is exactly one element of $E$ of the form $(\cdot,x,r,\cdot)$;
        \item \emph{coreversible} if for all $r\in Q$ and $y\in X$ there is exactly one element of $E$ of the form $(\cdot,\cdot,r,y)$.
    \end{itemize}
\end{definition}

\begin{definition}
A letter transducer is called 
    \begin{itemize}
        \item an \emph{IR-automaton} if it is both invertible and reversible; 
        \item \emph{bireversible} if it is Mealy, invertible, reversible, and coreversible. 
    \end{itemize}
\end{definition}

Reversible, invertible, and bireversible Mealy automata have been studied extensively in the literature. However, symmetry considerations suggest that it is more natural to define these notions for the class of  transducers. We note that in the case of reversible, invertible, and bireversible Mealy automata, the notions completely coincide.  The notion of a coreversible Mealy automaton introduced in~\cite{godin_kp:torsion-free_semigroups2015} is slightly different from the one introduced here.

There is a well-known natural interpretation of the classes of automata described above coming from the theory of square complexes and introduced in~\cite{gl_mo:compl}. While in that paper and others in the area the construction is  applied only to Mealy automata, in fact, it is more natural to apply it to letter transducers as we suggest here. Under this interpretation, each automaton $\A=(Q,X,E)$ corresponds to a square complex with the one-skeleton being a wedge (a one point union) of $(2|Q|+2|X|)$-many circles labeled by elements of $Q^{\pm1}$ and $X^{\pm1}$, and whose 2-cells are squares corresponding to transitions in $E$. Namely, the transition $(q,x,r,y)\in E$ corresponds to a topological square with a labeling scheme $xry^{-1}q^{-1}$ along its boundary as shown in Figure~\ref{fig:classes}. 

Under this interpretation, according to Definition~\ref{def:classes}, the classes of Mealy, invertible, reversible, and coreversible automata correspond to square complexes in which for every pair of labels at the SW-, NE-, SW-, and NE-corners, respectively, there is exactly one pair of labels of the other two sides corresponding to a transition in the automaton, as depicted in Figure~\ref{fig:classes}.  Bireversible automata, for which a transition is completely determined by the labels at any of the corners, correspond to complete directed square complexes as defined in~\cite{gl_mo:compl}.

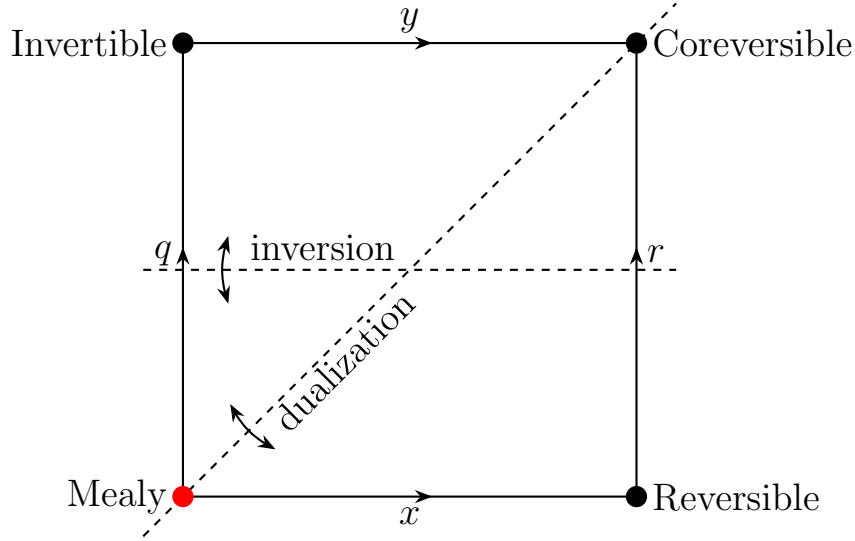
\begin{figure}
    \centering
\begin{tikzpicture}[
    thick,
    every node/.style={font=\Large},
    midarrow/.style={
        -,
        postaction={
            decorate,
            decoration={
                markings,
                mark=at position 0.55 with {\arrow{Stealth}}
            }
        }
    }
]

\coordinate (A) at (0,0);
\coordinate (B) at (6,0);
\coordinate (C) at (0,6);
\coordinate (D) at (6,6);

\draw[midarrow] (A) -- node[below] {$x$} (B);
\draw[midarrow] (A) -- node[left, yshift=6pt] {$q$} (C);
\draw[midarrow] (C) -- node[above] {$y$} (D);
\draw[midarrow] (B) -- node[right, yshift=6pt] {$r$} (D);

\draw[dashed,-] (-0.525,3) -- (6.525,3);

\draw[dashed,-] (-0.525,-0.525) -- (6.525,6.525);

\fill[red] (A) circle (4pt);
\fill (B) circle (4pt);
\fill (C) circle (4pt);
\fill (D) circle (4pt);

\node[left,xshift=-2pt]  at (C) {Invertible};
\node[right,xshift=2pt] at (D) {Coreversible};
\node[left,xshift=-2pt]  at (A) {Mealy};
\node[right,xshift=2pt] at (B) {Reversible};

\draw[<->,>=Stealth,bend left=18]
    (0.6,2.55) to (0.6,3.45);

\node[right] at (0.75,3.3) {inversion};

\draw[<->,>=Stealth,bend right=18]
    (0.625,1.225) to (1.225,0.625);

\node[right,rotate=45] at (1.18,0.855) {dualization};

\end{tikzpicture}
\caption{Classes of letter transducers}
    \label{fig:classes}
\end{figure}

For each letter transducer $\A$, one can define the inverse and the dual to $\A$ letter transducers as follows.

\begin{definition}
    Let $\A=(Q,X,E)$ be a letter transducer. Denote by $Q^{-1}$ the set $\{q^{-1}: q\in Q\}$ of formal inverses of states in $Q$. Then the \emph{inverse to $\A$} is the letter transducer $\iz\A$  defined by $\iz\A=(Q^{-1},X, E^{-1})$, where for $q,r\in Q$ and $x$, $y\in X$, $(q,x,r,y)\in E$ if and only if $(q^{-1},y,r^{-1},x)\in E^{-1}$. 
\end{definition}

The following characterization follows immediately from the above definition.

\begin{proposition}
    A letter transducer $\A$ is invertible if and only if $\iz\A$ is a Mealy automaton.
\end{proposition}

By definition, the rational relation of $\iz\A_{q^{-1}}$ is the reflection of the rational relation of $\A_q$ for each state $q$ of $\A$. Note that one can define the inverse of any finite state transducer in exactly the same way, but we will only need to talk about inverses of letter transducers. For a Mealy automaton $\A$, whose rational relations induce functions $\A_q\colon X^*\to X^*$, the inverse $\iz\A$ is a Mealy automaton if and only if all these functions are bijections and define graph automorphisms of $X^*$ as a regular tree. 

The dual letter transducer is defined by switching the roles of the set of states and the alphabet.

\begin{definition}
    Let $\A=(Q,X,E)$ be a letter transducer. Then the \emph{dual} $\dz\A$ of $\A$ is the letter transducer $\dz\A=(X, Q, E')$, where $E'=\{(x,q,y,r):(q,x,r,y)\in E\}$.
\end{definition}

It is easy to see that by definition, the transitions of $\iz\A$ are obtained by reflecting the labels of all the squares associated with transitions in $\A$ about the horizontal axis of symmetry of the square. Similarly, the transitions of $\dz\A$ are obtained by reflecting the labels of all the squares associated with transitions in $\A$ about the diagonal axis of symmetry of the square going from the SW corner to the NE corner. This is also shown in Figure~\ref{fig:classes}. This defines the action of the dihedral group $D_4$ of size 8 on the set of letter transducers. Note that none of these operations changes the topological square complex associated with $\A$ as discussed in Subsection~\ref{subsec:square-complex}.

Since dualization corresponds to applying a symmetry along the SW-NE diagonal and SE corner is swapped with NW corner in every transition square we immediately obtain:

\begin{proposition}
A letter transducer $\A$ is reversible if and only if the $\dz\A$ is invertible.
\end{proposition} 

Finally, a Mealy invertible reversible automaton is bireversible if and only if it is coreversible, which is equivalent to either $\dz\iz\A$ being invertible or $\iz\A$ being reversible. The action of $D_4$ on the transition squares also yields $\iz\dz\iz\dz\A=\dz\iz\dz\iz\A$.

In case of invertible Mealy automata, one can give an alternative, equivalent definition of the inverse automaton in line with Definition~\ref{def:mealy2}.

\begin{definition}
\label{def:invertible}
Let $\A=(Q,X,\pi,\lambda)$ be a Mealy automaton.
If for every state $q\in Q$ the output function $\lambda_q(x)=\lambda(q,x)$ induces
a permutation of $X$, then the automaton $\A$ is called \emph{invertible}.
In this case, we define its \emph{inverse} as the automaton $\iz\A=(Q^{-1},X,\tilde{\pi},\tilde{\lambda})$, where
\[\begin{array}{l}
\tilde{\pi}(q^{-1},x)=\pi(q,\lambda_q^{-1}(x))^{-1},\\
\tilde{\lambda}(q^{-1},x)=\lambda_q^{-1}(x).
\end{array}\]
\end{definition}

It follows immediately from the definition of the inverse automaton of $\A=(Q,X,\pi,\lambda)$ that for $q,r\in Q$ and $x,y\in X$,  $\A$  has a transition 
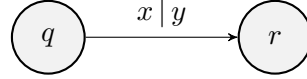
\begin{center}
    \begin{tikzpicture}
    \node[state] (q1) {$q$};
    \node[state, right of=q1] (q2) {$r$};
    \draw (q1) edge[above] node{$x\,|\,y$} (q2);
    \end{tikzpicture}
\end{center}
if and only if $\iz\A$ has the following transition
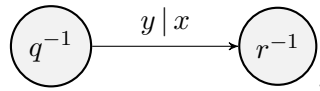
\begin{center}
    \begin{tikzpicture}
    \node[state] (q1) {$q^{-1}$};
    \node[state, right of=q1] (q2) {$r^{-1}$};
    \draw (q1) edge[above] node{$y\,|\,x$} (q2);
    \end{tikzpicture}.
\end{center}

\begin{definition}
The semigroup (group) generated by all states of an (invertible) Mealy automaton $\A$ viewed as graph endomorphisms (automorphisms) of the rooted tree $X^*$ under the operation of composition is called an \emph{automaton semigroup (group)} and is denoted by $\mathbb{S}(\A)$ (respectively, $\mathbb{G}(\A)$).
\end{definition}

Similarly, an alternative equivalent definition of a \emph{dual automaton} of a Mealy automaton in line with Definition~\ref{def:mealy2} is the following.

\begin{definition}
Given a finite automaton $\A=(Q,X,\pi,\lambda)$, its {\em dual} is the finite automaton $\dz\A=(X,Q,\hat{\lambda},\hat{\pi})$, where for every $x\in X$ and $q\in Q$,
\[\begin{array}{l}
\hat\lambda(x,q)=\lambda(q,x),\\
\hat\pi(x,q)=\pi(q,x).
\end{array}\]
\end{definition}

It follows immediately from the definition of the dual automaton of $\A=(Q,X,\pi,\lambda)$ that for $q,r\in Q$ and $x,y\in X$,
$\A$ has a transition 
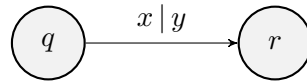
\begin{center}
    \begin{tikzpicture}
    \node[state] (q1) {$q$};
    \node[state, right of=q1] (q2) {$r$};
    \draw (q1) edge[above] node{$x\,|\,y$} (q2);
    \end{tikzpicture}
\end{center}
if and only if $\dz\A$ has the following transition
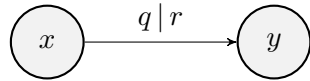
\begin{center}
    \begin{tikzpicture}
    \node[state] (q1) {$x$};
    \node[state, right of=q1] (q2) {$y$};
    \draw (q1) edge[above] node{$q\,|\,r$} (q2);
    \end{tikzpicture}.
\end{center}

It is easy to see that the dual of the dual of an automaton $\A$ coincides with $\A$.



The semigroup $\mathbb{S}(\dz\A)$ generated by $\dz\A$ acts on the free monoid $Q^*$. This induces the action on $\mathbb{S}(\A)$. Similarly, $\mathbb{S}(\A)$ acts on $\mathbb{S}(\dz\A)$.

We point out that classes of finite-state transducers more general than letter transducers that give rise to various classes of groups and semigroups have been studied in the literature. For example, if the set of transitions $E$ of $\A=(Q,X,E)$ is a subset of $Q\times X\times Q\times X^*$, then states of $\A$ define the so called rational transformations of the Cantor set $X^{\mathbb N}$ and $\A$ is referred to as asynchronous automaton. This class was initially studied in~\cite{gns00:automata} and recently drew a lot of attention, as it played a significant role in the solution of the Boone-Higman conjecture for hyperbolic groups~\cite{belk_bmz:hyperbolic_boone-higman}. Thompson's group $F$ belongs to this class too. On the other hand, one can extend the definition of a finite-state transducer to include the transitions from $Q\times X\times Q^*\times X$, in which case one obtains the class of finitely generated self-similar or functionally recursive semigroups (or groups) acting on rooted trees~\cite{nekrash:self-similar}. The (semi)groups from this class may have undecidable word and order problems~\cite{gillibert:order_undecidable18,bartholdi_m:order_undecidable}.

\subsection{Wreath recursion for rooted tree endomorphisms}

For endomorphisms of rooted trees we now introduce the notion of a \emph{state} (also called \emph{section}) at a vertex of the tree $X^*$. Let $s$ be an endomorphism of the tree and $x\in X$. For any $v\in
X^*$ we can write \[s(xv)=s(x)v'\] for some $v'\in X^*$. Then the map $s|_x\colon X^*\to X^*$ given by \[s|_x(v)=v'\] defines an endomorphism of $X^*$ which we call the \emph{state} of $s$ at $x$. We can inductively extend the definition of a section at a letter $x\in X$ to a section at any finite word $x_1x_2\ldots x_n\in X^*$ as follows: \[s|_{x_1x_2\ldots x_n}=s|_{x_1}|_{x_2}\ldots|_{x_n}.\]


We will adopt the following convention throughout this paper. If $s$ and $t$ are elements of some (semi)group acting on a set $Y$ and $y\in Y$, then $$st(y)=t(s(y)).$$
Hence, the state $s|_v$ at $v\in X^*$ of any product $s=s_1s_2\cdots s_n$ for some endomorphisms $s_i$ of $X^*$, $1\leq i\leq n$, can be computed as follows:
$$s|_v=s_1|_v\cdot s_2|_{s_1(v)}\cdots
s_n|_{s_1s_2\cdots s_{n-1}(v)}.$$

Let $\Tr(X)$ denote the monoid of all transformations of a set $X$ with the composition operation. For each automaton semigroup $S$ there is a natural embedding
\[S\hookrightarrow S \wr \Tr(X)\] given by
\begin{equation}
\label{eq:wreath}
G\ni s\mapsto (s_1,s_2,\ldots,s_d)\tau_s\in S\wr \Tr(X),
\end{equation}
where $s_1,s_2,\ldots,s_d$ are the states of $s$ at the vertices of the first level, and $\tau_s$ is the transformation of $X$ induced by the action of $s$ on the first level of the tree. If $\tau_s$ is the trivial permutation, it is customary to omit it in the notation. Similarly, in the case where $G$ is an automaton group, there is a natural embedding
\[G\hookrightarrow G \wr \Sym(X)\] given by
\begin{equation}
\label{eq:wreath_gp}
G\ni g\mapsto (g_1,g_2,\ldots,g_d)\sigma_g\in G\wr \Sym(X),
\end{equation}
where $g_1,g_2,\ldots,g_d$ are the states of $g$ at the vertices of the first level, and $\sigma_g$ is the permutation of $X$ induced by the action of $g$ on the first level of the tree. If $\sigma_g$ is the trivial permutation, it is customary to omit it in the notation. In the case of a binary rooted tree $\{0,1\}^*$, there is only one nontrivial permutation, namely the transposition $(01)$, which is usually denoted simply by $\sigma$. With a slight abuse of notation, we will identify an endomorphism with its image under the embedding~\eqref{eq:wreath} and write $s=(s_1,s_2,\ldots,s_d)\tau_s$ and $g=(g_1,g_2,\ldots,g_d)\sigma_g$ in the group case.

The decomposition of all generators of an automaton (semi)group under the embeddings~\eqref{eq:wreath} and~\eqref{eq:wreath_gp} is called the \emph{wreath recursion} defining the (semi)group. It is a convenient language when doing computations involving the states of endomorphisms. Indeed, the product of endomorphisms can be found as follows. If $s=(s_1,s_2,\ldots,s_d)\tau_s$ and $t=(t_1,t_2,\ldots,t_d)\tau_h$ are two endomorphisms of $X^*$, then \begin{equation}
    \label{eqn:prod}
st=\left(s_1t_{\tau_s(1)},s_2t_{\tau_s(2)},\ldots,s_{d}t_{\tau_s(d)}\right)\tau_s\tau_t
\end{equation}
and if $g=(g_1,g_2,\ldots,g_d)\sigma_g$ is an automorphism of $X^*$, then
\begin{equation}
    \label{eqn:inv}
    g^{-1}=\left(g^{-1}_{\sigma_g^{-1}(1)},g^{-1}_{\sigma_g^{-1}(2)},\ldots,g^{-1}_{\sigma_g^{-1}(d)}\right)\sigma_g^{-1}.
\end{equation}
We will refer to the above formula while discussing matrix representations of Mealy automata in Section~\ref{sec:matchedmatrix-rep}.

\section{Combinatorial encodings}
     \label{sec:comb-encode-automata}

For the rest of the paper, fix positive integers $m$ and $n$.     
For any positive integer $d$, write $[d]=\{1,2,\ldots, d\}$.

In Subsection~\ref{subsec:letter-transducers}, we discuss bijective correspondences among  
   letter transducers with $m$ states and $n$ letters, 
   subsets of $[m]\times[n]\times[m]\times[n]$,
   and certain four-tuples of partitions.
In Subsection~\ref{subsec:square-complex}, we discuss connections with square complexes.   
In Subsection~\ref{subsec:comb-encode-automata}, we examine Mealy automata, which admit an additional 
combinatorial encoding as a pair of matrices, as noted in Section~\ref{sec:introduction}.

\subsection{Combinatorial encodings of letter transducers}
     \label{subsec:letter-transducers}

Our path between automata and Latin square-like objects runs through an analog of orthogonal arrays.
We say more about orthogonal arrays in Section~\ref{sec:array-rep}, but for now we note that 
they are often viewed as a collection of rows (tuples) and that we shall use the term array in this sense.
With this in mind, we introduce the following.
By an \emph{interleaved array} of order $m\times n$ 
we mean a subset of 
   $[m]\times[n]\times [m]\times [n]$.  
Let $\mathsf{IA}(m,n)$  denote the set of interleaved arrays of order $m\times n$.
For any tuple $y$, let $y_s$ denote the $s$-th entry of $y$.

To implement the combinatorial encodings of a letter transducer $\A=(Q,X,E)$, we must fix
   an ordering $(q_1, \ldots, q_m)$ of the set of states $Q$ and 
   an ordering $(x_1, \ldots, x_n)$ of the set of letters $X$. 
We discuss the effect of permuting $Q$ and $X$ during our consideration of isotopy in Section~\ref{sec:isotopism}.
We may identify the elements of $Q$ and $X$ with their indices, 
      so there is no loss in taking $Q=[m]$ and $X=[n]$ when convenient.
At this point, the correspondence between letter transducers and interleaved arrays is natural.

\begin{proposition}
\label{prop:FA-EFLT-bijection}
Fix ordered lists $Q=(q_1, \ldots, q_m)$ and $X=(x_1, \ldots, x_n)$.
\begin{enumerate}
    \item Suppose $\A=(Q, X, E)$ is a letter transducer. 
          Then\\
           $\mathcal{Y}_\A := \{(i,j,k,\ell)\in [m]\times[n]\times [m]\times [n]
                                           : (q_i, x_j, q_k, x_\ell)\in E \}\in \mathsf{IA}(m,n)$.

    \item Suppose $\mathcal{Y}\in \mathsf{IA}(m,n)$.  
          Let $E_\mathcal{Y} = \{(q_i, x_j, q_k, x_\ell )\in Q\times X\times Q\times X
                                           : (i,j,k,\ell)\in\mathcal{Y} \}$.
          Then $\A_\mathcal{Y} := (Q , X, E_\mathcal{Y})$ is a letter transducer.  
\end{enumerate}
Moreover, the maps 
     $\A\mapsto \mathcal{Y}_\A$  and 
     $\mathcal{Y}\mapsto\A_\mathcal{Y}$ 
are inverse bijections between 
    the set of letter transducers of the form $(Q, X, E)$ and 
    elements of $\mathsf{IA}(m,n)$.
We say that an interleaved array and a letter transducer related by these bijections are \emph{associated}.     
\end{proposition}

The elements of an interleaved array may be organized according to the value in each given coordinate.   
This induces four ordered partitions  of $[m]\times[n]\times [m]\times [n]$, 
with one (possibly empty) cell for each value in a given entry. 
This collection of partitions can be viewed as a geometric object that will resemble a $(k,n)$-net in the bireversible automaton case. 
We will say more about $(k,n)$-nets in Subsection~\ref{subsec:reticulations}. 
For now, we note that they consist of families of parallel lines on a common set of points.
%
Write $\mathsf{P}(S, d)$ to denote the set of ordered partitions of a set $S$ into $d$-many parts, 
some of which may be empty.  
Write $\mathsf{R}(S,m,n)
=\mathsf{P}(S, m)\times\mathsf{P}(S, n)\times\mathsf{P}(S, m)\times\mathsf{P}(S, n)$. 
We call an element of $\mathsf{R}(S,m,n)$ a \emph{rete} on $S$ of order $m\times n$.
Say that a rete $\mathcal{P}=(\{P^1_i\}, \{P^2_j\}, \{P^3_k\}, \{P^4_\ell\})$ is \emph{proper} when $|P^1_i\cap P^2_j\cap P^3_k\cap P^4_\ell|\leq 1$ for all $(i,j,k,\ell)\in[m]\times[n]\times [m]\times [n]$.
Let $\mathsf{PR}(S,m,n)$ be the set of proper rete on $S$ of order $m\times n$.
Write $\mathsf{PR}(\mathcal{Y})$ for $\mathsf{PR}(\mathcal{Y},m,n)$ when $\mathcal{Y}\in\mathsf{IA}(m,n)$.


\begin{lemma}
\label{lem:bijection-FA-R}
Fix a finite set $S$ such that $\mathsf{PR}(S,m,n)\not=\emptyset$.
\begin{enumerate}
\item  Suppose $\mathcal{Y}\in \mathsf{IA}(m,n)$. Define four ordered partitions as follows.
      \begin{enumerate}
          \item For $i\in[m]$, let  $P^1_i= \{(i',j,k,\ell)\in\mathcal{Y} : i'=i\}$,
                and let $\mathcal{P}^1=\{P^1_i\}_{i\in[m]}$.   
          \item For $j\in[n]$, let  $P^2_j= \{(i,j',k,\ell)\in\mathcal{Y} : j'=j\}$,
                and let $\mathcal{P}^2=\{P^2_j\}_{j\in[n]}$.           
          \item For $k\in[m]$, let  $P^3_k= \{(i,j,k',\ell)\in\mathcal{Y} : k'=k\}$,
                and let $\mathcal{P}^3=\{P^3_k\}_{k\in[m]}$.                            
          \item For $\ell\in[n]$, let  $P^4_\ell= \{(i,j,k,\ell')\in\mathcal{Y} : \ell'=\ell\}$,
                and let $\mathcal{P}^4=\{P^4_\ell\}_{\ell\in[n]}$.                   
      \end{enumerate}
      Then $\mathcal{P}_{\!\mathcal{Y}} := (\mathcal{P}^1,\mathcal{P}^2,\mathcal{P}^3,\mathcal{P}^4)\in \mathsf{PR}(\mathcal{Y})$.  
      
\item Let $\mathcal{P}\in\mathsf{PR}(S,m,n)$. 
      Then $\mathcal{Y}_\mathcal{P}
          :=\{ (i,j,k,\ell)\in[m]\times[n]\times [m]\times [n] 
                             : P^1_i\cap P^2_j\cap P^3_k \cap P^4_\ell\not=\emptyset \}\in\mathsf{IA}(m,n)$.
\end{enumerate}
Moreover, the maps   
    $\mathcal{Y}\mapsto \mathcal{P}_{\!\mathcal{Y}}$ and 
    $\mathcal{P}\mapsto\mathcal{Y}_\mathcal{P}$ 
are inverse bijections between the set of interleaved arrays of order $m\times n$ and the set of proper rete on $S$ of order $m\times n$.  
We say that an interleaved array and a proper rete related by these bijections are \emph{associated}.
\end{lemma}


\begin{proof}
By construction $\mathcal{P}_{\!\mathcal{Y}}$ is a rete on $\mathcal{Y}$ of order $m\times n$.
Note that $\mathcal{P}_{\!\mathcal{Y}}$ is proper since  $\mathcal{Y}$ is a set.
By construction, $\mathcal{Y}_\mathcal{P}$ is an interleaved array.  
Since $\mathcal{P}$ is proper, $\mathcal{Y}_\mathcal{P}$ is uniquely determined, so the maps are bijections.
\end{proof}

Composing the bijections of Proposition~\ref{prop:FA-EFLT-bijection} and Lemma~\ref{lem:bijection-FA-R}, we directly describe a correspondence between letter transducers and proper rete.

\begin{proposition}
\label{prop:EFLT=rete-bijection}
Fix ordered lists $Q=(q_1, \ldots, q_m)$ and $X=(x_1, \ldots, x_n)$ and a finite set $S$.
\begin{enumerate}
    \item Suppose $\A=(Q, X, E)$ is a letter transducer. 
    For $s\in\{1,3\}$ and $h\in[m]$, let $P^s_h=\{y\in E : y_s=q_h\}$.
    For $s\in\{2,4\}$ and $h\in[m]$, let $P^s_h=\{y\in E : y_s=x_h\}$.
    Then 
        $\mathcal{P}_\A := (\{P^1_i\}_{i\in[m]},\{P^2_j\}_{j\in[n]},\{P^3_k\}_{i\in[m]},\{P^4_\ell\}_{i\in[n] })
                        \in \mathsf{PR}(E,m,n)$.

    \item Suppose 
        $\mathcal{P}=(\{P^1_i\}_{i\in[m]},\{P^2_j\}_{j\in[n]},\{P^3_k\}_{i\in[m]},\{P^4_\ell\}_{i\in[n] })\in \mathsf{PR}(S,m,n)$. 
    Let 
        $E_\mathcal{P}:=\{(q_i,x_j,q_k,x_\ell) : P^1_i\cap P^2_j\cap P^3_k\cap P^4_\ell\not=\emptyset\}$. 
    Then $\A_{\mathcal{P}}:= (Q,X, E_\mathcal{P})$ is a letter transducer.
 
\end{enumerate}
Moreover, the maps 
     $\A\mapsto \mathcal{P}_\A$  and 
     $\mathcal{P}\mapsto\A_\mathcal{Y}$ 
are inverse bijections between 
    the set of letter transducers of the form $(Q, X, E)$ and 
    proper rete on $E$.
We say that a letter transducer and a proper rete that are related by these bijections are \emph{associated}.     
\end{proposition}

\begin{figure}[ht]
\centering
\begin{tikzpicture}
  \matrix[ampersand replacement=\&, row sep=.9cm, column sep=2cm] {
     \node[block] (aut) {$\begin{matrix} 
                \hbox{letter transducer}\\ 
                \A=(Q,X,E) \\ 
                Q=\{q_i\},\ X=\{x_j\}
                \end{matrix}$ }; \&  \& \& \\
    \node[block] (arr) {$\begin{matrix} \hbox{interleaved array}\\ \mathcal{Y} \end{matrix}$ }; \&  \& \&
    \node[block] (par) {$\begin{matrix} 
         \hbox{proper rete}\\ 
         \{P^1_i\}_{i\in[m]},   
         \{P^2_j\}_{j\in[n]} \\
         \{P^3_k\}_{k\in[m]},  
         \{P^4_\ell\}_{\ell\in[n]} 
         \end{matrix}$ };
\\
  };

\draw[connector,<->] (aut) -- node[right] 
    {$\begin{array}{l} 
          (q_i,x_j,q_k,x_\ell)\in E \\ 
          (i,j,k,\ell)\in \mathcal{Y} 
    \end{array}$} (arr); 

\draw[connector,<->] (arr) -- node 
   {$\begin{array}{c} 
       P^s_h =\{y\in\mathcal{Y} : y_s=h\}  \\[4pt]
       \{(i,j,k,\ell) : P^1_i\cap P^2_j\cap P^3_k\cap P^4_\ell\not=\emptyset\}
    \end{array}$}  (par);

\draw[connector,<->] (aut) --  ++(5.5,0) node[yshift=7pt]
   {$\begin{array}{c} 
       P^s_h =\{y\in{E} : y_s=q_h\} \hbox{ ($s$ odd)} \\
       P^s_h =\{y\in{E} : y_s=x_h\} \hbox{ ($s$ even)}  \\[4pt]
       \{(q_i,x_j,q_k,x_\ell) : P^1_i\cap P^2_j\cap P^3_k\cap P^4_\ell\not=\emptyset\}
    \end{array}$}  -|   (par);
\end{tikzpicture}
\caption{Summary of bijections}
\end{figure}
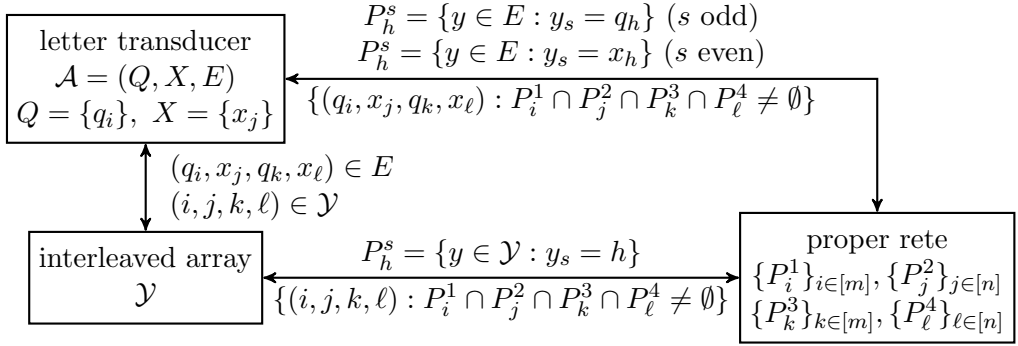

\begin{example}
\label{ex:eflt}
In Figure~\ref{fig:eftl} we present a letter transducer (and its diagram) together with its associated interleaved array and rete.
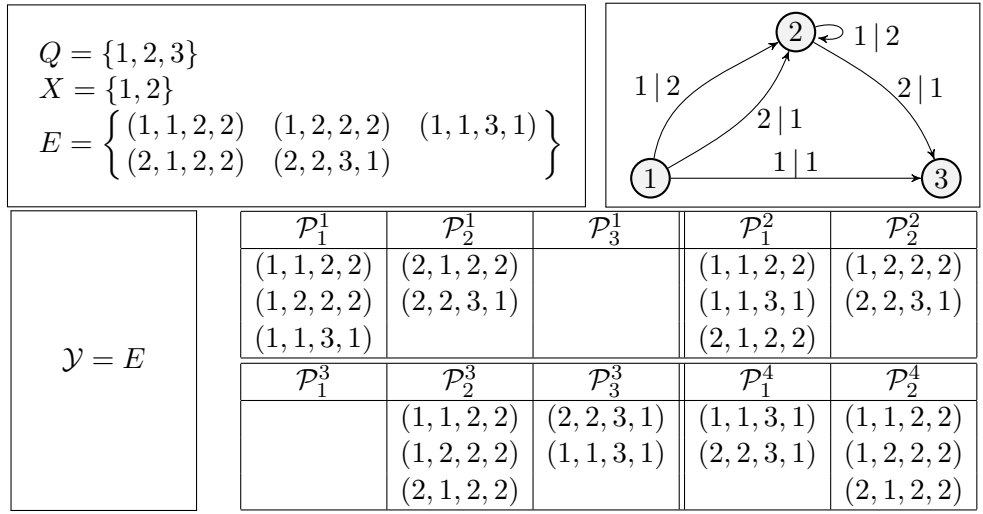
\begin{figure}[ht]
    \centering
    \noindent
    \fbox{
$\begin{array}{lll}
\\[-5pt]
Q=\{1,2,3\} \\ 
X=\{1,2\}\\
E=\left\{\!\! \begin{array}{ccc}(1,1,2,2) & (1,2,2,2) & (1,1,3,1) \\ 
                                (2,1,2,2) & (2,2,3,1)
\end{array}\!\!\right\}
\\[-7pt]
\\
\end{array}$} 
  \ 
\fbox{
    \begin{tikzpicture}[node distance=.85cm and 2cm,  baseline={([yshift=-3pt]current bounding box.center)}]
    \node[state, inner sep=2pt,minimum size=1pt] (q2) {$2$};
    \node[state,inner sep=2pt,minimum size=1pt, below left =of q2,xshift=13pt,yshift=-20pt] (q1) {$1$};
    \node[state,inner sep=2pt,minimum size=1pt, below right =of q2,xshift=-13pt,yshift=-20pt] (q3) {$3$};

    \draw (q2) edge[loop right] node[right,xshift=0pt,yshift=-2pt]{$1\,|\, 2$}  (q2);
    \draw (q2) edge [out=330,in=110]  node[right]{$2\,|\, 1$} (q3);
 
    \draw (q1) edge [out=0,in=180] node[above,yshift=-4pt]{$1\,|\,1$} (q3);
    \draw (q1) edge [out=80,in=210] node[left,xshift=-3pt]{$1\,|\,2$} (q2);
    \draw (q1) edge [out=30, in=250] node[right,xshift=2pt,yshift=0pt]{$2\,|\, 1$}  (q2);
    \end{tikzpicture}
}

\fbox{$\begin{array}{c}\\ \\ \\ \\[-8pt] \quad \mathcal{Y}=E \quad\\ \\ \\ \\ \\[-8pt] \end{array}$}
\quad\!\!\!
{
$\begin{array}{|c|c|c||c|c|c|c|}
\hline
\mathcal{P}^1_1 & \mathcal{P}^1_2 & \mathcal{P}^1_3 & \mathcal{P}^2_1& \mathcal{P}^2_2\\
\hline
(1,1,2,2)       & (2,1,2,2)        &                & (1,1,2,2)      &(1,2,2,2)\\ 
(1,2,2,2)       & (2,2,3,1)        &                & (1,1,3,1)      &(2,2,3,1)\\ 
(1,1,3,1)       &                  &                & (2,1,2,2)      & \\
\hline
\hline
\mathcal{P}^3_1 & \mathcal{P}^3_2 & \mathcal{P}^3_3 & \mathcal{P}^4_1& \mathcal{P}^4_2\\
\hline
                & (1,1,2,2)       & (2,2,3,1)       & (1,1,3,1)      & (1,1,2,2)\\
                & (1,2,2,2)       & (1,1,3,1)       & (2,2,3,1)      & (1,2,2,2)\\
                & (2,1,2,2)       &                 &                & (2,1,2,2)\\
\hline
\end{array}$}
    \caption{A letter transducer, its diagram, and the associated array and rete}
    \label{fig:eftl}
\end{figure}    
\end{example}

We next discuss conditions on interleaved arrays and rete that correspond to the Mealy, reversible, invertible, and coreversible properties for the associated letter transducer. 
For an interleaved array $\mathcal{Y}$ and distinct $s$, $t\in[4]$, say 
\begin{equation}
\hbox{
    $\mathcal{Y}$ has property $\un{st}$ whenever for all $y$, $z\in \mathcal{Y}$, 
    if $y_s=z_s$ and $y_t=z_t$, then $y=z$.
}
\tag{$\mathfrak{U}$}
\end{equation}
Note that $\mathcal{Y}$ has property $\un{st}$ if and only if it has property $\un{ts}$.

We refer to the cardinality of an interleaved array $\mathcal{Y}$ as its \emph{size}. 
The size of an interleaved array is the number of transitions in an associated letter transducer.

\begin{lemma}
\label{lem:UNsizebound}
If $\mathcal{Y}\in \mathsf{IA}(m,n)$ has property $\un{st}$ for $s$ and $t$ of different parity, 
then $\mathcal{Y}$ has size at most $mn$.
\end{lemma}

We are interested in the case where equality holds in the preceding lemma. We say that

\begin{equation}
\hbox{
    $\mathcal{Y}\in\mathsf{IA}(m,n)$ has property $\mathfrak{T}$ whenever $|\mathcal{Y}|=mn$.
}
\tag{$\mathfrak{T}$}
\end{equation}
Say that $\mathcal{Y}\in\mathsf{IA}(m,n)$ has property $\ut{s,t}$ when it has properties $\un{st}$ and $\mathfrak{T}$. 

\begin{definition}
Let $\mathcal{Y}\in \mathsf{IA}(m,n)$.
For distinct  $s$, $t\in \{1,2,3,4\}$,
   let $\mathcal{Y}_{\take{st}}  = \{ (y_s, y_t) \,|\,  y\in \mathcal{Y}\}$, 
   where $y_h$ denotes the $h$-th entry of $y$.
\end{definition}

View $\mathcal{Y}_{\take{st}}$ as taking columns $s$ and $t$ of the array whose rows are the tuples in $\mathcal{Y}$.

\begin{definition}
Let $S$ be a set, and let  
     ${\mathcal{P}^1}=\{P^{1}_i\}_{i\in \mathcal{M}}$ and 
     ${\mathcal{P}^2}=\{P^{2}_j\}_{j\in \mathcal{N}}$ 
be partitions of $S$. 
Then ${\mathcal{P}^1}$ and ${\mathcal{P}^2}$ are said to be \emph{orthogonal} whenever 
     $|P^1_i\cap P^2_j|=1$ for all $i\in \mathcal{M}$ and $j\in \mathcal{N}$. 
In this case, we write ${\mathcal{P}^1} \perp {\mathcal{P}^2}$.
\end{definition}

See~\cite{bailey:ASDEAC} for a discussion of orthogonal partitions.
If $\mathcal{P}^1$ and $\mathcal{P}^2$ are orthogonal partitions of some set $S$, then for $\{i,j\}=\{1,2\}$
every cell of $\mathcal{P}^i$ contains as many elements as $\mathcal{P}^j$ has cells. 
Moreover, for each of the partitions, the size of $S$ is the product of the common cell size and the number of cells.

\begin{theorem}
\label{thm:U(st)-Y<st>}\label{lem:GA-EAorth}
Let $\mathcal{Y}\in\mathsf{IA}(m,n)$.
Let $({\mathcal{P}^1}, {\mathcal{P}^2}, {\mathcal{P}^3},{\mathcal{P}^4})$ be the rete associated with $\mathcal{Y}$. 
Let $s\in\{1,3\}$ and $t\in\{2,4\}$.
Then the following are equivalent.
\begin{enumerate}
\item $\mathcal{Y}$  has property $\ut{st}$. 
%
\item $\mathcal{Y}_\take{st}=[m]\times [n]$. 
\item ${\mathcal{P}^s}$ and ${\mathcal{P}^t}$ are orthogonal.
\end{enumerate}
\end{theorem}

\begin{proof}
By assumption $|\mathcal{Y}|=mn$, and 
by construction  $\mathcal{Y}_\take{st}$ is a subset of    $[m]\times [n]$.
Now $\mathcal{Y}$ has property $\un{st}$ if and only if
no two elements of $\mathcal{Y}$ have the same entries in positions $s$ and $t$ if and only if
 $|\mathcal{Y}_\take{st}| = mn$ if and only if 
 $\mathcal{Y}_\take{st}=[m]\times [n]$.

Write ${\mathcal{P}^s}=\{P^{s}_i\}_{i\in [m]}$ and ${\mathcal{P}^t}=\{P^{t}_j\}_{j\in [n]}$.
In the rete of $\mathcal{Y}$, $P^s_i\cap P^t_j= \{(i,j,k,\ell)\in\mathcal{Y} : y_s=i, y_t=j\}$.
Now  
${\mathcal{P}^1}$ and ${\mathcal{P}^2}$ are orthogonal
         if and only if 
$|P^s_i\cap P^t_j|=1$ for all $(i,j)\in[m]\times[n]$ 
          if and only if
$|\mathcal{Y}|\geq mn$ and         
for $y$, $z\in\mathcal{Y}$,  $y_s=z_s$ and $y_t=z_t$ implies $y=z$, that is, $\un{st}$ holds.  
Now $\mathcal{Y}$ has property $\mathfrak{T}$ by Lemma~\ref{lem:UNsizebound}.
\end{proof}

We now relate the combinatorial properties of Theorem~\ref{thm:U(st)-Y<st>} to classes of letter transducers.

\begin{theorem}
\label{thm:classes-U(st)}
Let $\A$, $\mathcal{Y}$, and $\mathcal{P}= (\mathcal{P}^1,\mathcal{P}^2,\mathcal{P}^3,\mathcal{P}^4)$ 
be an associated  letter transducer, interleaved array, and rete, respectively.
Then $\A$ is  
       Mealy, 
       reversible, 
       coreversible, or 
       invertible, respectively, 
   if and only if 
   the equivalent conditions of Theorem~\ref{thm:U(st)-Y<st>} hold for 
       $(s,t)=(1,2)$, 
       $(s,t)=(2,3)$, 
       $(s,t)=(3,4)$, or  
       $(s,t)=(4,1)$, respectively.
\end{theorem}

\begin{proof}
The result is straightforward from Definitions~\ref{def:mealy} and~\ref{def:classes}.
\end{proof}

\subsection{Connection with $\mathcal{VH}$ square complexes}
     \label{subsec:square-complex}

A construction similar to interleaved arrays in the context of $\mathcal{VH}$ square complexes under the name $\mathcal{VH}$-datum was introduced in the paper of Burger and Mozes~\cite{burger_m:lattices_in_product_of_trees00}.

Fix finite sets $Q$ and $X$ of cardinalities $m\geq 2$ and $n\geq 2$, respectively. Let $Q^{-1}$ and $X^{-1}$ denote the sets containing formal inverses to elements of $Q$ and $X$. A $(2m,2n)$ $\mathcal{VH}$ \emph{square complex} $\mathcal S$ is a 1-vertex square complex whose 1-skeleton is a wedge of $2m+2n$ circles, labeled by elements of $Q$, $X$, $Q^{-1}$, and $X^{-1}$, and whose 2-cells are squares with the labeling scheme in which labels from $Q\sqcup Q^{-1}$ alternate with the ones from $X\sqcup X^{-1}$. A complex $\mathcal S$ is called \emph{complete} if for each $q\in Q\sqcup Q^{-1}$ and $x\in X\sqcup X^{-1}$ there are unique $q'\in Q\sqcup Q^{-1}$ and $x'\in X\sqcup X^{-1}$ such that there is a square in $\mathcal S$ with labeling scheme $qxq'x'$ (taking into account that each square gives rise to up to 8 labeling schemes coming from the action of the dihedral group $D_4$ on this square). This condition is equivalent to the, so-called, \emph{link condition}, requiring the link at the unique vertex of $\mathcal S$ to be complete bipartite graph with the two parts coming from the partition of edges into horizontal and vertical ones.

Such complexes appear naturally in the theory of lattices. Namely, the universal cover of a complete square complex $\mathcal S$ is homeomorphic to a product of regular unrooted trees $\mathcal T_m\times\mathcal T_n$, and its fundamental group $\pi_1(\mathcal S)<\Aut(\mathcal T_m\times\mathcal T_n)$ is a finitely presented torsion-free cocompact lattice that naturally acts on this space freely and transitively on the vertices. The notion of $\mathcal{VH}$-datum from~\cite{burger_m:groups_acting_on_trees00} in our notation is a 5-tuple $(Q\sqcup Q^{-1},X\sqcup X^{-1}, \phi_Q, \phi_X, R)$, where $\phi_Q\colon Q\sqcup Q^{-1}\to Q\sqcup Q^{-1}$ and $\phi_X\colon X\sqcup X^{-1}\to X\sqcup X^{-1}$ are involutions sending every element to its inverse, and 
\[R\subset (Q\sqcup Q^{-1})\times (X\sqcup X^{-1})\times (Q\sqcup Q^{-1})\times (X\sqcup X^{-1}).\] 
Therefore, the set $R$ corresponds to the set of transitions of a letter transducer and can be described by an interleaved array.

Glasner and Mozes~\cite{gl_mo:compl} were the first to connect these complexes to Mealy automata via a natural link associating a square form Figure~\ref{fig:classes} to each transition
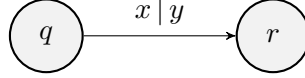
\begin{center}
    \begin{tikzpicture}
    \node[state] (q1) {$q$};
    \node[state, right of=q1] (q2) {$r$};
    \draw (q1) edge[above] node{$x\,|\,y$} (q2);
    \end{tikzpicture}
\end{center}
in a Mealy automaton. Clearly, the same construction can be applied in the more general case of letter transducers. Each such automaton $\A$ gives rise to a $\mathcal{VH}$ square complex $\mathcal S_{\A}$ as described in the previous paragraph. Moreover, each such complex will be \emph{directed} in the sense that the parallel edges in the squares are required to have the same orientation. One can also construct a letter transducer $\A_{\mathcal S}$ from a $\mathcal{VH}$ square complex $\mathcal S$, but this construction is not exactly inverse to the above procedure. Namely, each $(2m,2n)$ $\mathcal{VH}$ square complex $\mathcal S$ defines a letter transducer $\mathcal A_{\mathcal S}$ with the set of states $Q\sqcup Q^{-1}$ acting on $(X\sqcup X^{-1})^*$. Each of the squares of $\mathcal S$ defines 4 transitions in $\A_{\mathcal S}$ corresponding to two corners in each of the orientations of this square. For a letter transducer $\A$ with the set of states $Q$ acting on $X^*$, the letter transducer $\A_{\mathcal S_{\A}}$ describes not only the action of $Q$ on letters from $X$, but also of $Q\sqcup Q^{-1}$ on $X\sqcup X^{-1}$. The automata $\A, \iz\A, \dz\iz\dz\A$, and $\iz\dz\iz\dz\A=\dz\iz\dz\iz\A$ can be recovered from $\A_{\mathcal S_{\A}}$ by restricting the state sets to either $Q$ or $Q^{-1}$ and considering the actions on the subtrees $X^*$ or $(X^{-1})^*$ of $(X\sqcup X^{-1})^*$. In a similar fashion, the automata $\dz\A, \iz\dz\A, \dz\iz\A$, and $\dz\iz\dz\A$ can be recovered from $\A_{\mathcal S_{\dz\A}}$.

Bondarenko and Kivva in~\cite{bondarenko_k:automaton_groups_and_square_complexes22} used this connection to prove that fundamental groups of certain complete square complexes constructed in~\cite{wise:PhD96} as finitely presented small-cancellation, CAT(0), and automatic groups are non-residually finite, thus contributing to an open question of whether word-hyperbolic groups are residually finite. In particular, one of the complexes considered by Wise turned out to be associated with the Aleshin automaton.

Stix and Vdovina in~\cite{stix_v:simply_transitive_quaternionic_lattices17} produced a mass formula to calculate the number of complete $\mathcal{VH}$-square complexes generalizing earlier results of Rattaggi~\cite{rattaggi:PhD04}. More recently, variants of this structure were used in~\cite{bondarenko_gv:ramanujan_subshifts} in the context of quaternionic lattices and Ramanujan subshifts. In particular, families of Ramanujan graphs as Schreier graphs of automaton groups associated to certain $\mathcal{VH}$-square complexes were explicitly constructed. In a recent preprint~\cite{pask:full_mealy_automata_complete_square_complexes_and_anti-tori}, Pask provides conditions under which a square complex associated to co-reversible Mealy automaton contains an anti-torus, one of which requires an automaton to be bireversible.

We note that most of the mentioned papers do not concentrate on \emph{directed} $\mathcal{VH}$ square complexes. Therefore, while each of such complete square complex does define a bireversible automaton, such an automaton will have the set of states that contains inverses to all generators and will act on a rooted $2n$-regular tree corresponding to $X\sqcup X^{-1}$ (see, for example, Theorem~5.1 in~\cite{bondarenko_k:automaton_groups_and_square_complexes22}). In particular, the mass formula for complete square complexes given in~\cite{stix_v:simply_transitive_quaternionic_lattices17} describes this class of bireversible automata. In this paper, we investigate bireversible automata corresponding to directed $\mathcal{VH}$ square complexes. On one hand, this is a more restricted class of bireversible automata as we impose an extra condition on the corresponding square complexes; on the other hand, we are not requiring the state sets and alphabets of our automata to be symmetric (to contain inverses of all elements). 

We note that similar complexes consisting of triangles have been studied for Latin squares
\cite{Grannell:biembeddings,Griggs:maximumgenus}.

\subsection{Combinatorial encodings of Mealy automata}
     \label{subsec:comb-encode-automata}

We now consider combinatorial representations of Mealy automata. 
We begin by establishing some notation needed to describe the encodings. 
The encodings resemble 
      orthogonal arrays, 
      Latin squares, and 
      $(k,n)$-nets,
and they are related similarly.
As was the case for letter transducers,  orderings of the sets of states and letters must be fixed.

Viewing a Mealy $(m,n)$-automaton as a letter transducer (Definition~\ref{def:mealy}), we encode its transition set as the relevant interleaved arrays from Theorem~\ref{thm:classes-U(st)}.   
Let $\mathsf{GA}(m,n)$ denote the subset of $\mathsf{IA}(m,n)$ consisting of 
     interleaved arrays with property $\ut{12}$.  
We refer to an element of $\mathsf{GA}(m,n)$ as a \emph{grid array} of order $m\times n$.
The restriction of a grid array of order $m\times n$ to the first two components is $[m]\times[n]$ by Theorem~\ref{thm:U(st)-Y<st>}, 
so these components name positions in an $m\times n$ matrix or the points of an $m\times n$ grid. 
The latter observation motivates the name.

Another encoding of a Mealy $(m,n)$-automaton arises from specializing the associated rete with the corresponding conditions of Theorem~\ref{thm:U(st)-Y<st>}.
Write $\mathsf{GR}(m,n)$ to denote the set of rete on grid arrays of order $m\times n$.
In a grid rete the first two partitions are orthogonal by Lemma~\ref{lem:bijection-FA-R} and Theorem~\ref{thm:classes-U(st)}; in particular, for all $(i,j)\in[m]\times[n]$ there is a unique element in $\mathcal{P}^1_i\cap\mathcal{P}^2_j$.
We refer to an element of $\mathsf{GR}(m,n)$ as a \emph{grid rete}. 

The structure of a Mealy $(m,n)$-automaton admits another combinatorial encoding that is not available to general letter transducers. 
Starting from  Definition~\ref{def:mealy2}, a natural approach  is to encode each of 
     the transition function $\pi\colon Q\times X\rightarrow Q$ and 
     the output function $\lambda\colon Q\times X\rightarrow X$ as a matrix.
Let $\mathsf{M}_{m\times n}(E)$ denote the set of $m\times n$ matrices with entries from some set $E$.
Let  $\mathsf{MM}(m,n)=\mathsf{M}_{m\times n}([m])\times \mathsf{M}_{m\times n}([n])$. 
We refer to an element of $\mathsf{MM}(m,n)$ as a \emph{matched (pair of) matrices} of order $m\times n$. 
We elaborate on the encoding of finite automata by matched matrices in Section~\ref{sec:matchedmatrix-rep}.

\begin{example}
The letter transducer  in Example~\ref{ex:eflt} is neither deterministic nor complete, 
so it does not admit an encoding by matched matrices:  
Some entries would be assigned multiple values, and others would be assigned no value. 
\end{example}

To the bijections in 
    Proposition~\ref{prop:FA-EFLT-bijection}, 
    Lemma~\ref{lem:bijection-FA-R}, and 
    Proposition~\ref{prop:EFLT=rete-bijection}, 
we add bijections with matched matrices.
Here we use the notation of Definition~\ref{def:mealy2} for Mealy automata.

\begin{definition}
\label{def:bijections}
Fix ordered lists $Q=(q_1, \ldots, q_m)$ and $X=(x_1, \ldots, x_n)$.
\begin{enumerate}
\item  Let $(C,R)\in \mathsf{MM}(m,n)$.
	\begin{enumerate}
	\item 
            Define 
    		$\pi_{(C,R)}\colon Q\times X\rightarrow Q$ by $\pi_{(C,R)}(q_i, x_j) = q_{C(i,j)}$ and 
    		$\lambda_{(C,R)}\colon Q\times X\rightarrow X$ by $\lambda_{(C,R)}(q_i, x_j) = x_{R(i,j)}$ $((i,j)\in[m]\times[n])$.  
    		Let $\A_{(C,R)} = (Q, X, \pi_{(C,R)}, \lambda_{(C,R)})$.
	
	\item Define     
     		$\mathcal{Y}_{(C,R)} \in \mathsf{GA}(m,n)$ by\\
    		$\mathcal{Y}_{(C,R)} = \{ (i,j,k,\ell) :  (i,j)\in[m]\times[n], k=C(i,j), \ell=R(i,j)\}$.
		 
	\item Define  $\mathcal{P}_{(C,R)} \in \mathsf{GR}(m,n)$ by
         $\mathcal{P}_{\!(C,R)}=(\mathcal{P}^1_{\!(C,R)},\mathcal{P}^2_{\!(C,R)},\mathcal{P}^3_{\!(C,R)},\mathcal{P}^4_{\!(C,R)})$,\\
            where
               $\mathcal{P}^s_{\!(C,R)}=\{P^s_h\}_{h\in[d]}$ 
            with 
               $P^s_h = \{y=(i,j,C(i,j),R(i,j)) :  (i,j)\in[m]\times[n],y_s=h\}$ 
               ($d=m$ if $s\in\{1,3\}$, $d=n$ if $s\in\{2,4\}$).

   \end{enumerate}        

\item 
    \begin{enumerate}
    \item Let $\mathcal{A}=(Q, X, \pi,\lambda)$ be a Mealy automaton.
      Define 
    	 $(C_\A,R_\A)\in \mathsf{MM}(m,n)$ by 
    		 $C_\A(i,j) = k$, where $\pi(q_i, x_j)=q_k$, and 
   		 $R_\A(i,j) = \ell$, where $\lambda(q_i, x_j)=x_\ell$    
    		         $((i,j)\in[m]\times[n])$.

    \item Let $\mathcal{Y}\in\mathsf{GA}(m,n)$. 
      Define $(C_\mathcal{Y}, R_\mathcal{Y})\in \mathsf{MM}(m,n)$ by 
		$C_\mathcal{Y}(i,j)= k$ and 
		$R_\mathcal{Y}(i,j)=\ell$ 
		when $(i,j,k,\ell)\in \mathcal{Y}$.

    \item Let $\mathcal{P}=(\mathcal{P}^1, \mathcal{P}^2,\mathcal{P}^3, \mathcal{P}^4)\in\mathsf{GR}(m,n)$.
      Let $p_{i,j}$ denote the unique element of $\mathcal{P}^1_i\cap \mathcal{P}^2_j$.
      Define $(C_\mathcal{P}, R_\mathcal{P})\in \mathsf{MM}(m,n)$ by 
		$C_\mathcal{P}(i,j)= k$ where $p_{i,j}\in P^3_k$ and 
		$R_\mathcal{P}(i,j)=\ell$ where $p_{i,j}\in P^4_\ell$.
    \end{enumerate}              
\end{enumerate}
\end{definition}

Let $\mathsf{MA}(m,n)$ denote the set of Mealy automata with states $Q=[m]$ and letters $X=[n]$.

\begin{theorem}
\label{thm:bijection-pairs}
\label{lem:bijection-automata-matrices}
The following pairs of maps in Proposition~\ref{prop:FA-EFLT-bijection}, 
    Lemma~\ref{lem:bijection-FA-R},  
    Proposition~\ref{prop:EFLT=rete-bijection}, and 
    Definition~\ref{def:bijections} are inverse bijections.
\[
\begin{array}{| rcl || rl | rl |}
\hline
\mathsf{MA}(m,n) &\leftrightarrow & \mathsf{GA}(m,n)   & \A&\mapsto\mathcal{Y}_\A                      & \mathcal{Y}&\mapsto \A_\mathcal{Y}                 \\ \hline
\mathsf{MA}(m,n) &\leftrightarrow & \mathsf{MM}(m,n)   & \A&\mapsto(C_\A, R_\A)                        & (C,R)&\mapsto \A_{(C,R)}                           \\ \hline
\mathsf{MA}(m,n) &\leftrightarrow & \mathsf{GR}(m,n)   & \A&\mapsto\mathcal{P}_\A                      & \mathcal{P}&\mapsto \A_\mathcal{P}                 \\ \hline
\mathsf{MM}(m,n)&\leftrightarrow & \mathsf{GA}(m,n)   & (C,R)&\mapsto \mathcal{Y}_{(C,R)}             & \mathcal{Y}&\mapsto(C_\mathcal{Y}, R_\mathcal{Y})  \\ \hline
\mathsf{MM}(m,n)&\leftrightarrow & \mathsf{GR}(m,n)   & (C,R)&\mapsto \mathcal{P}_{(C,R)}             & \mathcal{P}&\mapsto (C_\mathcal{P}, R_\mathcal{P}) \\ \hline
\mathsf{GA}(m,n)&\leftrightarrow & \mathsf{GR}(m,n)   & \mathcal{Y}&\mapsto \mathcal{P}_\mathcal{Y}   &\mathcal{P}&\mapsto \mathcal{Y}_\mathcal{P}         \\ \hline
 \end{array}
\]
We say that a 
    Mealy automaton, 
    matched pair of matrices, 
    grid array, and 
    grid rete 
related by these bijections are \emph{associated}.
\end{theorem}

\begin{proof}
Each pair of maps are inverses by construction and hence bijections.
\end{proof}

The bijections among matched matrices, grid arrays, and grid rete are essentially those relating 
    orthogonal Latin squares,  
    (certain) orthogonal arrays, and 
    $(k,n)$-nets.  
As we move toward 
    cooperative pairs, 
    narrow semi-orthogonal arrays, and 
    reticulations 
(the bireversible Mealy automata case), the analogy grows stronger.   
Indeed,  in Example~\ref{ex:LatinSquare} we saw that orthogonal Latin squares are cooperative pairs.  
We shall show that orthogonal arrays are to a narrow semi-orthogonal array in Subsection~\ref{subsec:birev-array}.
We shall show that $(k,n)$-nets are reticulations in Subsection~\ref{subsec:reticulations}.

\begin{figure}[ht]
\begin{tikzpicture}
  \matrix[ampersand replacement=\&, row sep=2.45cm, column sep=1.9cm] {
     \node[block] (aut) {$\begin{matrix} 
                             \hbox{automaton}\\ 
                             \A=(Q,X,\pi,\lambda) \\ 
                             Q=\{q_i\},\ X=\{x_j\}
                          \end{matrix}$ }; \& 
 \& \&
     \node[block] (mat) {$\begin{matrix} \hbox{matched matrices}\\ (C,R) \end{matrix}$ };\\
     \node[block] (arr) {$\begin{matrix} \hbox{grid array}\\ \mathcal{Y} \end{matrix}$ }; \& 
 \& \&
    \node[block] (par) {$\begin{matrix} \hbox{grid rete}\\ 
        \{P^1_i\}_{i\in[m]},
        \{P^2_j\}_{j\in[n]} \\
        \{P^3_k\}_{k\in[m]}, 
        \{P^4_\ell\}_{\ell\in[n]} \end{matrix}$ };
\\
  };

\draw[connector,<->] (aut) -- node {
          $\begin{array}{c}
             C(i,j)=k\\[0.6mm] 
            \pi(q_i,x_j)=q_k\end{array}$ 
         $\begin{array}{c}
             R(i,j)=\ell\\[0.6mm] 
             \lambda(q_i,x_j)=x_\ell \end{array}$ } (mat);

\draw[connector,<->] (aut) -- node[left] 
   {$\begin{array}{cc}  \pi(q_i,x_j)=q_k\\  
                        \lambda(q_i,x_j)=x_\ell\\ 
                        (i,j,k,\ell)\in \mathcal{Y}\\ \end{array}$} (arr); 

\draw[connector,<->] (arr) -- node[sloped] {$\begin{array}{c}C(i,j)=k \quad R(i,j)=\ell\\[0.6mm] (i,j,k,\ell)\in \mathcal{Y} \end{array}$ } (mat);

\draw[connector,<-> ] (mat) -- node[xshift=-0.6mm]
{$\begin{array}{@{}c@{}} 
    \  C(i,j)=k, \ R(i,j)=\ell  :\\ \qquad\quad P^1_i\cap P^2_j\cap P^3_k\cap P^4_\ell\not=\emptyset\\
     \hline
    \!\!\!\!\!\! \{(i,j,k=C(i,j),\ \ \ \ell=R(i,j)) :\\ 
    i'\overset{P^1_{i'}}{=}i | j'\overset{P^2_{j'}}{=}j | k'\overset{P^3_{k'}}{=}k |\ell'\overset{P^4_{\ell'}}{=}\ell 
    \} \ \
 \end{array}$ } (par);

\draw[connector,<->] (arr) -- node
  {$\begin{array}{c} 
       P^s_h =\{y\in\mathcal{Y} : y_s=h\}  \\[4pt]
       \{(i,j,k,\ell) : P^1_i\cap P^2_j\cap P^3_k\cap P^4_\ell\not=\emptyset\}
    \end{array}$}  (par);

\draw[connector,<->] (aut.north) |- ++(0,5mm) -- ++(5.5,0) node[yshift=7pt]{
$\begin{array}{c}
\begin{array}{rcc|c|c|c} 
\!\! && P^1_{i'}& P^2_{j'} & P^3_{k'} & P^4_{\ell'}\\
\{(q_{i}, x_{j},q_{k},x_{\ell}) \!\!\!&\!\!\! :\!\!\!&\!\!\!
i=i'& j=j'& k=k' &\ell=\ell'\}
\end{array}\\[8pt]
\pi(q_{i},x_{j})=q_{k},\lambda(q_{i},x_{j})=x_{\ell}  : P^1_i\cap P^2_j\cap P^3_k\cap P^4_\ell\not=\emptyset
\end{array}$}
-|  ($(mat.north east) +(12mm,0)$) |- (par);

\end{tikzpicture}
\caption{Summary of bijections}
\end{figure}
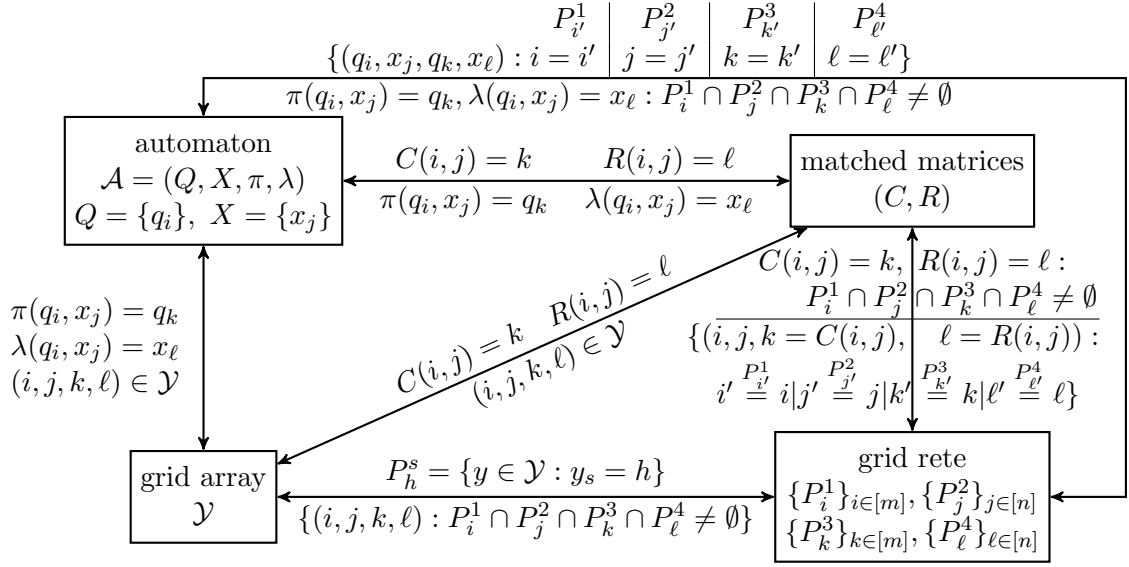

\begin{example}
\label{ex:runningexample1}
In Figure~\ref{fig:runningexample}, we present an associated 
    automaton,  
    matched pair of matrices,  
    grid array, and 
    grid rete.
We shall refer to this example throughout the paper, as they are in fact associated
    bireversible Mealy automaton, 
    cooperative pair, 
    narrow semi-orthogonal array, and 
    reticulation.

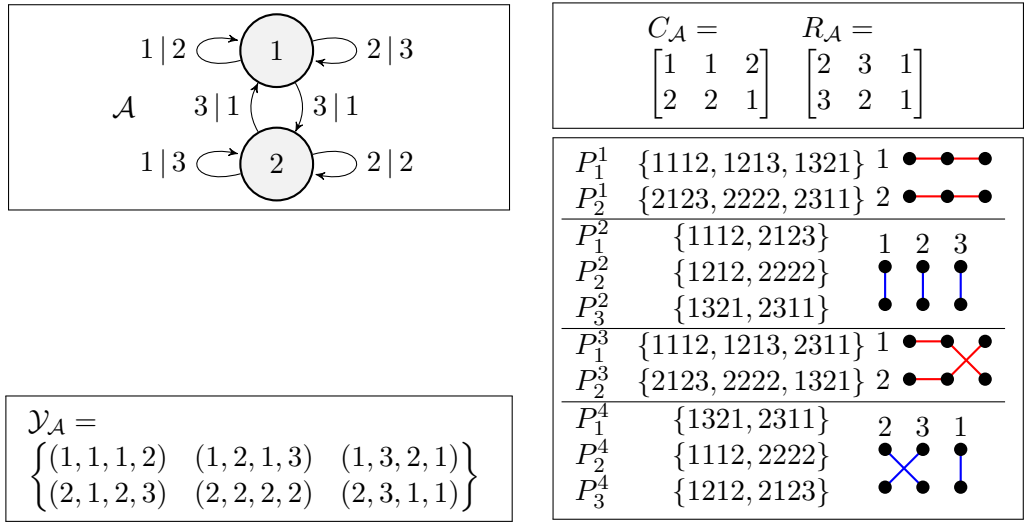
\begin{figure}[htbp!]
    \centering
\begin{tikzpicture}
     \node[blocknf,yshift=-16pt] (aut) at (0,5.21) { \fbox{\qquad\ \  \begin{tikzpicture}[node distance=1.5cm, baseline={([yshift=-.8ex]current bounding box.center)}]
    \node[state] (q1) {$1$};
    \node[state, below of=q1] (q2) {$2$};
    \draw (q1) edge[right, bend left] node{$3\,|\,1$} (q2);
    \draw (q2) edge[left, bend left] node (e) {$3\,|\,1$} (q1);
    \node[left=.4cm of e] (A) {$\A$};
    \draw (q1) edge[loop left] node{$1\,|\, 2$} (q1);
    \draw (q1) edge[loop right] node{$2\,|\, 3$}  (q1);
    \draw (q2) edge[loop left] node{$1\,|\, 3$} (q2);
    \draw (q2) edge[loop right] node{$2\,|\, 2$}  (q2);
    \end{tikzpicture}\qquad\ \ }}; 
 %
     \node[blocknf] (mat) at (7,5.2) {  \fbox{\qquad\ \,$\begin{array}{ll}
    C_\A = & R_\A=\\
    \begin{bmatrix} 1&1&2\\2&2&1\end{bmatrix} & \begin{bmatrix} 2&3&1\\3&2&1\end{bmatrix}
    \end{array}$\qquad\ \,}};
 %
     \node[blocknf] (arr) at (0,0) {\fbox{$ \begin{array}{l}
      \mathcal{Y}_\A=\\
      \left\{\!\!\! \begin{array}{ccc}
                         (1,1,1,2) & (1,2,1,3) &  (1,3,2,1)\\ 
                         (2,1,2,3) & (2,2,2,2) & (2,3,1,1) \end{array} \!\!\!\right\}                            
    \end{array}$ }}; 
 %
    \node[blocknf, yshift=49pt] (par) at (7,0){\fbox{$      \begin{array}{lccc}
       P^1_1 &  \{1112,  1213, 1321\} &   \!\!\!\!\!   \!\!\!\!\! 
       \raisebox{0pt}[0pt][0pt]{
       \begin{tikzpicture}[scale=.5,baseline=5ex]
          \draw[thick, red, -] (1,1)--(2,1)--(3,1); 
          \draw[thick, red, -] (1,2)--(2,2)--(3,2); 
          \node (l1) at (.3,2) {$1$};
          \node (l2) at (.3,1) {$2$};
          \gridpoints{2}{3};
       \end{tikzpicture}}\\
    P^1_2 & \{2123, 2222, 2311\} \\
       \hline
     P^2_1&  \{1112, 2123\}\\
     P^2_2 &  \{1212, 2222\}\\
     P^2_3 &  \{1321,2311\} &       \!\!\!\!\!    \!\!\!\!\! 
         \raisebox{0pt}[0pt][0pt]{
         \begin{tikzpicture}[scale=.5,baseline=2ex]
             \draw[thick, blue, -] (1,1)--(1, 2);
             \draw[thick, blue, -] (2,1)--(2, 2);
             \draw[thick, blue, -] (3,1)--(3, 2);
       	    \gridpoints{2}{3};  
	    \node (l1) at (1,2.6) {$1$};
            \node (l2) at (2,2.6) {$2$};
            \node (l3) at (3,2.6) {$3$};   
	\end{tikzpicture}}  \\
\hline         
       P^3_1 &  \{1112,  1213, 2311\} &      \!\!\!\!\!  \!\!\!\!\! 
       \raisebox{0pt}[0pt][0pt]{
       \begin{tikzpicture}[scale=.5,baseline=5ex]
          \draw[thick, red, -] (1,1)--(2,1)--(3,2); 
          \draw[thick, red, -] (1,2)--(2,2)--(3,1); 
          \node (l1) at (.3,2) {$1$};
          \node (l2) at (.3,1) {$2$};
          \gridpoints{2}{3};
       \end{tikzpicture}}\\
     P^3_2 & \{2123, 2222, 1321\} \\
       \hline
     P^4_1&  \{1321, 2311\}\\
     P^4_2 &  \{1112, 2222\}\\
     P^4_3 &  \{1212,2123\} &         \!\!\!\!\!  \!\!\!\!\! 
         \raisebox{0pt}[0pt][0pt]{
         \begin{tikzpicture}[scale=.5,baseline=2ex]
             \draw[thick, blue, -] (1,1)--(2, 2);
             \draw[thick, blue, -] (1,2)--(2, 1);
             \draw[thick, blue, -] (3,1)--(3, 2);
       	    \gridpoints{2}{3};  
	    \node (l1) at (1,2.6) {$2$};
            \node (l2) at (2,2.6) {$3$};
            \node (l3) at (3,2.6) {$1$};   
	\end{tikzpicture}}  \\
     \end{array}$ }};
\end{tikzpicture}
    \caption{Running example of associated objects}
    \label{fig:runningexample}
\end{figure}
\end{example}

\begin{example}
\label{ex:latin_squares_correspondence}
In Figure~\ref{fig:latin_correspondence}, for the matched pair of orthogonal Latin squares from Example~\ref{ex:LatinSquare}, we present an associated automaton, grid array, and partition pair. 
Note that the automaton $\A$ in this example is isomorphic to the inverse of the automaton that generates the lamplighter group $\mathcal L_3=(\Z/3\Z)\wr\Z$ and studied in~\cite{bondarenko_dr:lamplighter}. 
Therefore, the group $\mathbb{G}(\A)$ generated by $\A$ is also isomorphic to $\mathcal L_3$.

\begin{figure}[htbp!]
    \centering
\begin{tikzpicture}
     \node[blocknf,yshift=4pt] (aut) at (0, 4.9) { \fbox{
     \begin{tikzpicture}[node distance=3cm, baseline={([yshift=-.8ex]current bounding box.center)}]
    \node[state] (q2) {$2$};
    \node[state, below left of=q2,xshift=13pt,yshift=-20pt] (q1) {$1$};
    \node[state, below right of=q2,xshift=-13pt,yshift=-20pt] (q3) {$3$};

    \draw (q2) edge[loop right] node[right,xshift=3pt,yshift=-2pt]{$1\,|\, 2$}  (q2);
    \draw (q2) edge [out=310,in=110]  node[right]{$2\,|\, 1$} (q3);
    \draw (q2) edge [out=250,in=50]  node[right,xshift=-6pt,yshift=-8pt]{$3\,|\, 3$} (q1);
    
    \draw (q1) edge [out=10,in=170] node[above,yshift=-4pt]{$3\,|\,2$} (q3);
    \draw (q1) edge [out=70,in=230] node[left,xshift=1pt]{$2\,|\,3$} (q2);
    \draw (q1) edge [loop left] node[below,xshift=2pt,yshift=-2pt]{$1\,|\, 1$}  (q1);

    \draw (q3) edge [out=190,in=350] node[below,yshift=1pt]{$2\,|\,2$} (q1);
    \draw (q3) edge [out=130,in=290] node[left,xshift=6pt,yshift=-8pt]{$3\,|\,1$} (q2);
    \draw (q3) edge [loop right] node[below,xshift=-3pt,yshift=-2pt]{$1\,|\, 3$}  (q1);

    \node[left=.9cm of q2] (A) {$\A$};
    
    \end{tikzpicture}
     }}; 
 %
     \node[blocknf] (mat) at (7, 6.2) {\fbox{\qquad\ \,$\begin{array}{ll}
    C_\A = & R_\A=\\
    \begin{bmatrix} 1&2&3\\2&3&1\\3&1&2\end{bmatrix} & 
    \begin{bmatrix} 1&3&2\\2&1&3\\3&2&1\end{bmatrix}
    \end{array} $\qquad\ \,}};
 %
     \node[blocknf, yshift=-4pt] (arr) at (0, 0) {\fbox{$ \begin{array}{l}
      \mathcal{Y}_\A=\\
      \left\{\!\!\!\begin{array}{c} (1,1,1,1), (1,2,2,3), (1,3,3,2),\\
    (2,1,2,2), (2,2,3,1), (2,3, 1,3),\\
    (3,1,3,3), (3,2,1,2), (3,3,2,1)\end{array} \!\!\!\right\}                            
    \end{array}$ }};
 %
    \node[blocknf, yshift=54pt] (par)  at (7,0) {\fbox{$      \begin{array}{lccc} 
       P^1_1 &\!\!  \{1111,  1223, 1332\} &    \!\!\!\!\! 
       \raisebox{-15pt}[0pt][0pt]{
       \begin{tikzpicture}[scale=.5,baseline=5ex]
          \draw[thick, red, -] (3,1)--(2,1)--(1,1); 
          \draw[thick, red, -] (3,2)--(2,2)--(1,2);
          \draw[thick, red, -] (3,3)--(2,3)--(1,3);
          \node (l1) at (.3,3) {$1$};
          \node (l2) at (.3,2) {$2$};
          \node (l3) at (.3,1) {$3$};
          \gridpoints{3}{3};
       \end{tikzpicture}}\\
        P^1_2 &\!\! \{2122, 2231, 2313\} \\
        P^1_3 &\!\! \{3133, 3212, 3321\} \\     
       \hline
       P^2_1& \!\! \{1111, 2122, 333\}\\
       P^2_2 & \!\! \{1223, 2231, 3212\}\\
       P^2_3 & \!\! \{1332,2313, 3133\} &      \!\!\!\!\!
         \raisebox{-10pt}[0pt][0pt]{
         \begin{tikzpicture}[scale=.5,baseline=2ex]
             \draw[thick, blue, -] (3,1)--(3, 2)--(3,3);
             \draw[thick, blue, -] (2,1)--(2, 2)--(2,3);
             \draw[thick, blue, -] (1,1)--(1, 2)--(1,3);
       	    \gridpoints{3}{3};  
	        \node (l1) at (1,3.55) {$1$};
            \node (l2) at (2,3.55) {$2$};
            \node (l3) at (3,3.55) {$3$};   
	\end{tikzpicture}}  \\[2mm]
\hline  
       P^3_1 &\!\!  \{1111,  3212, 2313\} &    \!\!\!\!\! 
       \raisebox{-15pt}[0pt][0pt]{
       \begin{tikzpicture}[scale=.5,baseline=5ex]
          \draw[thick, red, -] (3,1)--(2,3)--(1,2); 
          \draw[thick, red, -] (3,2)--(2,1)--(1,3);
          \draw[thick, red, -] (3,3)--(2,2)--(1,1);
          \node (l1) at (.3,3) {$1$};
          \node (l2) at (.3,2) {$2$};
          \node (l3) at (.3,1) {$3$};
          \gridpoints{3}{3};
       \end{tikzpicture}}\\
       P^3_2 &\!\! \{2122, 1223, 3321\} \\
       P^3_3 &\!\! \{3133, 2231, 1332\} \\     
       \hline
         P^4_1& \!\! \{1111, 2232, 3321\}\\
      P^4_2 & \!\! \{1332, 2122, 3212\}\\
      P^4_3 & \!\! \{1223,2313, 3133\} &      \!\!\!\!\!
         \raisebox{-10pt}[0pt][0pt]{
         \begin{tikzpicture}[scale=.5,baseline=2ex]
             \draw[thick, blue, -] (3,1)--(2, 2)--(1,3);
             \draw[thick, blue, -] (2,1)--(1, 2)--(3,3);
             \draw[thick, blue, -] (1,1)--(3, 2)--(2,3);
       	    \gridpoints{3}{3};  
	    \node (l1) at (1,3.55) {$1$};
            \node (l2) at (2,3.55) {$3$};
            \node (l3) at (3,3.55) {$2$};   
	\end{tikzpicture}}  \\[1mm]
     \end{array}$ }};
\end{tikzpicture}
    \caption{Objects associated with the pair of orthogonal Latin squares in Example~\ref{ex:LatinSquare}}
    \label{fig:latin_correspondence}
\end{figure}
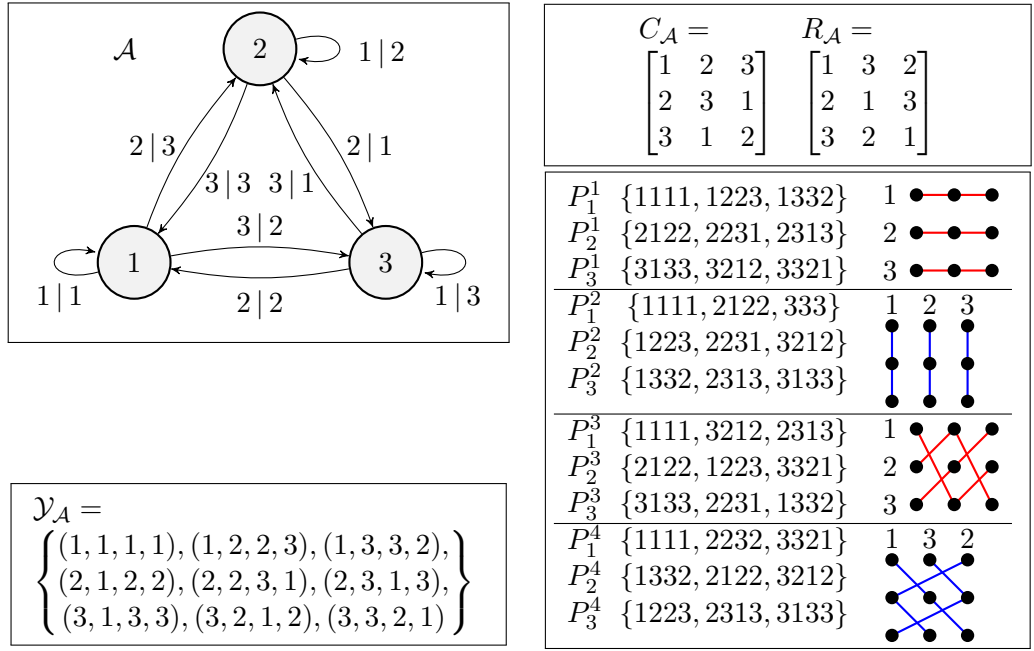
\end{example}

Examples~\ref{ex:runningexample1} and~\ref{ex:latin_squares_correspondence}
consist of associated Mealy automata/matched matrices/grid arrays/grid rete.
This allows the points of a rete to be drawn on a grid, with line segments defining a path  
among points in the same cell of the partition. 
The (positions of) entries in the matched matrices 
immediately give the (positions of) points in each line of $\mathcal{P}^3$ and $\mathcal{P}^4$.
To obtain this correspondence of positions, 
the points must be arranged so that 
  $\mathcal{P}^1$ consists of horizontal lines arranged in increasing order from top to bottom
and  
  $\mathcal{P}^2$ consists of vertical lines arranged in increasing order from left to right.
We refer to such an arrangement of points as the {\em natural arrangement}.   
In Section~\ref{sec:isotopism}, we shall meet transformations that disrupt the natural arrangement 
and then adjust accordingly to restore it.

\section{Matched Matrix representations of Mealy automata}
\label{sec:matchedmatrix-rep}

In Subsection~\ref{subsec:comb-encode-automata}, we saw a means to encode a Mealy automaton as a matched pair of matrices.  
To formalize this relationship, we describe isomorphic monoid structures.  
See, for example,~\cite{SmithRomanowsak:PMA} for an introduction to monoids.
Such a monoid structure on matrices was described in~\cite{Norton:rowlatin}.  
Automata have been viewed as pairs of monoids~\cite{akhavi_klmp:finiteness_problem}.

\subsection{Monoid of Matched Matrices}
\label{subsec:matchedmatrix}

%
Let $\mathsf{F}_d=\{f\colon[d]\rightarrow[d]\}$ be the monoid of functions from $[d]$ to $[d]$ with composition as the monoid operation, and 
let $\mathsf{F}_d^e$ denote the direct product of $\mathsf{F}_d$ with itself $e$ times. Naturally, elements of $\mathsf{F}_d$ can be represented by row or column vectors. 
The symmetric group $\Sym(d)$ is the group of units of $\mathsf{F}_d$ and its elements are represented by permutational vectors in which every element of $[d]$ appears exactly once. 
In the following proposition, we endow the set of matrices 
$\mathsf{M}_{m\times n}([m])$ (resp., $\mathsf{M}_{m\times n}([n])$) 
with a monoid structure corresponding to 
 column-wise (resp., row-wise) multiplication of matrix columns (resp., rows) 
 as elements of $\mathsf{F}_m$ (resp., $\mathsf{F}_n$).  
Below, we denote the corresponding operations by $\colmult$ and $\rowmult$, where the vertical and horizontal bars in the signs indicate column and row multiplication, respectively.


\begin{proposition}
\label{prop:monoidisom}
The following hold.
 \begin{enumerate}
\item \label{prop:monoid-C}
    $\mathsf{M}_{m\times n}([m])$ is a monoid with operation 
         $(C_1 \colmult  C_2) (i,j) = C_2(C_1(i,j),j)$ for all $C_1$, $C_2\in \mathsf{M}_{m\times n}([m])$ and 
         identity $\Idc$, 
         where $\Idc(i,j)=i$ $((i,j)\in[m]\times[n])$.

\item \label{prop:monoidisom-C}
    For $C\in \mathsf{M}_{m\times n}([m])$ and $j\in [n]$, let $C^{(j)}\colon[m]\rightarrow[m]$ be the map 
     given by $C^{(j)}(i) = C(i,j)$ for all $i\in[m]$. 
    Then the map $C\mapsto(C^{(1)}, C^{(2)},\ldots, C^{(n)})$ is a monoid isomorphism 
       $\mathsf{M}_{m\times n}([m])\rightarrow \mathsf{F}_m^n$.
          
\item \label{prop:monoid-R}
    $\mathsf{M}_{m\times n}([n])$ is a monoid with operation 
           $(R_1\rowmult R_2)(i,j) = R_1(i, R_2(i,j))$ for all  $R_1$, $R_2\in \mathsf{M}_{m\times n}([n])$
           and identity $\Idr$, 
           where $\Idr(i,j)=j$ $((i,j)\in[m]\times[n])$.

\item \label{prop:monoidisom-R}
    For $R\in \mathsf{M}_{m\times n}([n])$ and $i\in [m]$, let $R_{(i)}\colon[n]\rightarrow[n]$  be the map 
    given by $R_{(i)}(j)=R(i,j)$ for all $j\in[n]$.
    Then the map $R\mapsto(R_{(1)}, R_{(2)}, \ldots, R_{(m)})$ is a monoid isomorphism 
       $\mathsf{M}_{m\times n}([n])\rightarrow \mathsf{F}_n^m$.

\item $\mathsf{MM}(m,n)$ is a monoid with operation
$(C_1,R_1) \mmmult  (C_2,R_2)=(C_1 \colmult  C_2,  R_1\rowmult R_2)$ and identity $(\Idc,\Idr)$.

\item 
$(\mathsf{MM}(m,n), \mmmult)$ is isomorphic to $\mathsf{F}_m^n\times \mathsf{F}_n^m$.
\end{enumerate}
\end{proposition}

\begin{proof}
Let $C_1$, $C_2\in \mathsf{M}_{m\times n}([m])$.
By construction $C_1 \colmult  C_2\in \mathsf{M}_{m\times n}([m])$.
Observe that 
$((C_1 \colmult  C_2)\colmult  C_3) (i,j) 
    = C_3((C_1 \colmult  C_2)(i,j),j) = C_3(C_2(C_1(i,j),j),j)
    =(C_2\colmult  C_3))(C_1(i,j),j)
    =(C_1 \colmult  (C_2\colmult  C_3)) (i,j)$.
Thus, $\colmult $ is associative. 

Compute  $(C_1\colmult \Idc )(i,j)  = \Idc(C_1(i,j),j) = C_1(i,j)$ and
$(\Idc \colmult C_1)(i,j)= C_1(\Idc(i,j),j) = C_1(i,j)$. 
Thus, $\Idc$ is the identity of $\colmult $.
Hence, $(\mathsf{M}_{m\times n}([m]),\, \colmult )$ is a monoid.
Similarly, $(\mathsf{M}_{m\times n}([n]), \,\rowmult )$ is a monoid with identity
$\Idr$.

The map in (ii) is a bijection, as it maps $(f_1, \ldots f_n)$ to $C(i,j)=f_j(i)$.
Since $(C_1\colmult C_2)(i,j) = C_2(C_2(i,j))$, $(C_1\colmult C_2)^{(j)}$ maps $i$ to $C_2(C_1(i,j),j)$.
In addition, $C_1^{(j)}\circ C_2^{(j)}$ maps $i$ to $C_2^{(j)}(C_1^{(j)}(i)) = C_2(C_1(i,j),j)$. 
Thus, (ii) holds.  A similar argument gives (iv).   
Parts (v) and (vi) are consequences of the direct product construction.
\end{proof}


Write $M^\transp$ to denote the transpose of any matrix. In Proposition~\ref{prop:dagger-dual-correspond}  we shall see that transposition plays a role in relating the matched matrices associated with a Mealy  automaton and its dual automaton.  As such, we note the following.

\begin{corollary}
For all  $C_1$, $C_2\in\mathsf{M}_{m\times n}([m])$  and 
    all  $R_1$, $R_2\in \mathsf{M}_{m\times n}([n])$,
    \[ 
    \begin{array}{rclrcl}
    (C_1\colmult C_2)^\transp &=& C_2^\transp\rowmult C_1^\transp, &  
    (R_1\rowmult R_2)^\transp & =& R_2^\transp \colmult R_1^\transp.
    \end{array}
    \]
\end{corollary}

\begin{proof}
Compute as follows:
\[\begin{array}{r}
(C_1\colmult C_2)^\transp(j,i) 
   = (C_1\colmult C_2)(i,j)
   = C_2(C_1(i,j),j) 
   = C_2^\transp(j,C_1^\transp(j,i)) 
   = (C_2^\transp\rowmult C_1^\transp)(j,i),\\
(R_1\rowmult R_2)^\transp(j,i) 
   = (R_1\rowmult R_2)(i,j)
   = R_1(i,R_2(i,j))
   = R_1^\transp(R_2^\transp(j,i),i)  
   = (R_2^\transp\colmult R_1^\transp)(j,i).\\
\end{array}
\]
Thus, the equations  are verified entry-by-entry.
\end{proof}

Below (Theorem~\ref{thm:array-matrix-automaton-cond}), we will add the conditions,   individually and together,  on a matched pair of matrices that $C$, $R$, and $(C,R)$ are invertible elements in the appropriate monoid.
Ultimately (Theorem~\ref{thm:automatafamily-birevconditions}), we shall show that bireversible automata are in bijective correspondence  with cooperative pairs (all three of $C$, $R$, and $(C,R)$ are invertible).

\begin{lemma}
\label{lem:col-row-Latin}
The following hold.
\begin{enumerate}
\item \label{lem:col-Latin} 
         Let $C\in  \mathsf{M}_{m\times n}([m])$.  The following are equivalent.
         \begin{enumerate}
         \item Every element of $[m]$ appears exactly once in each column of $C$.
         \item Each column of $C$ is a permutation of $[m]$.
         \item $C$ has an inverse in the monoid $(\mathsf{M}_{m\times n}([m]),\, \colmult)$. 
         \end{enumerate}
         Whenever these equivalent conditions hold, we say that $C$ is \emph{column-Latin}.
         In this case, the $j$-th column of $C$ is the transpose of the one row form of the permutation $C^{(j)}$.  We write $C^\ic$ to denote the inverse of $C$ in this monoid.
\item  \label{lem:row-Latin} 
       \cite{Norton:rowlatin}
         Let $R\in  \mathsf{M}_{m\times n}([n])$. The following are equivalent.
         \begin{enumerate}
         \item Every element of $[n]$ appears exactly once in each row of $R$.
         \item Each row of $R$ is a permutation of $[n]$.
         \item $R$ has an inverse in the monoid $(\mathsf{M}_{m\times n}([n]),\, \rowmult)$. 
         \end{enumerate}
         Whenever these equivalent conditions hold, we say that $R$ is \emph{row-Latin}.
         In this case, the $i$-th row of $R$ is the one-row form of the permutation $R_{(i)}$.
         We write $R^\ir$ to denote the inverse of $R$ in this monoid.
\item \label{lem:orthogonal}
      Let $(C,R)\in  \mathsf{MM}(m,n)$.  The following are equivalent.
         \begin{enumerate}
         \item Every element of $[m]\times[n]$ appears exactly once among the pairs $(C(i,j), R(i,j))$ as 
        $(i,j)$ runs over $[m]\times[n]$.
         \item The map $(m,n)\mapsto(C(m,n), R(m,n))$ is a permutation of $[m]\times[n]$.
         \item  $(C,R)$ has an inverse in the monoid $(\mathsf{MM}(m,n), \mmmult)$.         
         \end{enumerate}
         Whenever these equivalent conditions hold, we say that $C$ and $R$ are \emph{orthogonal}. 
\end{enumerate}
\end{lemma}

\begin{proof}
Since $C$ and $R$ have $m$ rows and $n$ columns, conditions (a) and (b) are equivalent in (i) and (ii), 
where we view permutation in the sense of a rearrangement of a set.  
However, a permutation in this sense induces a permutation in the function sense via the construction of Proposition~\ref{prop:monoidisom}: 
The column/row is the one-row form of the permutation. 
This gives the equivalence of (b) and (c) in (i) and (ii).    
Since $(C,R)\subset[m]\times[n]$, equality holds if and only if every element of $[m]\times[n]$ appears exactly once in 
$\{(C(m,n), R(m,n)) \colon (m,n)\in[m]\times[n] \}$.  
Thus, (a) and (b) are equivalent in (iii).   
The equivalence of (b) and (c) in (iii) follows since permutations as rearrangements and functions coincide and permutations are exactly invertible functions.
\end{proof}

Conditions (ia), (iia), and (iiia) in Lemma~\ref{lem:col-row-Latin} are 
usually taken as the definitions of row-Latin, column-Latin, and orthogonal,
as in Section~\ref{sec:introduction}.

\begin{example}
\label{ex:CL-RL-O-independent}
The column-Latin, row-Latin, and orthogonal properties are independent. 

\begin{tabular}{|c|c|c|c|c||c|c|c|c|c|}  
\hline
    $C$ & $R$& CL & RL & O & $C$ & $R$& CL & RL & O \\
    \hline
    $\begin{bmatrix} 1&2&1\\2&1&2 \end{bmatrix}$& $\begin{bmatrix}1&2&3\\3&2&1 \end{bmatrix}$ & Y & Y & Y &
    $\begin{bmatrix} 1&1&2\\1&2&2 \end{bmatrix}$ & $\begin{bmatrix} 1&3&3\\1&2&2 \end{bmatrix}$ & N & N &N\\
    \hline
    $\begin{bmatrix}1&2&2\\1&2&1 \end{bmatrix}$ & $\begin{bmatrix} 1&2&3\\2&1&3 \end{bmatrix}$ & N &Y &Y&
    $\begin{bmatrix}1&1&2\\2&2&1 \end{bmatrix}$ & $\begin{bmatrix} 1&1&3\\2&3&2 \end{bmatrix}$ & Y &N &N\\
    \hline
    $\begin{bmatrix}1&1&2\\2&2&1 \end{bmatrix}$ & $\begin{bmatrix} 1&2&1\\2&3&3 \end{bmatrix}$ & Y &N &Y &
    $\begin{bmatrix}1&1&2\\2&1&2 \end{bmatrix}$ & $\begin{bmatrix} 1&2&3\\3&2&1 \end{bmatrix}$ & N &Y &N\\
    \hline
    $\begin{bmatrix}1&1&2\\2&2&1 \end{bmatrix}$ & $\begin{bmatrix}1&2&3\\2&3&1 \end{bmatrix}$ & Y & Y& N &
    $\begin{bmatrix} 1&1&2\\2&1&2 \end{bmatrix}$ & $\begin{bmatrix} 1&3&1\\2&2&3 \end{bmatrix}$ & N & N &Y\\
    \hline
\end{tabular}
\end{example}

Let  $\mathsf{CL}(m,n)$ and $\mathsf{RL}(m,n)$ respectively denote the sets of column- and row-Latin matrices of order $m\times n$.   
Observe that 
    $\mathsf{RL}(m,n)=\{C^\transp:C\in\mathsf{CL}(n,m)\}$ and  
    $\mathsf{CL}(m,n)=\{R^\transp:R\in\mathsf{RL}(n,m)\}$.
Let $\Sym(d)^e$ denote the direct product of $\Sym(d)$ with itself $e$ times.

\begin{corollary}
\label{cor:CL-CR-monoids}
The following hold.
\begin{enumerate}
\item  $\mathsf{CL}(m,n)$ is the group of units of the monoid $(\mathsf{M}_{m\times n}([m]), \colmult)$.
The monoid isomorphism of Proposition~\ref{prop:monoidisom-C} maps $\mathsf{CL}(m,n)$ onto the group of units $\Sym(m)^n$  of $\mathsf{F}_m^n$.
\item $\mathsf{RL}(m,n)$ is the group of units of the monoid $(\mathsf{M}_{m\times n}([n]), \rowmult)$.
The monoid isomorphism of Proposition~\ref{prop:monoidisom-R} maps $\mathsf{RL}(m,n)$ onto is the group of units $\Sym(n)^m$ of $\mathsf{F}_n^m$.
\item  $\mathsf{CL}(m,n)\times \mathsf{RL}(m,n)$  is the group of units of the monoid $(\mathsf{MM}(m,n),\mmmult)$ described in Lemma~\ref{lem:orthogonal}.  
The monoid isomorphism formed by combining the above maps  $\mathsf{CL}(m,n)\times \mathsf{RL}(m,n)$  onto the group of units $\Sym(m)^n\times \Sym(n)^m$ of $\mathsf{F}_m^n\times \mathsf{F}_n^m$.
\end{enumerate}
\end{corollary}

\begin{proof}
By Lemma~\ref{lem:col-row-Latin},   
$\mathsf{CL}(m,n)$ is precisely the set of invertible elements of the monoid $(\mathsf{M}_{m\times n}([m]),\, {\colmult })$. 
The set of invertible elements of $\mathsf{F}_m^n$ is precisely $\Sym(m)^n$. 
Taking only invertible elements gives a subgroup.  
Thus, (i) follows from Proposition~\ref{prop:monoidisom-C}.  
A similar argument gives (ii). Combining (i) and (ii) via the direct product  gives (iii).
\end{proof}

The preceding result is very similar to those appearing in~\cite{Norton:rowlatin} concerning row-Latin squares, which is noted in 
\cite[pp.~89--91]{Denes+Keedwell-LSa_SecondEdition_2015}.  
Row-Latin squares are also discussed in~\cite[Subsections 6.3--6.5]{Laywin+Mullen:discretemath}.  

We note that by Corollary~\ref{cor:CL-CR-monoids} there is a direct relationship between this terminology and the one introduced in~\cite{akhavi_klmp:finiteness_problem}, 
where Mealy automata are viewed as elements of $\mathsf{F}_m^n\times \mathsf{F}_n^m$, 
      invertible automata as elements of $\mathsf{F}_m^n\times \Sym(n)^m$, and 
      reversible automata as elements of $\Sym(m)^n\times \mathsf{F}_n^m$ 
(with slightly different notation).

\subsection{Matrix Action}
\label{subsec:matrixactions}

The matrix monoids $\mathsf{M}_{m\times n}([m])$ and $\mathsf{M}_{m\times n}([n])$ discussed in Subsection~\ref{subsec:matchedmatrix} act on each other.  
 We use the notation $\ltacti$ and $\rtacti$ for the actions, where the arrow points to the matrix being acted upon.
 
 \begin{proposition}
 \label{prop:monoidactions}
The monoids $(\mathsf{M}_{m\times n}([m]),\, {\colmult })$ and $(\mathsf{M}_{m\times n}([n]),\, \rowmult)$
act on each other as follows.  Let  $C\in\mathsf{M}_{m\times n}([m])$ and $R\in \mathsf{M}_{m\times n}([n])$.
\begin{enumerate}
\item  $C$ acts on $R$ by 
          \begin{equation}
                (R\rtacti C)(i,j) = R(C(i,j),j)
                    \qquad (i,j)\in[m]\times[n].
           \label{eq:CactonR}
           \end{equation}    
\item $R$ acts on $C$ by 
         \begin{equation}
               (R\ltacti C)(i,j)  = C(i,R(i,j))
                    \qquad (i,j)\in[m]\times[n].
         \label{eq:RactonC}      
         \end{equation}
\end{enumerate}
 \end{proposition}

\begin{proof}
By construction $(R\rtacti C)\in \mathsf{M}_{m\times n}([n])$.
Verify that $R  \rtacti \Idc = R$ by computing
$(R\rtacti \Idc)(i,j) = R(\Idc(i,j),j)=R(i,j)$.
For $C_1$, $C_2\in \mathsf{M}_{m\times n}([m])$, compute
$(R \rtacti (C_1 \colmult  C_2))(i,j)
 = R((C_1 \colmult  C_2)(i,j),j)
 = R(C_2(C_1(i,j),j),j)$
and 
$((R\rtacti C_2)\rtacti C_1)(i,j)
  =(R\rtacti C_2)(C_1(i,j),j)
  = R(C_2(C_1(i,j),j),j)$.
Thus, $\rtacti$ defines a right action.  

By construction, $(R\ltacti C)\in \mathsf{M}_{m\times n}([m])$.
Verify that $\Idr\ltacti C = C$ by computing
$(\Idr \ltacti C)(i, j) = C(i,  \Idr(i, j)) = C(i,j)$.
For $R_1$, $R_2\in \mathsf{M}_{m\times n}([n])$, compute
$((R_1 \rowmult  R_2)\ltacti C)(i,j)
 = C(i, (R_1 \rowmult  R_2)(i,j))
 = C(i, R_1(i, R_2(i,j)))$ and
$(R_2\ltacti ( R_1\ltacti C))(i,j)
   =(R_1\ltacti C)(i, R_2(i,j))
   = C(i, R_1(i,R_2(i,j)))$.
Thus, $\ltacti$ defines a left action.   
\end{proof}

The actions in Proposition~\ref{prop:monoidactions} are consistent with the respective actions of $(\mathsf{M}_{m\times n}([m]),\, {\colmult })$ and $(\mathsf{M}_{m\times n}([n]),\, \rowmult)$ on themselves by left and right multiplication in Proposition~\ref{prop:monoidisom}.  
To discuss the special families of automata, we consider
$\mathsf{RL}(m,n)$ acting on $\mathsf{M}_{m\times n}([m])$ by inverses and  $\mathsf{CL}(m,n)$ acting on $\mathsf{M}_{m\times n}([n])$ by inverses. Doing so turns a left action into a right action, and vise versa.

\begin{lemma}
\label{lem:CLRLactions}
With reference to Proposition~\ref{prop:monoidactions}, the following hold.
\begin{enumerate}
\item Define a left action of $\mathsf{CL}(m,n)$ on $\mathsf{M}_{m\times n}([n])$ by
      $C\ltact R = R \rtacti C^\ic$:
   \[ (C\ltact R)(C(i,j),j ) = R(i,j)\qquad (i,j)\in[m]\times[n].\]
      Then column $j$ of $C\ltact R$ is formed by permuting column $j$ of $R$ with the permutation $C^{(j)}$ in column $j$ of $C$. 
\item Define a right action of $\mathsf{RL}(m,n)$ on $\mathsf{M}_{m\times n}([m])$ by 
      $C\rtact R = R^\ir \ltacti C$:
  \[ (C\rtact R)(i, R(i,j))=C(i,j)\qquad(i,j)\in[m]\times[n].\]
      Then row $i$ of $C\rtact R$ is formed by permuting row $i$ of $C$ with the permutation $R_{(i)}$ in row $i$ of $R$.      
\end{enumerate}
\end{lemma}

\begin{proof}
Inverting the acting elements and reversing left and right actions gives another action.    
The entry formula in (i) holds since $(C\ltact R) = (R \rtacti C^\ic)(i,j) = R( C^\ic(i,j),j)$.
Since $C$ and $C^\ic$ are the appropriate inverses, substituting $C(i,j)$ in place of $i$ gives the formula.  It is now clear that the action permutes each column of $R$ by the corresponding column of $C$. A similar argument completes (ii). 
\end{proof}

\begin{example}
Compute $C\rtact R=  R^\ir \ltacti C$ entry-wise for\\
   $C=\begin{bmatrix} 1&2&3&4&5\\4&2&5&1&3\end{bmatrix}$,
   $R=\begin{bmatrix} 2&3&4&5&1\\5&1&3&2&4\end{bmatrix}$, where
   $R^\ir= \begin{bmatrix} 5&1&2&3&4\\2&4&3&5&1\end{bmatrix}$, by

$(C \rtact R) (i,j) = C(i,R^\ir(i,j)) = 
\begin{array}{l|ccccc}  
i\backslash j &1 & 2 &3 & 4 &5 \\
\hline
1 & 5 & 1& 2&3&4\\
2 & 2 & 1& 5&3 &4
\end{array}$. \newline
Thus,  $C\rtact R$ is formed by 
permuting each row of $C$ by the corresponding row of $R$.
\end{example}

\begin{corollary}
\label{cor:actions-transpose}
With reference to Proposition~\ref{prop:monoidactions} and Lemma~\ref{lem:CLRLactions}, the following hold.
\begin{enumerate}
\item For all 
$C\in\mathsf{M}_{m\times n}([m])$  and $R\in \mathsf{M}_{m\times n}([m])$,
\[ (C\rtacti R)^\transp =R^\transp \ltacti C^\transp, \qquad
   (R\ltacti C)^\transp = C^\transp \rtacti R^\transp.\]  
\item For all 
$C\in\mathsf{CL}(m,n)$  and $R\in \mathsf{RL}(m, n)$,   
\[ (C\ltact R)^\transp =C^\transp \rtact R^\transp, \qquad
   (R\rtact C)^\transp = C^\transp \ltact R^\transp.\]     
\end{enumerate}
\end{corollary}

\begin{proof}   
For (i), compute as follows:
\[\begin{array}{r}
(C\rtacti R)^\transp(j,i) 
    = (C\rtacti R)(i,j)
    = R(C(i,j),j)
    = R^\transp(j, C^\transp(j,i))
    = (C^\transp \ltacti R^\transp)(j,i),\\
(R\ltacti C)^\transp(j,i) 
    = (R\ltacti C)(i,j)
    = C(i,R(i,j))
    = C^\transp( R^\transp(j,i),i)
    = (C^\transp \rtacti R^\transp)(j,i).
\end{array}
\]
Thus, both equations of (i) are verified entry-by-entry.
Note that $(C^\ic)^T=(C^T)^\ir$ and  $(R^\ir)^T=(R^T)^\ic$. 
Now, (ii) follows from (i).
\end{proof}

The actions $\ltact$ and $\rtact$ provide a connection between automata with each of the invertible, reversible, IR, and bireversible properties and  matched matrices with the properties of being row-Latin,  column-Latin, and orthogonal.  

\begin{lemma}
\label{lem:invertible-action-matrices}
 Let $\A$ be an invertible Mealy $(m,n)$-automaton. 
 \begin{enumerate}
     \item $R_\A$ is invertible in $(\mathsf{M}_{m\times n}([n]), \rowmult)$, and $R_{\iz\A} = R_\A^{\ir}$.
     \item $C_{\iz\A}=C_\A\rtact R_\A$.
 \end{enumerate}
\end{lemma}

\begin{proof}
According to the inversion formula~\eqref{eqn:inv} the matrix $C_{\iz\A}$ is obtained from $C_\A$ by applying permutations from the rows of $R_\A$ to the corresponding rows of $C_\A$. 
In other words,
   \[C_{\iz\A}=C_\A\rtact R_\A.\]
Also, by Equation~\eqref{eqn:inv} the matrix $R_{\iz\A}$ is obtained from $R_\A$ by replacing each row $r$ corresponding to permutation $\sigma_r$ by the row corresponding to the permutation $\sigma_r^{-1}$.
\end{proof}

We now add the conditions of Lemma~\ref{lem:col-row-Latin} to those of Theorems~\ref{thm:U(st)-Y<st>} and~\ref{thm:classes-U(st)} in the case of Mealy automata. 

\begin{theorem}
\label{thm:array-matrix-automaton-cond}
Let $\A$ be a Mealy $(m,n)$-automaton, and let $(C,R)$ be an associated matched pair of matrices of order $m\times n$.
Then the following hold.                                                                                                   
\begin{enumerate}

\item \label{thm:array-RL-matrix-automaton-cond}
$\A$ is reversible if and only if $C$ is column-Latin.

\item \label{thm:array-CL-matrix-automaton-cond}
$\A$ is invertible if and only if $R$ is row-Latin.

\item \label{thm:array-O-matrix-automaton-cond}
$\A$ is coreversible if and only if $C$ and $R$ are orthogonal.

\item 
$\A$ is an IR-automaton if and only if $C$ is column-Latin and $R$ is row-Latin.
\end{enumerate}
\end{theorem}

\begin{proof}
Part (i)  follows from Lemma~\ref{lem:col-Latin}
since $\A$ is reversible when for all $x\in X$, $\pi(q,x)$ induces a permutation of $Q$.
Part (ii) follows from Lemma~\ref{lem:row-Latin}
since $\A$ is invertible when for all $q\in Q$, $\lambda(q,x)$ induces a permutation of $X$.
Part (iii) follows from Lemma~\ref{lem:orthogonal} in light of Theorem~\ref{thm:U(st)-Y<st>}.
Part (iv) follows from  (i) and (ii).
\end{proof}

\section{Bireversible automata}
\label{sec:array-rep}

In Section~\ref{sec:comb-encode-automata}, we encoded Mealy automata as grid arrays and as grid rete, and in Section~\ref{sec:matchedmatrix-rep} we encoded them as matched matrices. 
We now specialize these combinatorial structures to those which correspond to bireversible automata.

\subsection{From grid arrays to semi-orthogonal arrays}
\label{subsec:birev-array}

In Theorem~\ref{thm:array-matrix-automaton-cond}, we considered grid arrays with additional properties $\un{23}$, $\un{34}$, and $\un{41}$ separately. We now examine what happens when they all hold simultaneously.
Ultimately (Theorem~\ref{thm:automatafamily-birevconditions}), we shall show that this case is associated with bireversible automata.

\begin{definition}
\label{def:NSOA}
By a \emph{narrow semi-orthogonal array} of order $m\times n$, 
we mean an interleaved array $\mathcal{Y}\in\mathsf{IA}(m,n)$ 
such that for all $(s,t)\in\{(1,2), (2,3), (3,4), (4,1)\}$,
$\mathcal{Y}$ has property $\ut{st}$:
For distinct $y$, $y'\in \mathcal{Y}$,    $(y_s, y_t)$ and $(y_s', y_t')$ are distinct.
\end{definition}

\begin{example}
\label{ex:runningexample-verifyNSOA}
In Example~\ref{ex:runningexample1}, we verify that $\mathcal{Y}_\A$ is a narrow semi-orthogonal array by showing that every possible value appears in the following four pairs of components:
\[
\begin{array}{|l|l|}
\hline
(1,2)& \{ (\circled{1},\circled{1},1,2), (\circled{1},\circled{2}, 1,3), (\circled{1},\circled{3}, 2,1), 
               (\circled{2},\circled{1},2,3), (\circled{2},\circled{2},2,2), (\circled{2},\circled{3},1,1)\}\\
\hline
(2,3)& \{ (1,\circled{1},\circled{1},2), (1,\circled{2}, \circled{1},3), (1,\circled{3}, \circled{2},1), 
               (2,\circled{1},\circled{2},3), (2,\circled{2},\circled{2},2), (2,\circled{3},\circled{1},1)\}\\
\hline
(3,4)& \{ (1,1,\circled{1},\circled{2}), (1,2, \circled{1},\circled{3}), (1,3, \circled{2},\circled{1}), 
               (2,1,\circled{2},\circled{3}), (2,2,\circled{2},\circled{2}), (2,3,\circled{1},\circled{1})\}\\
\hline
(4,1)& \{ (\circled{1},1,1,\circled{2}), (\circled{1},2, 1,\circled{3}), (\circled{1},3, 2,\circled{1}), 
               (\circled{2},1,2,\circled{3}), (\circled{2},2,2,\circled{2}), (\circled{2},3,1,\circled{1})\}\\
\hline
\end{array}
\]
\end{example}

The construction of an interleaved array introduced above can be viewed as an analog of orthogonal arrays. 
For more on orthogonal arrays see~\cite{HedayatSloneSufken:OA,LinStufken:ortharray}.

\begin{definition}
\label{def:OA}
An \emph{orthogonal array of type $(M,n,s,t)$} 
      (\emph{order} $m\times n$, 
       $s$ \emph{levels}, 
       \emph{strength} $t$ and 
       \emph{index} $\lambda$)
is an $M\times n$ array with entries from a set $S$ of size $s$ such that every $M\times t$  subarray 
contains each $t$-tuple of elements of $S$ exactly $\lambda$-many times as a row.  
The strength $t$ is generally taken to be as large as possible with the given parameters.  
The index is determined by $\lambda=M/s^t$. 
\end{definition}

An orthogonal array of type $(m^2, 4, m, 2)$ 
   is an element of $\mathsf{IA}(m,m)$ with property $\ut{st}$ for \emph{all} distinct $s$, $t\in[4]$.
In particular, such an orthogonal array is a narrow semi-orthogonal array.   
Indeed, narrow semi-orthogonal arrays need only have property $\ut{st}$ for  $s$, $t\in[4]$ with different parity, since
Definition~\ref{def:NSOA} does not require $\un{13}$ or $\un{24}$ (see Example~\ref{ex:repeting} below). 
Hence, we use the modifier ``semi''.   
We use the adjective ``narrow'' to signify that there are only four columns.  
Narrow semi-orthogonal arrays also resemble 
         mixed orthogonal arrays of type  $(mn, m^2n^2, 2)$~\cite{HedayatSloneSufken:OA}.  
In~\cite{curtin2026generalizationsnetslatinsquares} the first author considers a mild generalization consisting of subsets of $[m]^k\times [n]^\ell$ 
where pairs of entries taken from different blocks uniquely determine the element, a case dubbed ``svelte''.  
One might consider generalizations with uniform repetition of pairs and/or more than two types of blocks.  
However, here we opt to focus on the task at hand rather than consider a candidate that more fully generalizes orthogonal arrays. 

The columns of an orthogonal array with $s$ levels and strength $t\geq 2$ are orthogonal in the sense that 
their dot product is zero when the set $S\subset\mathbb{R}$ is taken so that $-u\in S$ for all $u\in S$. 
However, it is customary to take $S=[s]$.

The bijection between matched matrices and grid arrays parallels the following relationship between orthogonal arrays and Latin squares (Example~\ref{ex:latin_squares_correspondence}).

\begin{proposition}
\label{prop:MOLS-OA}
\cite{Denes+Keedwell-LSa_SecondEdition_2015,HedayatSloneSufken:OA, Laywin+Mullen:discretemath,vanLintWilson:coursecomb}.
If $L_1,L_2,\ldots,L_k$ are mutually orthogonal arrays of order $n$ (possibly $k=1$), 
then $\{(i,j, L_1(i,j), \ldots, L_k(i,j))\,|\, 1\leq i,j\leq n\}$ comprise the rows of an orthogonal array of type $(n^2, k+2,n,2)$ and index 1.
Conversely, any orthogonal array with these parameters encodes a set of $k$-many mutually orthogonal Latin squares of order $n$. 
\end{proposition}

\subsection{From rete to reticulations}
\label{subsec:reticulations}

A rete (particularly of a narrow semi-orthogonal array) can be viewed as a geometric object. 
The points are the four-tuples and the cells of each partition comprise a family of parallel (or empty) lines.
This geometric perspective will lead us to another connection to the theory of Latin squares.
For our purposes, a point-line incidence structure consists of a set of points and a set of lines (subsets of points). A point and a line are \emph{incident} when the point is an element of the line.
Two lines are \emph{parallel} when they are equal or have an empty intersection. By a \emph{family of lines} we mean a set of lines.  

 \begin{definition}\label{def:psuedoaffineplane}
By a \emph{narrow reticulation}, we mean a set of points and two multisets (``types'') of two families of lines satisfying the following axioms. 
\begin{description}
\item[(R-1)] Two lines of different types meet in exactly one point. 
\item[(R-2)] Each family of lines partitions the set of points.
\end{description}
We say that a narrow reticulation is \emph{ordered} when the types are ordered, each pair of types is ordered, and all families of lines are ordered. We shall refer to the two types of lines as \emph{weft} and \emph{warp}.
 \end{definition}

Here we focus on the case relevant to bireversible automata, where there are two families of lines of each type,
using the adjective ``narrow'' to signify that this is a special case.  
In~\cite{curtin2026generalizationsnetslatinsquares}, the first author examines general reticulations with more families of lines of each type and adapts a few additional ideas from the theory of Latin squares. 
We note that a three-web consist of three (infinite) families of lines of different types that satisfy Axioms (R-1) and (R-2)~\cite{Baer:NG,BlaschkeBol:GdGTFdD, Reidemeister:WG}.  

\begin{example}
\label{ex:multiset-identity}
The line families are kept in multisets to allow for automata in which  
the transition function or the output function is an identity map,
as in Figure~\ref{fig:identityautomaton}.
Repetition of lines in distinct families of the same type is  intrinsic to  a rete (Example~\ref{ex:runningexample1}).
\end{example}

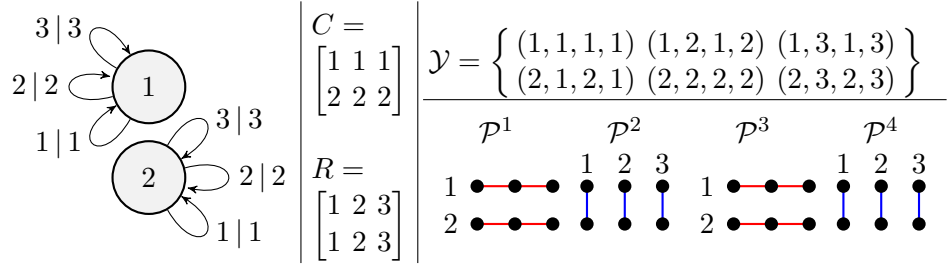
\begin{figure}
\[\!\!\!
\setlength{\arraycolsep}{2pt}
\begin{array}{c|c|cc}
    \begin{tikzpicture}[node distance=1.2cm, baseline={([yshift=-.8ex]current bounding box.center)}]
    \node[state] (q1) {$1$};
    \node[state, below of=q1] (q2) {$2$};
    \draw (q1) edge[loop left] node{$2\,|\, 2$} (q1);
    \draw (q1) edge[loop, , out=150, in=120, looseness=8] node[left]  {$3\,|\, 3$}  (q1);
    \draw (q1) edge[loop, out=240, in=210, looseness=8] node[left]  {$1\,|\, 1$}  (q1);
    \draw (q2) edge[loop, out=-60, in=-30, looseness=8] node[right] {$1\,|\, 1$}(q2);
    \draw (q2) edge[loop right] node{$2\,|\, 2$}  (q2);
    \draw (q2) edge[loop, out=60, in=30, looseness=8] node[right] {$3\,|\, 3$}  (q2);
    \end{tikzpicture}
&
    \begin{array}{l} C=  \\
    			   \begin{bmatrix} 1&1&1\\2&2&2\end{bmatrix} \\  \\
			   R=  \\
			   \begin{bmatrix} 1&2&3\\1&2&3\end{bmatrix}
    \end{array}
 &
  \begin{array}{l}
      \setlength{\arraycolsep}{2pt}
       \mathcal{Y}= 
                    \left\{ \begin{array}{ccc}(1,1,1,1) & (1,2,1,2) &  (1,3,1,3)\\ 
                                                               (2,1,2,1) & (2,2,2,2) & (2,3,2,3) \end{array}\right\}  \\  
\hline \\[-1.6ex]
 	\begin{array}{cccc}
	{\mathcal{P}^1}&{\mathcal{P}^2}&{\mathcal{P}^3}&{\mathcal{P}^4}\\[-3pt]
 			\begin{tikzpicture}[scale=.5,baseline=4.5ex]
	 			\foreach \y in {1,2}
	         		\draw[thick, red, -] (1,\y)--(3,\y);
	  			\gridpoints{2}{3};  
				\node (l1) at (.3,2) {$1$};
         			 \node (l2) at (.3,1) {$2$};
			\end{tikzpicture}&
            \begin{tikzpicture}[scale=.5,baseline=4.5ex]
      				\draw[thick, blue, -] (1,1)--(1, 2);
      				\draw[thick, blue, -] (2,1)--(2, 2);
      				\draw[thick, blue, -] (3,1)--(3, 2);
 	 			\gridpoints{2}{3};  
				\node (l1) at (1,2.6) {$1$};
           			\node (l2) at (2,2.6) {$2$};
           			\node (l3) at (3,2.6) {$3$};  
			\end{tikzpicture}  &
			\begin{tikzpicture}[scale=.5,baseline=4.5ex]
	 			\foreach \y in {1,2}
	         		\draw[thick, red, -] (1,\y)--(3,\y);
	  			\gridpoints{2}{3};  
				\node (l1) at (.3,2) {$1$};
         			 \node (l2) at (.3,1) {$2$};
			\end{tikzpicture}    &
			\begin{tikzpicture}[scale=.5,baseline=4.5ex]
      				\draw[thick, blue, -] (1,1)--(1, 2);
      				\draw[thick, blue, -] (2,1)--(2, 2);
      				\draw[thick, blue, -] (3,1)--(3, 2);
 	 			\gridpoints{2}{3};  
				\node (l1) at (1,2.6) {$1$};
           			\node (l2) at (2,2.6) {$2$};
           			\node (l3) at (3,2.6) {$3$};  
			\end{tikzpicture}        
	\end{array}
   \end{array}                
\end{array}
\]
    \caption{An identity automaton and associated encodings}
    \label{fig:identityautomaton}
\end{figure}

\begin{lemma}
\label{lem:CPgeom}
The rete of a narrow semi-orthogonal array is an ordered reticulation.
\end{lemma}

\begin{proof} 
By the construction of the rete 
   $({\mathcal{P}^1},{\mathcal{P}^2}, {\mathcal{P}^3}, {\mathcal{P}^4})$ 
of an interleaved array $\mathcal{Y}$, 
every point is in exactly one line of each family according to its entries 
(since $(i,j,k,\ell)$ belongs to $P^1_i$, $P^2_j$, $P^3_k$, and $P^4_\ell$).  
It follows that these four sets which make up a rete are disjoint, that is, (R-2) holds.

We take the first type of line to be those of ${\mathcal{P}^1}$ and ${\mathcal{P}^3}$, and 
we take the second type of line to be those of ${\mathcal{P}^2}$ and ${\mathcal{P}^4}$.
When $\mathcal{Y}$ is a narrow semi-orthogonal array, 
conditions $\un{st}$ hold for all $s$, $t\in[4]$ with different parity. 
Arguing as in the proof of Lemma~\ref{lem:GA-EAorth} gives that two lines of 
different types meet in exactly one point, that is, (R-1) holds.
\end{proof}

Note that the line families $({\mathcal{P}^1},{\mathcal{P}^2}, {\mathcal{P}^3}, {\mathcal{P}^4})$ of a rete are grouped into two types, which we dub weft and warp,  
  $\mathcal{R}_{\mathrm{weft}}=\{ {\mathcal{P}^1},{\mathcal{P}^3}\}$ and $\mathcal{R}_{\mathrm{warp}}=\{{\mathcal{P}^2}, {\mathcal{P}^4}\}$ 
in a reticulation. 

\begin{example}
\label{ex:runningexample-retic}
To verify that the rete in Example~\ref{ex:runningexample1} is a reticulation, we compute intersections of lines:
\[
\begin{array}{c|ccc}
\cap & P^4_1 & P^4_2 & P^4_3\\
\hline
P^1_1 & 1321 & 1112 & 1213\\
P^1_2 & 2311 & 2222 & 2123 
\end{array}
\quad
\begin{array}{c|ccc}
\cap & P^2_1 & P^2_2 & P^2_3\\
\hline
P^3_1 & 1112 & 1213 & 2311\\
P^3_2 & 2123 & 2222 & 1321 
\end{array}
\quad
\begin{array}{c|ccc}
\cap & P^4_1 & P^4_2 & P^4_3\\
\hline
P^3_1 & 2311 & 1112 & 1213\\
P^3_2 & 1321 & 2222 & 2123 
\end{array}
\]
\end{example}

Note that in a reticulation 
the size of both families of a given type equals the size of any line of the other type 
since the lines of any family are disjoint but meet the given line in just one point. 
Furthermore, the total number of points is the product of these two sizes.

\begin{theorem}
\label{thm:retic2nsoa}
Let $(\mathcal{R}_\mathrm{points}, \mathcal{R}_{\mathrm{weft}}=\{{\mathcal{P}^1},{\mathcal{P}^3}\}, \mathcal{R}_{\mathrm{warp}}=\{{\mathcal{P}^2}, {\mathcal{P}^4} \})$ be a narrow  reticulation.
Then the following construction yields a narrow semi-orthogonal array $\mathcal{Y}$.
\begin{enumerate}
\item Allow for the possibility of swapping $\mathcal{R}_{\mathrm{weft}}$ and $\mathcal{R}_{\mathrm{warp}}$.
\item Enumerate the elements of the two families of warp lines: \\ \qquad
         ${\mathcal{P}^1}=\{P^1_i\}_{i\in[m]}$ and  ${\mathcal{P}^3}=\{P^3_k\}_{k\in[m]}$.

\item Enumerate the elements of the two families of weft lines: \\ \qquad
         ${\mathcal{P}^2}=\{P^2_j\}_{j\in[n]}$ and ${\mathcal{P}^4}=\{P^4_\ell\}_{\ell\in[n]}$.
\item Let $p_{i,j}\in\mathcal{R}_\mathrm{points}$ be the unique point of intersection of $P^1_i$ and $P^2_j$ $(i\in[m], j\in[n])$.
\item Let $\mathcal{Y}=\{ (i, j, k, \ell): (i,j)\in[m]\times[n], P^3_k\cap  P^2_\ell=\{p_{i,j}\} \}$.
\end{enumerate}

 \end{theorem}

\begin{proof}
Since every weft line meets every warp line in exactly one point, as 
$(i,j,k,\ell)$ runs over the elements of $\mathcal{Y}$, the pairs $(i,j)$, $(i,\ell)$, $(k, \ell)$, and $(j, k)$ each take on every value in $[m]\times[n]$.  Since there are exactly $mn$-many points, the value of any of these pairs uniquely determines the whole tuple, that is, conditions $\ut{12}$, $\ut{23}$, $\ut{34}$, and $\ut{41}$ hold.
Hence, $\mathcal{Y}$ is a narrow semi-orthogonal array.
\end{proof}

The preceding two results give a correspondence between the rete of any narrow semi-orthogonal array and  a narrow reticulation.
The effect of choosing different orderings is addressed in Sections~\ref{sec:parastrophism} and~\ref{sec:isotopism}.

\begin{definition}
A \emph{$(k,n)$-net} (also referred to as a \emph{$k$-net of order $n$}  or \emph{Bruck net of order $k$ and degree $n$}) is an incidence structure consisting of $n^2$-many points and $nk$-many lines such that 
the lines comprise $k$-many parallel classes, each containing $n$-many lines (so each parallel class partitions the set of points), 
every line contains $n$-many points and every point is incident with $k$-many lines, one from each parallel class.
\end{definition}

Splitting the parallel classes of  an  $m$-net of size four into two families of size two yields a reticulation. 
The bijection between grid arrays and rete parallels the following relationship between orthogonal arrays and $(k,n)$-nets.

\begin{proposition}
\label{prop:OA-net}
\cite[Section 11.1]{Denes-Keedwell-LSnewdev}~%
\cite[Subsection 1.4.1]{Evans:orthLSbasedongroups}~%
\cite[Section 15.3]{Laywin+Mullen:discretemath} 
\cite[Section 8.2]{Denes+Keedwell-LSa_SecondEdition_2015}.
Let  $\mathcal{Y}=\{(i,j, L_1(i,j), \ldots, L_k(i,j)\,|\, 1\leq i,j\leq n\}$ be the orthogonal array from Proposition~\ref{prop:MOLS-OA}.  
For $1\leq h\leq k+2$, define  ${N}^h=\{N_i^h\}_{i\in [n]}$, 
      where $N_i^h=\{y\in \mathcal{Y} : y_h=i\}$.  
Then $\{{N}^h\}_{h\in[k+2]}$ is a $(k+2, n)$-net.
In particular, any $(k+2,n)$-net is equivalent to a set of $k$-many mutually orthogonal Latin squares of order $n$. 
\end{proposition}

\begin{example}
The $(4,3)$-net associated with the pair of orthogonal Latin squares of Examples~\ref{ex:LatinSquare} and~\ref{ex:latin_squares_correspondence} has the following parallel families of lines.
 \[
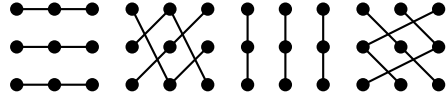
    \begin{array}{cccc}
 			\begin{tikzpicture}[scale=.5,baseline=4.5ex]
	 			\foreach \y in {1,2,3}
	         		\draw[thick, -] (1,\y)--(3,\y);
	  			\gridpoints{3}{3};  
			\end{tikzpicture}&
			\begin{tikzpicture}[scale=.5,baseline=4.5ex]
	 			\draw[thick, -] (1,3)--(2, 1)--(3,2);
                \draw[thick, -] (1,2)--(2, 3)--(3,1);
                \draw[thick, -] (1,1)--(2, 2)--(3,3);
	  			\gridpoints{3}{3};  
			\end{tikzpicture}    &
			\begin{tikzpicture}[scale=.5,baseline=4.5ex]
      			\draw[thick, -] (1,1)--(1, 3);
      			\draw[thick, -] (2,1)--(2, 3);
      			\draw[thick, -] (3,1)--(3, 3);
 	 			\gridpoints{3}{3};  
			\end{tikzpicture}  &
			\begin{tikzpicture}[scale=.5,baseline=4.5ex]
      			\draw[thick, -] (3,1)--(2, 2) -- (1,3);
      			\draw[thick, -] (2,3)--(3, 2) -- (1,1);
      			\draw[thick, -] (3,3)--(1, 2) -- (2,1);
 	 			\gridpoints{3}{3};  
			\end{tikzpicture}        
	\end{array}
\]      
\end{example}

Given a bireversible automaton $\A$, its associated rete is a reticulation.  
In this case, one may also represent $\A$ with a combinatorial structure analogous to a transversal design, 
which is the geometric dual of a $(k,n)$-net.  
We postpone further consideration to~\cite{curtin2026generalizationsnetslatinsquares}.

\subsection{From matched matrices to cooperative pairs}

We study matched matrices with the following properties.

\begin{definition}\label{def:cooppair}
A pair $(C,R)\in \mathsf{MM}(m, n)$ is said to be a \emph{cooperative pair} of order $m\times n$ whenever the following hold.
\begin{description}
\item[(CP-1)] $C\in\mathsf{CL}(m,n)$.
\item[(CP-2)]  $R\in\mathsf{RL}(m,n)$.
\item[(CP-3)] $C\rtact R\in\mathsf{CL}(m,n)$.
\item[(CP-4)] $C\ltact R\in\mathsf{RL}(m,n)$.
\end{description}
Let $\mathsf{CP}(m,n)$ denote the set of all cooperative pairs of order $m\times n$.
\end{definition}

\begin{example}
\label{ex:runningexample-CP}
In Example~\ref{ex:runningexample1}, 
   the automaton $\A$ is bireversible, 
   the pair $(C_\A,R_\A)$ is cooperative, 
   $\mathcal{Y}_\A$  is a narrow semi-orthogonal array,  and 
   the rete is a reticulation.
Indeed, $C_\A$ is column-Latin (CP-1), and  $R_\A$ is row-Latin (CP-2). 
The respective results of the operations in (CP-3) and (CP-4) above are
\[
C_\A\rtact R_\A=\begin{bmatrix} 2&1&1\\1&2&2\end{bmatrix}
\qquad \hbox{and}\qquad
C_\A\ltact R_\A=\begin{bmatrix} 2&3&1\\3&2&1\end{bmatrix},
\]
which are respectively column-Latin (CP-3) and row-Latin (CP-4).
Note that $C_\A$ and $R_\A$ are orthogonal.
\end{example}
    
\begin{theorem}
\label{thm:automatafamily-birevconditions}
Let $\A$, $(C,R)$,  $\mathcal{Y}$, and $\mathcal{P}=({\mathcal{P}^1},{\mathcal{P}^2}, {\mathcal{P}^3}, {\mathcal{P}^4})$ be associated Mealy $(m,n)$-automaton, matched pair of matrices, grid array, and grid rete.
The following are equivalent.
\begin{enumerate}
\item $\A$ is bireversible. 
\item $(C, R)$ is a cooperative pair.
\item $C$ is column-Latin, $R$ is row-Latin, and $C\rtact R$ is column-Latin.
\item $C$ is column-Latin, $R$ is row-Latin, and $C\ltact R$ is row-Latin.
\item $C$ is column-Latin, $R$ is row-Latin, and $C$ and $R$ are orthogonal.
\item $\mathcal{Y}$ is a narrow semi-orthogonal array.
\item $\mathcal{P}$ is a narrow reticulation. 
\end{enumerate}
\end{theorem}
    
\begin{proof} 
Clearly, (ii) implies each of (iii) and (iv), and together (iii) and (iv) imply (ii).  
Suppose that $C$ is column-Latin and that $R$ is row-Latin.
The action of $R$ on $C$ maps the entry of $C$ in position $(i,j)$ to position $(i,R(i,j))$ of $C\rtact R$.
That is, the $(i,R(i,j))$ entry of $C\rtact R$ is $C(i,j)$.
Since $C$ is column-Latin, $C$ has $n$ copies of each value in $[m]$, and
since $R$ is row-Latin, $R$ has $m$ copies of each value of $[n]$.
Now  $C$ and $R$ are not orthogonal if and only if there are distinct $(i,j)$, $(i', j')\in[m]\times [n]$ 
such that $C(i,j)=C(i',j')$ and $R(i,j)=R(i',j')$ (note $i\not=i'$ since $R$ is row-Latin).
This occurs if and only if there are distinct $(i,j)$, $(i', j')\in[m]\times [n]$ such that
$R(i,j)=R(i',j')$ of $C\rtact R$ contains the same value $C(i,j)=C(i',j')$ in distinct rows
if and only if $C\rtact R$ is not column-Latin.  
Thus, (iii) and (v) are equivalent.  
A similar argument gives that (iv) and (v) are equivalent.

By Theorems~\ref{thm:classes-U(st)} and~\ref{thm:array-matrix-automaton-cond}, 
properties $\ut{12}$, $\ut{23}$, $\ut{34}$, and $\ut{41}$ 
are, respectively, equivalent to  
     $C$ and $R$ being matrices, 
     $C$ being column-Latin, 
     $C$ and $R$ being orthogonal, and 
     $R$ being row-Latin.
Thus, (v) is equivalent to (vi) (Definition~\ref{def:NSOA}).
By Lemma~\ref{lem:GA-EAorth} and Theorem~\ref{thm:array-matrix-automaton-cond}
(the equivalences of parts (b) and (e)), part (v) holds if and only if part (vii) holds by Lemma~\ref{lem:CPgeom} and Theorem~\ref{thm:retic2nsoa}.  Thus, (ii) through (vii) are equivalent.

We now prove that (i) is equivalent to (iii).
Note that $\A$ is bireversible when it is invertible, reversible, and $\iz \A$ is reversible.  
According to the inversion formula~\eqref{eqn:inv} the matrix $C_{\iz\A}$ is obtained from $C_\A$ by applying permutations from the rows of $R_\A$ to the corresponding rows of $C_\A$. 
In other words,
   \[C_{\iz\A}=C_\A\rtact R_\A.\]
By Equation~\eqref{eqn:inv} the matrix $R_{\iz\A}$ is obtained from $R_\A$ by replacing row $i$ corresponding to the permutation $\sigma_i$ by the row corresponding to the permutation $\sigma_i^{-1}$ $(i\in[m])$. 
If $\A$ is bireversible, then $C_\A$ is column-Latin and $R_\A$ is row-Latin by Theorem~\ref{thm:array-matrix-automaton-cond} since $\A$ is both invertible and reversible. 
Since $\iz \A$ is reversible, $C_{\iz A}$ is column-Latin by (ii), which gives $C\rtact R$ is column-Latin.
Thus, $(C_\A,R_\A)$ is a cooperative pair.
\end{proof}

\begin{corollary}
The bijective correspondences of Lemma~\ref{lem:bijection-automata-matrices} 
restrict to bijective correspondences between the following families of Mealy automata.
\[
\begin{array}{|l|c|l|l|c|l|}
    \hline
    \hbox{\bf\begin{tabular}{c}Automaton\\family \end{tabular}} & \hbox{\bf \begin{tabular}{c}Matched\\matrices \end{tabular}} & \hbox{\bf \begin{tabular}{c}Grid array\\properties \end{tabular} } & \ \ \hbox{\bf \begin{tabular}{c}Rete\\properties \end{tabular}}\\
    \hline
    \hline
    \hbox{\emph{Mealy}}              &\mathsf{MM}(m,n)                   & \ut{12}; \mathsf{GA}  & {\mathcal{P}^1}\perp{\mathcal{P}^2}\\
    \hline
    \hbox{\emph{Invertible}}            &\mathsf{M}(m,n)\times \mathsf{RL}(m,n)  & \ut{12}, \ut{41} & {\mathcal{P}^1}\perp{\mathcal{P}^2}, {\mathcal{P}^4}\perp{\mathcal{P}^1}\\
    \hline
    \hbox{\emph{Reversible} }           &\mathsf{CL}(m,n)\times\mathsf{M}(m,n)   & \ut{12}, \ut{23}& {\mathcal{P}^1}\perp{\mathcal{P}^2}, {\mathcal{P}^2}\perp{\mathcal{P}^3}\\
    \hline
    \hbox{\emph{Coreversible}}           &\hbox{\em orthogonal}                    & \ut{12}, \ut{34}& {\mathcal{P}^1}\perp{\mathcal{P}^2}, {\mathcal{P}^3}\perp{\mathcal{P}^4} \\
    \hline
    \hbox{\emph{Invertible-}} &\mathsf{CL}(m,n)\times\mathsf{RL}(m,n)  & \ut{12}, \ut{23},  &{\mathcal{P}^1}\perp{\mathcal{P}^2}, {\mathcal{P}^4}\perp{\mathcal{P}^1},\\
    \hbox{\emph{reversible}}&&\ut{41}&{\mathcal{P}^2}\perp{\mathcal{P}^3}\\
    \hline
    \hbox{\emph{Bireversible}}          &\mathsf{CP}(m,n)                   & \ut{12}, \ut{23}       & {\mathcal{P}^1}\perp{\mathcal{P}^2}, {\mathcal{P}^2}\perp{\mathcal{P}^3}\\
                                 &                              & \ut{34}, \ut{41}; \mathsf{NSOA} & {\mathcal{P}^3}\perp{\mathcal{P}^4}, {\mathcal{P}^4}\perp{\mathcal{P}^1}\\
    \hline
\end{array}
\]
\end{corollary}

\begin{proof}
    Clear from Lemma~\ref{lem:bijection-automata-matrices} and Theorem~\ref{thm:automatafamily-birevconditions}.
\end{proof}

We view cooperative pairs as analogous to Latin squares, as discussed in
    Definition~\ref{def:latin-square}, 
    Proposition~\ref{prop:OLS-coop}, and 
    Example~\ref{ex:LatinSquare}.
We consider the corresponding analogs of mutually orthogonal Latin squares in~\cite{curtin2026generalizationsnetslatinsquares}.

\section{Parastrophy}
\label{sec:parastrophism}

We adapt the notion of parastrophy for Latin squares~\cite[Section 1.4]{Denes+Keedwell-LSa_SecondEdition_2015}~\cite[Chapter 17]{vanLintWilson:coursecomb} to  interleaved arrays in Subsection~\ref{subsec:parastrophy-basic}.
In Subsection~\ref{subsec:parastrophism-automata}, we relate parastrophisms of cooperative pairs to duals and inverses of the corresponding finite automata.

\subsection{Elements and examples}
\label{subsec:parastrophy-basic}

Parastrophisms are most naturally defined in terms of interleaved arrays.  
Let $\tau\in\Sym(4)$ act on a 4-tuple $y$ by permuting the components of $y$:
$(y^\tau)_{s} = y_{(\tau(s))}$.
Note that $y^\tau=(y_{\tau^{-1}(1)}, y_{\tau^{-1}(2)},  y_{\tau^{-1}(3)},  y_{\tau^{-1}(4)})$.
Let $\tau$ act on sets of 4-tuples $\mathcal{Y}$ by $\mathcal{Y}^{\tau}=\{y^\tau : y\in \mathcal{Y}\}$.
We are only interested in the permutations of components that preserve the interleaved array property (and ultimately the narrow semi-orthogonal array property).
Such permutations form a group that is isomorphic to the dihedral group $D_4$, the symmetries of the square with vertices $\{1, 2, 3, 4\}$.

\begin{definition}
By a \emph{parastrophism} we mean one of the following permutations of $[4]$ acting on four-tuples by permuting components:
\[
\begin{array}{|c|c|c|c|c|c|c|c|}
    \hline
    \iota & \sharp & \dagger\sharp\dagger & \sharp\dagger\sharp\dagger & \dagger &  \dagger\sharp & \sharp\dagger & \sharp\dagger\sharp\\
    \hline
    ()   &(2,4)    & (1,3)                 & (1,3)(2,4)                   & (1,2)(3,4) &  (1,2,3,4)       &(4,3,2,1)         &(1,4)(2,3)\\
    \hline
\end{array}
\]
We call the image of $\mathcal{Y}\in\mathsf{IA}(m,n)$  under a parastrophism a \emph{parastrophe} of $\mathcal{Y}$.
\end{definition}

\begin{lemma}
\label{lem:parastrophes}
Let $\mathcal{Y}\in\mathsf{IA}(m,n)$. Then the following hold.
\begin{enumerate}
\item $\mathcal{Y}^\iota$, 
      $\mathcal{Y}^\sharp$, 
      $\mathcal{Y}^{\dagger\sharp\dagger}$, 
      $\mathcal{Y}^{\sharp\dagger\sharp\dagger} \in \mathsf{IA}(m,n)$.
\item $\mathcal{Y}^\dagger$, 
      $\mathcal{Y}^{\dagger\sharp}$, 
      $\mathcal{Y}^{\sharp\dagger}$,    
      $\mathcal{Y}^{\sharp\dagger\sharp}\in \mathsf{IA}(n,m)$.
\end{enumerate}
\end{lemma}

\begin{proof}
The parastrophes correspond to the permutations of components that preserve the property that 
    the first and third components take entries from the same set and 
    the second and fourth components take entries from the same set.  
Thus, (i) and (ii) hold.
\end{proof}

\begin{proposition}
\label{prop:para-un}
Let $\mathcal{Y}$ be an interleaved array.  If $\mathcal{Y}$ has property $\un{st}$, 
then for all parastrophisms $\flat$, $\mathcal{Y}^\flat$ has property $\un{\flat(s), \flat(t)}$.    
In particular, if $\mathcal{Y}$ is a narrow semi-orthogonal array, then so are all its parastrophes.
\end{proposition}

\begin{proof}
Clear.    
\end{proof}

In the case of a grid array, we can extend the notion of parastrophy to the associated matched matrices and rete.

\begin{definition}
Let $\flat$ be a parastrophism. Let $\mathcal{Y}$ be a grid array  with associated matched matrices $(C,R)$ and  rete $\mathcal{P}=({\mathcal{P}^1},{\mathcal{P}^2}, {\mathcal{P}^3}, {\mathcal{P}^4})$.
Assume that $\mathcal{Y}$ has property $\un{\flat^{-1}(1),\flat^{-1}(2)}$, so $\mathcal{Y}^\flat$ has property $\un{12}$. 
In this case, $\mathcal{Y}^\flat$ has associated 
      matched matrices $(C,R)^\flat = (C^\flat, R^\flat)$ and 
      rete  $\mathcal{P}^\flat=({\mathcal{P}^1}^\flat,{\mathcal{P}^2}^\flat,
                                {\mathcal{P}^3}^\flat, {\mathcal{P}^4}^\flat)$.   
We refer to $(C,R)^\flat$ and $\mathcal{P}^\flat$ as \emph{parastrophes} of $(C,R)$ and $\mathcal{P}$, respectively. 
\end{definition}

\begin{example}
\label{ex:repeting}
The eight parastrophes of an interleaved array need not be distinct.  
The narrow semi-orthogonal array $\mathcal{Y}$ in Example~\ref{ex:multiset-identity} only has two distinct parastrophes.  They are $\mathcal{Y}=\mathcal{Y}^\sharp=\mathcal{Y}^{\dagger\sharp\dagger}=\mathcal{Y}^{\sharp\dagger\sharp\dagger}$ 
    and
 $\mathcal{Y}^\dagger=\mathcal{Y}^{\dagger\sharp}=\mathcal{Y}^{\dagger\sharp}=\mathcal{Y}^{\sharp\dagger\sharp}$.
Here, the associated cooperative pair is $(C=\Idc, R=\Idr)$.
\end{example}

As is the case for Latin squares~\cite{Pflugfleder:QGLI}, 
parastrophisms of cooperative pairs can be understood in terms of natural transformations of the geometry of the associated rete. 
Swapping $\mathcal{R}_{\mathrm{weft}}$ and $\mathcal{R}_{\mathrm{warp}}$ corresponds to the parastrophism $\dagger$,
and swapping $\mathcal{P}^2$ and $\mathcal{P}^4$ corresponds to the parastrophism $\sharp$.  
We present an example here and give the full result in Theorem~\ref{thm:birev-all-parastrophes} below.

\begin{example}
\label{ex:runningexample-para}
The parastrophism $\sharp$ transforms the narrow semi-orthogonal array $\mathcal{Y}$  in Example~\ref{ex:runningexample1}  
by swapping the second and fourth components. 
The associated automaton is the inverse $\iz\A$ of $\A$ (the roles of input and output are swapped).
\[\!\!
\begin{array}{c|c|c}
    \begin{array}{l}
      \mathcal{Y}^\sharp=\\
      \left\{\!\!\! \begin{array}{ccc}(1,1,2,3) & (1,2,1,1) &  (1,3,1,2)\\ 
                                (2,1,1,3) & (2,2,2,2) & (2,3,2,1) \end{array} \!\!\!\right\}
    \end{array}
 \!\! &  \!\!
    \begin{array}{ll}
    C^\sharp = & R\sharp=\\
    \begin{bmatrix} 2&1&1\\1&2&2\end{bmatrix} & \begin{bmatrix} 3&1&2\\3&2&1\end{bmatrix}
    \end{array}
 \!\! & \!\!
    \begin{tikzpicture}[node distance=1.5cm, baseline={([yshift=-.8ex]current bounding box.center)}]
    \node[state] (q1) {$1$};
    \node[state, below of=q1] (q2) {$2$};
    
    \draw (q1) edge[right, bend left] node{$1\,|\,3$} (q2);
    \draw (q1) edge[loop left] node{$2\,|\, 1$} (q1);
    \draw (q1) edge[loop right] node{$3\,|\, 2$}  (q1);
    
    \draw (q2) edge[left, bend left] node (e) {$1\,|\,3$} (q1);
    \draw (q2) edge[loop left] node{$3\,|\, 1$} (q2);
    \draw (q2) edge[loop right] node{$2\,|\, 2$}  (q2);

    \node[left=.4cm of e] (A) {$\iz\A$};
    \end{tikzpicture}
\end{array}
\]
The associated reticulation is formed by swapping ${\mathcal{P}^2}$ and ${\mathcal{P}^4}$ and then sorting the columns so that the first warp family consists of vertical lines in increasing order.
\[ 
\begin{array}{l|cccc}
&{\mathcal{P}^1}^\sharp=\mathcal{P}^1&{\mathcal{P}^3}^\sharp=\mathcal{P}^3&{\mathcal{P}^2}^\sharp=\mathcal{P}^4&{\mathcal{P}^4}^\sharp=\mathcal{P}^2\\
\hline
\hbox{permute families} &
 			\begin{tikzpicture}[scale=.5,baseline=4.5ex]
	 			\foreach \y in {1,2}
	         		\draw[thick, red, -] (1,\y)--(3,\y);
	  			\gridpoints{2}{3};  
				\node (l1) at (.3,2) {$1$};
         		 \node (l2) at (.3,1) {$2$};
			\end{tikzpicture}&
			\begin{tikzpicture}[scale=.5,baseline=4.5ex]
	   			\draw[thick, red, -] (1,1)--(2,1)--(3,2); 
       			\draw[thick, red, -] (1,2)--(2,2)--(3,1); 
       			\gridpoints{2}{3};
				\node (l1) at (.3,2) {$1$};
          		\node (l2) at (.3,1) {$2$};
			\end{tikzpicture}     &
			\begin{tikzpicture}[scale=.5,baseline=4.5ex]
      			\draw[thick, blue, -] (1,1)--(2, 2);
      			\draw[thick, blue, -] (1,2)--(2, 1);
      			\draw[thick, blue, -] (3,1)--(3, 2);
 	 			\gridpoints{2}{3};  
				\node (l1) at (1,2.6) {$3$};
           		\node (l2) at (2,2.6) {$1$};
           		\node (l3) at (3,2.6) {$2$};  
			\end{tikzpicture}     &
			\begin{tikzpicture}[scale=.5,baseline=4.5ex]
      			\draw[thick, blue, -] (1,1)--(1, 2);
      			\draw[thick, blue, -] (2,1)--(2, 2);
      			\draw[thick, blue, -] (3,1)--(3, 2);
 	 			\gridpoints{2}{3};  
				\node (l1) at (1,2.6) {$1$};
           		\node (l2) at (2,2.6) {$2$};
           		\node (l3) at (3,2.6) {$3$};  
			\end{tikzpicture} \\
\hline
\hbox{sort row entries}&
 			\begin{tikzpicture}[scale=.5,baseline=4.5ex]
	 			\foreach \y in {1,2}
	         		\draw[thick, red, -] (1,\y)--(3,\y);
	  			\gridpoints{2}{3};  
				\node (l1) at (.3,2) {$1$};
         		 \node (l2) at (.3,1) {$2$};
			\end{tikzpicture}&
			\begin{tikzpicture}[scale=.5,baseline=4.5ex]
	   			\draw[thick, red, -] (1,1)--(2,2)--(3,2); 
       			\draw[thick, red, -] (1,2)--(2,1)--(3,1); 
       			\gridpoints{2}{3};
				\node (l1) at (.3,2) {$2$};
          		\node (l2) at (.3,1) {$1$};
			\end{tikzpicture}     &
			\begin{tikzpicture}[scale=.5,baseline=4.5ex]
      			\draw[thick, blue, -] (1,1)--(1, 2);
      			\draw[thick, blue, -] (2,1)--(2, 2);
      			\draw[thick, blue, -] (3,1)--(3, 2);
 	 			\gridpoints{2}{3};  
				\node (l1) at (1,2.6) {$1$};
           		\node (l2) at (2,2.6) {$2$};
           		\node (l3) at (3,2.6) {$3$};  
			\end{tikzpicture}     &
			\begin{tikzpicture}[scale=.5,baseline=4.5ex]
      			\draw[thick, blue, -] (1,1)--(1, 2);
      			\draw[thick, blue, -] (2,1)--(3, 2);
      			\draw[thick, blue, -] (3,1)--(2, 2);
 	 			\gridpoints{2}{3};  
				\node (l1) at (1,2.6) {$3$};
           		\node (l2) at (2,2.6) {$1$};
           		\node (l3) at (3,2.6) {$2$};  
			\end{tikzpicture} 
	\end{array}
\]

\end{example}

\begin{example}
The parastrophism $\dag$ transforms the narrow semi-orthogonal array $\mathcal{Y}$ in Example~\ref{ex:runningexample1}  by applying the permutation (1,2)(3,4) to the components.  
The associated automaton  is the dual $\dz\A$ of $\A$ (the roles of states and inputs are reversed).
\[
\begin{array}{c|c|c}
    \begin{array}{l}
      \mathcal{Y}^\dag=\\
      \left\{\!\!\! \begin{array}{ccc}(1,1,2,1) & (1,2,3,2) \\
                                      (2,1,3,1)&  (2,2,2,2)\\ 
                                      (3,1,1,2) & (3,2,1,1) \end{array} \!\!\!\right\}
    \end{array}
  & 
    \begin{array}{ll}
    C^\dag = & R^\dag=\\
    \begin{bmatrix} 2&3\\3&2 \\1&1 \end{bmatrix} & \begin{bmatrix} 1&2\\1&2\\ 2&1\end{bmatrix}
    \end{array}
 &  
    \begin{tikzpicture}[node distance=2cm, baseline={([yshift=-.8ex]current bounding box.center)}]
    \node[state] (q2) {$2$};
    \node[state, below left of=q2] (q1) {$1$};
    \node[state, below right of=q2] (q3) {$3$};

    \draw (q2) edge[loop right] node{$2\,|\, 2$}  (q2);
    \draw (q2) edge[right]  node{$1\,|\, 1$} (q3);
    
    \draw (q1) edge[below, bend left] node{$2\,|\,2$} (q3);
    \draw (q1) edge[left] node{$1\,|\,1$} (q2);

    \draw (q3) edge[below, bend left] node{$1\,|\,2$} (q1);
    \draw (q3) edge[below] node{$2\,|\,1$} (q1);    

    \node[left=.3cm of q2] (A) {$\dz\A$};
    
    \end{tikzpicture}
\end{array}
\]
The associated reticulation is formed by swapping ${\mathcal{P}^1}$ and ${\mathcal{P}^2}$ and by swapping ${\mathcal{P}^3}$ and ${\mathcal{P}^2}$. 
The orientation is changed  in the process, but here sorting is not needed.
\[ 
\begin{array}{l|cccc}
&{\mathcal{P}^1}^\dag=\mathcal{P}^2 & {\mathcal{P}^3}^\dag=\mathcal{P}^4 & {\mathcal{P}^2}^\dag=\mathcal{P}^1 & {\mathcal{P}^4}^\dag=\mathcal{P}^3\\
\hline
\hbox{permute families/reorient} &
 			\begin{tikzpicture}[scale=.5,baseline=5ex]
	 			\foreach \y in {1,2, 3}
	         		\draw[thick, red, -] (1,\y)--(2,\y);
	  			\gridpoints{3}{2};  
                \node (l1) at (.3,3) {$1$};
				\node (l1) at (.3,2) {$2$};
         		 \node (l2) at (.3,1) {$3$};
			\end{tikzpicture}&
			\begin{tikzpicture}[scale=.5,baseline=4.5ex]
	   			\draw[thick, red, -] (1,1)--(2,1); 
       			\draw[thick, red, -] (1,2)--(2,3); 
                \draw[thick, red, -] (1,3)--(2,2); 
       			\gridpoints{3}{2};
				\node (l1) at (.3,3) {$2$};
          		\node (l2) at (.3,2) {$3$};
                \node (l2) at (.3,1) {$1$};
			\end{tikzpicture}     &
			\begin{tikzpicture}[scale=.5,baseline=5ex]
      			\draw[thick, blue, -] (1,1)--(1,3);
      			\draw[thick, blue, -] (2,1)--(2, 3);

 	 			\gridpoints{3}{2};  
				\node (l1) at (1,3.6) {$1$};
           		\node (l2) at (2,3.6) {$2$};

			\end{tikzpicture}     &
			\begin{tikzpicture}[scale=.5,baseline=4.5ex]
      			\draw[thick, blue, -] (1,3)--(1, 2)--(2,1);
      			\draw[thick, blue, -] (2,3)--(2, 2)--(1,1);

 	 			\gridpoints{3}{2};  
				\node (l1) at (1,3.6) {$1$};
           		\node (l2) at (2,3.6) {$2$};

			\end{tikzpicture} 
	\end{array}
\]
\end{example}

We now return to Example~\ref{def:latin-square} and Proposition~\ref{prop:OLS-coop}, which showed that a pair of orthogonal Latin squares encode a bireversible Mealy automaton. 
This class admits additional parastrophism-like operations.  
Here, the symmetries extend beyond $D_4$ to all of $\Sym(4)$.

\begin{proposition}
Let $\A=(Q=[m],X=[m],E)$ be a Mealy automaton.  Then
$\A^\flat$ is a bireversible Mealy $(m,m)$-automaton for all $\flat\in \Sym(4)$
if and only if 
$C_\A$ and $R_\A$ are orthogonal Latin squares. 
\end{proposition}

\begin{proof}
By Proposition~\ref{prop:MOLS-OA}, 
$C_\A$ and $R_A$ are orthogonal Latin squares of order $m$ 
    if and only if 
the associated grid array $\mathcal{Y}_{\!\A}$ is an orthogonal array of type $(m^2, 4,m,2)$.
By Definition~\ref{def:OA}, 
this is the case 
    if and only if 
$\mathcal{Y}_{\!\A}^\flat$ is a semi-orthogonal array for all   $\flat\in \Sym(4)$.  
By Theorem~\ref{thm:automatafamily-birevconditions}, 
this occurs 
    if and only if 
$\A^\flat$ is a bireversible Mealy $(m,m)$-automaton for all $\flat\in \Sym(4)$.
\end{proof}

The parastrophism $\dagger$ is defined for all matched matrices and their associated rete since $\un{12}$ holds if and only if $\un{21}$ holds.     
We now state a corollary to Lemma~\ref{lem:parastrophes}.

\begin{corollary}
\label{cor:dagger}
Let $(C,R)\in \mathsf{MM}(m,n)$. 
\begin{enumerate}
    \item  $(C,R)^\dagger=(R^\transp,C^\transp)\in  \mathsf{MM}(n,m)$.
    \item If $(C,R)\in \mathsf{CP}(m,n)$, then $(C,R)^\dagger\in\mathsf{CP}(n,m)$.
            In particular, $\dagger$ is a bijection between cooperative pairs of order $m\times n$ and cooperative pairs of order $n\times m$.

\end{enumerate}
\end{corollary}

\begin{proof}
For (i), it suffices to show that  $(R^\transp,C^\transp)$ and $(C,R)^\dagger$ have the same associated narrow semi-orthogonal array by Theorem~\ref{thm:bijection-pairs}.
By definition, $(C,R)^\dagger$ has associated narrow semi-orthogonal array
    $\{(j,i,R(i,j), C(i,j)  :  (i,j)\in[m]\times[n]\}$
and $(R^\transp, C^\transp)$ has associated narrow semi-orthogonal array
    $\{(j,i,R^\transp(j,i), C^\transp(j,i))  : (j,i)\in[n]\times[m]\}$.
These arrays are equal, so (i) holds.    
Part (ii) follows from the definition of $(C,R)^\dagger$ and Lemma~\ref{lem:parastrophes}.
\end{proof}

The actions under taking the transpose/dual are described in Corollary~\ref{cor:actions-transpose}.

\subsection{Parastrophisms  and automata}
\label{subsec:parastrophism-automata}

\begin{proposition}
\label{prop:dagger-dual-correspond}
    For any automaton $\A$, we have $(C_{\dz\A},R_{\dz\A})=(R_\A^\transp, C_\A^\transp)=(C_\A,R_\A)^\dagger$. In particular, the cooperative pairs of an automaton and its dual are parastrophes of each other.
\end{proposition}

\begin{proof}
Since $\dz\dz\A=\A$, to prove the first equality, it is enough to verify that the first components of pairs are equal. 
Suppose that $R^\transp_\A(i,j)=k$. 
Then $R_\A(j,i)=k$, and since the $j$-th row of $R_\A$ contains the permutation of the $j$-th state of $\A$,  there is a following transition in $\A$ of the form
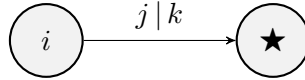
\begin{center}
    \begin{tikzpicture}
    \node[state] (q1) {$i$};
    \node[state, right of=q1] (q2) {$\bigstar$};
    \draw (q1) edge[above] node{$j\,|\,k$} (q2);
    \end{tikzpicture}
\end{center}
for some state $\bigstar$ of $\A$. 
But then, by the definition of the dual automaton, there is a transition in $\dz A$ of the form
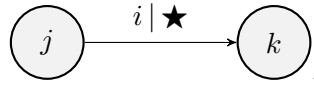
\begin{center}
    \begin{tikzpicture}
    \node[state] (q1) {$j$};
    \node[state, right of=q1] (q2) {$k$};
    \draw (q1) edge[above] node{$i\,|\,\bigstar$} (q2);
    \end{tikzpicture},
\end{center}
which means that the $i$-th entry of the $j$-th row of $C_{\dz A}$ is also $k$. 
That is to say,
\[C_{\dz\A}(i,j)=k=R_\A(j,i)=R^\transp_\A(i,j).\]
Finally, Corollary~\ref{cor:dagger} contains the second equality in the statement.
\end{proof}

Thus, the operation $\dagger$ on grid arrays corresponds to the operation $\dz$ on the associated Mealy automata.  Corollary~\ref{cor:dagger} corresponds to the claim that the $\dz$ operation defines a bijection between the sets of $(m,n)$- and $(n,m)$-automata.

\begin{proposition}
\label{prop:sharp-inv-correspond}
    For any invertible automaton $\A$ we have $(C_{\iz\A},R_{\iz\A})=(C_\A,R_\A)^\sharp$. In particular, the cooperative pairs of an invertible automaton and its inverse are parastrophes of each other.
\end{proposition}

\begin{proof}
By Lemma~\ref{lem:invertible-action-matrices}, $(C_{\iz\A},R_{\iz\A}) = (C_\A\rtact R_\A,R_\A^{\ir})$, 
which allows us to relate the narrow semi-orthogonal arrays associated with $(C_{\iz\A},R_{\iz\A})$ and $(C_\A,R_\A)$.
Suppose $(i,j,k,\ell)\in\mathcal{Y}_{\A}$.  
Then the $(i,j)$-entries of $C_\A$ and $R_\A$ are equal to $k$ and $\ell$, respectively.
By the above, the $(i,\ell)$-entries of $C_{\iz\A}$ and $R_{\iz\A}$ are equal to $k$ and $j$, respectively.
Thus, $(i,\ell, k,j)\in \mathcal{Y}_{(C_{\iz\A},R_{\iz\A})}$.  
It is clear that the argument can be reversed because $\iz\iz\A=\A$, 
so $(i,j,k,\ell)\in\mathcal{Y}_{\A}$ if and only if $(i,\ell, k,j)\in \mathcal{Y}_{(C_{\iz\A},R_{\iz\A})}$.
Thus, $\mathcal{Y}_{\A}^\sharp = \mathcal{Y}_{(C_{\iz\A},R_{\iz\A})}$.
\end{proof}

We are now able to relate the 8 bireversible automata defined by the parastrophes of a semi-orthogonal array associated with a given bireversible automaton (Lemma~\ref{lem:parastrophes})  (see, for example,~\cite{akhavi_klmp:finiteness_problem}).

\begin{theorem}
\label{thm:birev-all-parastrophes}
Let $\A$ be a letter transducer.  
The parastrophes of the associated narrow semi-orthogonal arrays, cooperative pairs, and rete are related as follows, 
where we assume that $\A$ is such that the corresponding automaton operations are well-defined and the associated objects are defined.

{\em 
\[
\begin{array}{|c||c|c|c|c|c|c|c|c|}
\hline
\hbox{Parastrophism}
&\iota & \sharp \\
\hline
\hbox{Automaton}
&\A    & \iz\A    \\
\hline
\hbox{\begin{tabular}{c} Array component\\ permutation\end{tabular}}
&()  &(24)      \\
   \hline
\hbox{Cooperative pair}
 & (C,R)  & (C \rtact R,R^{\ir})\\
   \hline
\hbox{\begin{tabular}{c} Reticulation\\ line families\end{tabular} }
 & \begin{array}{c|c}
  \mathcal{R}_\mathrm{weft} & \mathcal{R}_\mathrm{warp}\\ 
  {\mathcal{P}^1},  {\mathcal{P}^3}&{\mathcal{P}^2},{\mathcal{P}^4} 
  \end{array}
 & \begin{array}{c|c}
  \mathcal{R}_\mathrm{weft} & \mathcal{R}_\mathrm{warp}\\
  {\mathcal{P}^1},  {\mathcal{P}^3}&{\mathcal{P}^4},{\mathcal{P}^2} 
  \end{array}
\\
\hline
\hline
\hbox{Parastrophism}
& \dagger & \sharp\dagger\\
\hline
\hbox{Automaton}
&\dz\A  & \iz\dz\A     \\
\hline
\hbox{\begin{tabular}{c} Array component\\ permutation\end{tabular}}
& (1,2)(3,4) &  (1,2,3,4)      \\
   \hline
\hbox{Cooperative pair}
 & (R^\transp, C^\transp) & (R^\transp\rtact C^\transp,C^{\transp\ir}=C^{\ic \transp})\\
   \hline
\hbox{\begin{tabular}{c} Reticulation\\ line families\end{tabular} }
 & \begin{array}{c|c}
  \mathcal{R}_\mathrm{weft} & \mathcal{R}_\mathrm{warp}\\
  {\mathcal{P}^2},  {\mathcal{P}^4}&{\mathcal{P}^1},{\mathcal{P}^3} 
  \end{array}
 & \begin{array}{c|c}
  \mathcal{R}_\mathrm{weft} & \mathcal{R}_\mathrm{warp}\\
  {\mathcal{P}^3}, {\mathcal{P}^1}&  {\mathcal{P}^2},{\mathcal{P}^4} 
  \end{array}\\
\hline
\hline
\hbox{Parastrophism}
&\dagger\sharp & \sharp\dagger\sharp \\
\hline
\hbox{Automaton}
& \dz\iz\A     & \iz\dz\iz\A        \\
\hline
\hbox{\begin{tabular}{c} Array component\\ permutation\end{tabular}}
& (1,4)(2,3)     & (4,3,2,1)              \\
   \hline
\hbox{Cooperative pair}
 & (R^{\ir \transp},(C\rtact R)^\transp)  & (R^{\ir \transp} \rtact (C\rtact R)^\transp,(C\rtact R)^{\transp\ir}) \\
   \hline
\hbox{\begin{tabular}{c} Reticulation\\ line families\end{tabular} }
 &
  \begin{array}{c|c}
  \mathcal{R}_\mathrm{weft} & \mathcal{R}_\mathrm{warp}\\
  {\mathcal{P}^4},  {\mathcal{P}^2}&{\mathcal{P}^3},{\mathcal{P}^1} 
  \end{array}
  & 
  \begin{array}{c|c}
  \mathcal{R}_\mathrm{weft} & \mathcal{R}_\mathrm{warp}\\
  {\mathcal{P}^2},{\mathcal{P}^4}&{\mathcal{P}^3},{\mathcal{P}^1} 
  \end{array}
\\
\hline
\hline
\hbox{Parastrophism}
& \dagger\sharp\dagger   & \sharp\dagger\sharp\dagger\\
\hline
\hbox{Automaton}
& \dz\iz\dz\A            & \iz\dz\iz\dz\A\\
\hline
\hbox{\begin{tabular}{c} Array component\\ permutation\end{tabular}}
& (1,3)                   & (1,3)(2,4) \\
   \hline
\hbox{Cooperative pair}
& (C^{\ic},C\ltact R) & (C^{\ic}\rtact (C\ltact R),(C\ltact R)^{\ir})\\
   \hline
\hbox{\begin{tabular}{c} Reticulation\\ line families\end{tabular} }
 &
  \begin{array}{c|c}
  \mathcal{R}_\mathrm{weft} & \mathcal{R}_\mathrm{warp}\\
  {\mathcal{P}^3},  {\mathcal{P}^1}&{\mathcal{P}^2},{\mathcal{P}^4} 
  \end{array}
  & 
  \begin{array}{c|c}
  \mathcal{R}_\mathrm{weft} & \mathcal{R}_\mathrm{warp}\\
  {\mathcal{P}^3},  {\mathcal{P}^1}&{\mathcal{P}^4},{\mathcal{P}^2} 
  \end{array}
\\
\hline
\end{array}
\]   
}
\end{theorem}

\begin{proof}
Compose the actions of the parastrophisms $\dagger$  and $\sharp$  given in  Propositions~\ref{prop:dagger-dual-correspond} and~\ref{prop:sharp-inv-correspond}.  
To see that the associated cooperative pairs are as claimed, we may also argue as follows.
Since $R$ is row-Latin, $R^\ir$ is also, so  as $j$ runs over $[n]$, so does $R^\ir(i,j)$ for each fixed $i$.
Thus, 
\[\begin{array}{rcl}
\mathcal{Y} & =& \{(i,j, C(i,j), R(i,j)) :(i,j)\in[m]\times [n]\}  \\
            &=& \{(i,R^\ir(i,j), C(i,R^\ir(i,j)), R(i,R^\ir(i,j))) :(i,j)\in[m]\times [n]\}  \\
           &= &\{(i,R^\ir(i,j), (C\rtact R)(i,j)), j)  :(i,j)\in[m]\times [n]\}.
\end{array}
\]
Hence, $\mathcal{Y}^\sharp$ is associated with $((C\rtact R),R^\ir)$.
A similar argument also gives 
\[
     \mathcal{Y} = \{(C^\ir(i,j),j, i, (C\ltact R)(i,j))  :(i,j)\in[m]\times [n]\},
\]     
so $\mathcal{Y}^{\dagger\sharp\dagger}$ is associated with $((C\rtact R),R^\ir)$. 
Composing these two transformations gives that $\mathcal{Y}^{\sharp\dagger\sharp\dagger}$ is associated with $(C^{\ic},C\ltact R), (C^{\ic}\rtact (C\ltact R),(C\ltact R)^{\ir})$. 
\end{proof}

\section{Isotopy}
\label{sec:isotopism}

We adapt the notion of isotopy for Latin squares~\cite[Section 1.1]{Evans:orthLSbasedongroups}~\cite[Section 1.3]{Denes+Keedwell-LSa_SecondEdition_2015}~\cite[Chapter 17]{vanLintWilson:coursecomb} to interleaved arrays in Subsection~\ref{subsec:isotopism-basics}, and 
discuss the significance of isotopy of interleaved arrays for automata in Subsection~\ref{subsec:isotopism-automata}. 

\subsection{Elements and examples}
\label{subsec:isotopism-basics}

We begin our discussion of isotopy with interleaved arrays. 

\begin{definition}
\label{def:isotopismFA}
By an {\em isotopism} of $\mathsf{IA}(m,n)$, we mean a four-tuple
   ${\sigma}=({\rho},{\kappa},{\delta},{\epsilon})\in \Sym(m)\times \Sym(n)\times \Sym(m)\times \Sym(n)$
acting on $\mathsf{IA}(m,n)$ by
   $\sigma(\mathcal{Y}) = \{\sigma(y) : y\in \mathcal{Y}\}$ for $\mathcal{Y}\in\mathsf{IA}(m,n)$, where
\begin{equation}
     \sigma(y) = (\rho(y_1), \kappa(y_2), \delta(y_3), \epsilon(y_4)).
       \label{eq:isotopismofNSOA}
\end{equation}
\end{definition}
 
We record a few basic properties of isotopisms.

\begin{lemma}
Let $\mathcal{Y}\in\mathsf{IA}(m,n)$, and
let ${\sigma}=({\rho},{\kappa},{\delta},{\epsilon})\in \Sym(m)\times \Sym(n)\times \Sym(m)\times \Sym(n)$.
\begin{enumerate}
    \item \label{lem:isotopy-preserve-FA} 
        $\sigma(\mathcal{Y})\in\mathsf{IA}(m,n)$.
    \item Let
    $\sigma_1=(\rho, \iota_n, \iota_m,\iota_n)$, 
    $\sigma_2=(\iota_m,\kappa, \iota_m,\iota_n)$, 
    $\sigma_3=(\iota_m, \iota_n, \delta,\iota_n)$, and 
    $\sigma_4=(\iota_m,\iota_n, \iota_m, \epsilon)$, 
      where $\iota_m$ and $\iota_n$ are the identity permutations in $\Sym(m)$ and $\Sym(n)$, respectively. 
    Then $\sigma_1$, $\sigma_2$, $\sigma_3$, $\sigma_4$ commute, and 
         $\sigma$ is the composition of $\sigma_1$, $\sigma_2$, $\sigma_3$, $\sigma_4$.
    \item \label{lem:isotpy-preserve-U}
          Pick distinct $s$, $t\in[4]$.
          If $\mathcal{Y}$ has property $\un{st}$, then so does $\sigma(\mathcal{Y})$.
\end{enumerate}    
\end{lemma}

\begin{proof}
Note that $\sigma(\mathcal{Y})\in\mathsf{IA}(m,n)$ since the $\mathcal{Y}\in\mathsf{IA}(m,n)$ and $\sigma$ only permutes entries within components. 
Also, the operations of ${\rho}$, ${\kappa}$, ${\delta}$, ${\epsilon}$ on $[m]\times[n]\times[m]\times[n]$ clearly commute since they operate on different direct factors.
Typical elements of $\sigma(Y)$ are 
   $\sigma(y) = (\rho(y_1), \kappa(y_2), \delta(y_3), \epsilon(y_4))$ and 
   $\sigma(z) = (\rho(z_1), \kappa(z_2), \delta(z_3), \epsilon(z_4))$
for $y$, $z\in \mathcal{Y}$.
If $\sigma(y)_s=\sigma(z)_s$, then $y_s=z_s$ since $\rho$, $\kappa$, $\delta$, and $\epsilon$ are permutations.
A similar argument treats the $t$ components.  
Now $y=z$ by property $\un{st}$ for $\mathcal{Y}$, so $\sigma(y)=\sigma(z)$.  
Thus, $\sigma(\mathcal{Y})$ has property $\un{st}$.
\end{proof}

An isotopism of interleaved arrays induces a map on the set of letter transducers.

\begin{corollary}
\label{cor:family-correspondence}
Let $\A$ be a letter transducer, and let $\mathcal{Y}$ be the associated interleaved array.
Let  $\sigma$ be an isotopism of $\mathsf{IA}(m,n)$.
If $\A$ is Mealy, invertible, reversible, coreversible, IR, or bireversible, then so is the 
 letter transducer $\sigma(A)$ associated with $\sigma(\mathcal{Y})$. 
\end{corollary}

\begin{proof}
By Lemma~\ref{lem:isotpy-preserve-U}, isotopy preserves the properties $\un{st}$. 
\end{proof}

We now specialize to isotopisms of grid arrays, which correspond to  Mealy automata and have associated matched matrices and rete.

\begin{theorem}
\label{thm:isotopy}
Fix an associated automaton $\A=([m],[n], \pi,\lambda)$, 
     matched pair of matrices $(C,R)$,  
     grid array $\mathcal{Y}$, and 
     ordered rete $({\mathcal{P}^1}=\{P^1_i\}_{i\in[m]},{\mathcal{P}^2}=\{P^2_j\}_{j\in[n]},
       {\mathcal{P}^3}=\{P^3_k\}_{k\in[m]}, {\mathcal{P}^4}=\{P^4_\ell\}_{\ell\in[n]})$. 
Let ${\sigma}=({\rho},{\kappa},{\delta},{\epsilon})\in \Sym(m)\times \Sym(n)\times \Sym(m)\times \Sym(n)$.
Then the following are associated automaton, matched pair of matrices, grid array, and ordered rete. 
\begin{enumerate}
\item The isotopy $\sigma$ transforms each transition  
          $i  \xrightarrow{j|\lambda(i,j)}  \pi(i,j)$ 
      of $\A$ into the transition
          $\rho(i) \xrightarrow{\kappa(j)|\epsilon(\lambda(i,j))} \delta(\pi(i,j))$
      of $\sigma(\A)$.
    In particular, 
\begin{eqnarray*}
    \sigma(\A) & =&  (Q^\sigma, X^\sigma, \pi^\sigma, \lambda^\sigma), \hbox{ where }\\
    Q^\sigma&=&\{q^\sigma_1, q^\sigma_2,\ldots q^\sigma_m\} \hbox{ with } q^\sigma_i = \rho(i),\\
    X^\sigma &=& \{x^\sigma_1, x^\sigma_2,\ldots x^\sigma_n\} \hbox{ with } x^\sigma_j = \kappa(j),\\
    \pi^\sigma(q^\sigma_i, x^\sigma_j) &=& \delta ( \pi(\rho^{-1}(q^\sigma_i), \kappa^{-1}(x^\sigma_j))),\\
    \lambda^\sigma(q^\sigma_i, x^\sigma_j) &=& \epsilon ( \lambda(\rho^{-1}(q^\sigma_i), \kappa^{-1}(x^\sigma_j)).
\end{eqnarray*}    
\item The isotopy $\sigma$ acts on matched matrices by 
\begin{description}
    \item[\bf{(I-R)}]  permuting the order of the rows of $C$ and $R$ with $\rho$,
    \item[\bf{(I-C)}]  permuting the order of the columns of $C$ and $R$ with $\kappa$, 
    \item[\bf{(I-Ec)}] permuting the entries of $C$ with $\delta$, and 
    \item[\bf{(I-Er)}] permuting the entries of $R$ with $\epsilon$.
\end{description}
 In particular, 
    $\sigma(C,R) = (\sigma(C), \sigma(R))$, where
    \begin{equation}
    {\sigma}(C)(i,j) = \delta(C(\rho^{-1}(i),\kappa^{-1}(j)))
              \hbox{\quad and \quad}  
    {\sigma}(R)(i,j) = \epsilon(R(\rho^{-1}(i),\kappa^{-1}(j))).
    \label{eq:isotopismmap}
    \end{equation} 
\item $\sigma(\mathcal{Y}) = \{\sigma(y) : y\in \mathcal{Y}\}$, where
    \begin{equation}
         \sigma(y) = (\rho(y_1), \kappa(y_2), \delta(y_3), \epsilon(y_4)).
    \end{equation}
\item \label{thm:isotopy:geom}
The isotopy $\sigma$ acts on ordered rete by permuting the indices of the lines in
     $\mathcal{P}^1$, $\mathcal{P}^2$, $\mathcal{P}^3$, and $\mathcal{P}^4$ by $\rho$, $\kappa$, $\delta$, $\epsilon$, respectively.
     In particular, 
\[
    \sigma({\mathcal{P}^1})_i = \sigma( P^1_{\rho^{-1}(i)}), \,
    \sigma({\mathcal{P}^2})_j =  \sigma(P^2_{\kappa^{-1}(j)}), \,
    \sigma({\mathcal{P}^3})_k =  \sigma(P^3_{\delta^{-1}(k)}, \,
    \sigma({\mathcal{P}^4})_\ell = \sigma( P^4_{\epsilon^{-1}(\ell)}).
\]
  Note: the points should be placed in the natural arrangement when drawn.
\end{enumerate}
\end{theorem}

\begin{proof}
By Lemma~\ref{lem:isotopy-preserve-FA}, $\sigma(\mathcal{Y})$ is a grid array. 
We appeal to Definition~\ref{def:bijections} and Theorem~\ref{thm:bijection-pairs} to show 
the given automaton, matched pair of matrices, grid array, and ordered rete are associated.  
Observe that 
$\sigma(\mathcal{Y}) = \{(\rho(i), \kappa(j), \delta(k),\epsilon(\ell)) : (i,j,k,\ell)\in\mathcal{Y}\}$.
The associated matched matrices are $\sigma(C)$, $\sigma(R)$ given by (\ref{eq:isotopismmap}).
The associated automaton has transitions given in (i).
Note that $\sigma({\mathcal{P}^1})_i =\{(a,b,c,d)\in\sigma(\mathcal{Y}) : \rho(y_1)=a=i\}
 = \sigma(\{y\in\mathcal{Y} : y_1=\rho^{-1}(i)\}) = \sigma(\mathcal{P}^{1}_{\rho^{-1}(i)})$.  
The other three families are treated similarly.
\end{proof}

\begin{definition}
\label{def:isotopism}
By an \emph{isotopism} of a finite automaton, matched pair of matrices, grid array, or ordered rete, 
we mean a four-tuple of permutations 
     ${\sigma}=({\rho},{\kappa},{\delta},{\epsilon})\in \Sym(m)\times \Sym(n)\times \Sym(m)\times \Sym(n)$
acting as in Theorem~\ref{thm:isotopy}.
We say that $X$ is \emph{isotopic} to ${\sigma}(X)$ or that $X$ and $\sigma(X)$ are \emph{isotopes} of one another
for any $X$ among the above objects.
\end{definition}

The geometric interpretation of isotopisms in Theorem~\ref{thm:isotopy:geom} is analogous to that of isotopisms for Latin squares in terms of $(k,n)$-nets~\cite[Chapter 2]{Pflugfleder:QGLI}.

Isotopisms respect the various properties considered in previous sections.

\begin{corollary}
\label{cor:family-correspondence-matrix-rete}
With reference to Theorem~\ref{thm:isotopy}, the following hold.
\begin{enumerate}
    \item If $C$ is column-Latin, $R$ is row-Latin, or $C$ and $R$ are orthogonal, then $\sigma(C)$ is column-Latin, $\sigma(R)$ is row-Latin, or $\sigma(C)$ and $\sigma(R)$ are orthogonal, respectively.
    \item If two line families of different types are orthogonal, then the images of these families under $\sigma$ are also orthogonal.
\end{enumerate}    
\end{corollary}

\begin{proof}
By Lemma~\ref{lem:isotpy-preserve-U}, isotopy preserves the properties $\un{st}$.
The result follows from Corollary~\ref{cor:family-correspondence}.    
\end{proof}

\begin{example}
\label{ex:runningexample-isotopism-rho}
Apply the isotopism $\sigma=( (1,2), \iota_n, \iota_m, \iota_n)$ to the objects in Example~\ref{ex:runningexample1}.
Permute the values in the first component of every element of $\mathcal{Y}$ by $\rho=(1,2)$.  
Permute the rows of both $C$ and $R$ by $\rho^{-1}$.
The permutation $\rho$ is applied to the indices of the lines in $\mathcal{P}^1$; 
sorting the points into the natural arrangement is done by permuting the rows of all points by $\rho^{-1}$.  
The transitions of the automaton are $\rho(i) \xrightarrow{j|\lambda(i,j)} \pi(i,j)$.
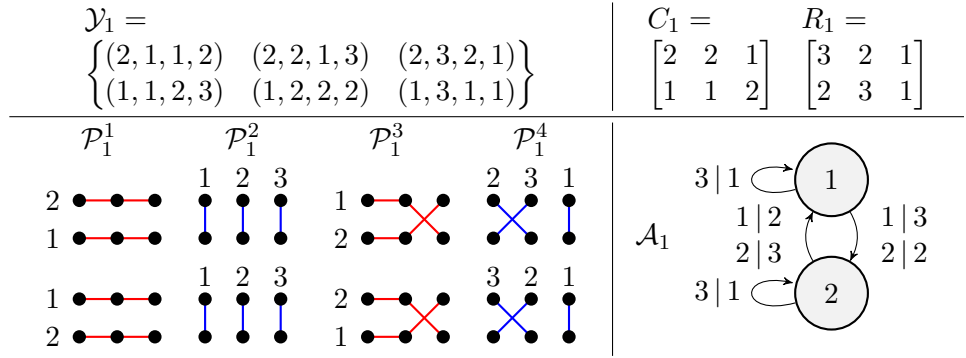
\[
\begin{array}{c|c}
    \begin{array}{l}
      \mathcal{Y}_1=\\
      \left\{\!\!\! \begin{array}{ccc} (2,1,1,2) & (2,2,1,3) & (2,3,2,1)\\
                                       (1,1,2,3) & (1,2,2,2) &  (1,3,1,1)\\ \end{array} \!\!\!\right\}
    \end{array}
  &  
    \begin{array}{ll}
    C_1 = & R_1=\\
    \begin{bmatrix} 2&2&1\\1&1&2\end{bmatrix} & \begin{bmatrix} 3&2&1\\2&3&1\end{bmatrix}
    \end{array}
\\
\noalign{\vskip 2pt}
\hline
\noalign{\vskip 2pt}
\begin{array}{cccc}
\mathcal{P}^1_1& \mathcal{P}^2_1&\mathcal{P}^3_1&\mathcal{P}^4_1\\
			\begin{tikzpicture}[scale=.5,baseline=7ex]
	   			\draw[thick, red, -] (1,1)--(2,1)--(3,1); 
       				\draw[thick, red, -] (1,2)--(2,2)--(3,2); 
       				\gridpoints{2}{3};
				\node (l1) at (.3,2) {$2$};
          			\node (l2) at (.3,1) {$1$};
			\end{tikzpicture}
&
			 \begin{tikzpicture}[scale=.5,baseline=7ex]
      				\draw[thick, blue, -] (1,1)--(1, 2);
      				\draw[thick, blue, -] (2,1)--(2, 2);
      				\draw[thick, blue, -] (3,1)--(3, 2);
 	 			\gridpoints{2}{3};  
				\node (l1) at (1,2.6) {$1$};
           			\node (l2) at (2,2.6) {$2$};
           			\node (l3) at (3,2.6) {$3$};  
			\end{tikzpicture}        
&
			\begin{tikzpicture}[scale=.5,baseline=7ex]
	   			\draw[thick, red, -] (1,1)--(2,1)--(3,2); 
       				\draw[thick, red, -] (1,2)--(2,2)--(3,1); 
       				\gridpoints{2}{3};
				\node (l1) at (.3,2) {$1$};
          			\node (l2) at (.3,1) {$2$};
			\end{tikzpicture}
&
			 \begin{tikzpicture}[scale=.5,baseline=7ex]
      				\draw[thick, blue, -] (1,1)--(2, 2);
      				\draw[thick, blue, -] (1,2)--(2, 1);
      				\draw[thick, blue, -] (3,1)--(3, 2);
 	 			\gridpoints{2}{3};  
				\node (l1) at (1,2.6) {$2$};
           			\node (l2) at (2,2.6) {$3$};
           			\node (l3) at (3,2.6) {$1$};  
			\end{tikzpicture}
\\
			\begin{tikzpicture}[scale=.5,baseline=7ex]
	   			\draw[thick, red, -] (1,1)--(2,1)--(3,1); 
       				\draw[thick, red, -] (1,2)--(2,2)--(3,2); 
       				\gridpoints{2}{3};
				\node (l1) at (.3,2) {$1$};
          			\node (l2) at (.3,1) {$2$};
			\end{tikzpicture}
&
			 \begin{tikzpicture}[scale=.5,baseline=7ex]
      				\draw[thick, blue, -] (1,1)--(1, 2);
      				\draw[thick, blue, -] (2,1)--(2, 2);
      				\draw[thick, blue, -] (3,1)--(3, 2);
 	 			\gridpoints{2}{3};  
				\node (l1) at (1,2.6) {$1$};
           			\node (l2) at (2,2.6) {$2$};
           			\node (l3) at (3,2.6) {$3$};  
			\end{tikzpicture}
&
			\begin{tikzpicture}[scale=.5,baseline=7ex]
	   			\draw[thick, red, -] (1,1)--(2,1)--(3,2); 
       				\draw[thick, red, -] (1,2)--(2,2)--(3,1); 
       				\gridpoints{2}{3};
				\node (l1) at (.3,2) {$2$};
          			\node (l2) at (.3,1) {$1$};
			\end{tikzpicture}
&
			 \begin{tikzpicture}[scale=.5,baseline=7ex]
      				\draw[thick, blue, -] (1,1)--(2, 2);
      				\draw[thick, blue, -] (1,2)--(2, 1);
      				\draw[thick, blue, -] (3,1)--(3, 2);
 	 			\gridpoints{2}{3};  
				\node (l1) at (1,2.6) {$3$};
           			\node (l2) at (2,2.6) {$2$};
           			\node (l3) at (3,2.6) {$1$};  
			\end{tikzpicture}         
\end{array} 
&
    \begin{tikzpicture}[node distance=1.5cm, baseline={([yshift=-.8ex]current bounding box.center)}]
    \node[state] (q1) {$1$};
    \node[state, below of=q1] (q2) {$2$};
    
    \draw (q1) edge[right, bend left] node{$\begin{array}{c}1\,|\,3\\ 2\,|\,2\end{array}$} (q2);
    \draw (q1) edge[loop left] node{$3\,|\, 1$} (q1);
    
    \draw (q2) edge[left, bend left] node (e) {$\begin{array}{c}1\,|\,2 \\ 2\,|\,3\end{array}$} (q1);
    \draw (q2) edge[loop left] node{$3\,|\, 1$} (q2);

    \node[left=.4cm of e] (A) {$\A_1$};
    \end{tikzpicture}
\end{array}
\]
\end{example}

\begin{example}
\label{ex:runningexample-isotopism-kappa}
Apply the isotopism $\sigma=( \iota_m, (123), \iota_m, \iota_n)$ to the objects in Example~\ref{ex:runningexample1}.
Permute the values in the second component of every element of $\mathcal{Y}$ by $\kappa=(123)$.  
Permute the columns of both $C$ and $R$ by $\kappa^{-1}$.
The permutation $\kappa$ is applied to the indices of the lines in $\mathcal{P}^2$; 
sorting the points into the natural arrangement is done by permuting the columns of all points by $\kappa^{-1}$.  
The transitions of the automaton are $i \xrightarrow{\kappa(j)|\lambda(i,j)} \pi(i,j)$.
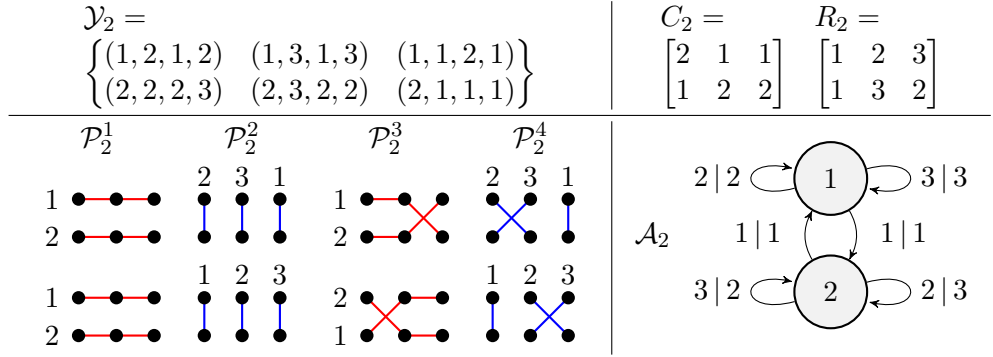
\[
\begin{array}{c|c}
    \begin{array}{l}
      \mathcal{Y}_2=\\
      \left\{\!\!\! \begin{array}{ccc}(1,2,1,2) & (1,3,1,3) &  (1,1,2,1)\\ 
                                (2,2,2,3) & (2,3,2,2) & (2,1,1,1) \end{array} \!\!\!\right\}
    \end{array}
  &  
    \begin{array}{ll}
    C_2 = & R_2=\\
    \begin{bmatrix} 2&1&1\\ 1&2&2 \end{bmatrix} & \begin{bmatrix} 1&2&3\\ 1&3&2\end{bmatrix}
    \end{array}
\\
\noalign{\vskip 2pt}
\hline
\noalign{\vskip 2pt}
\begin{array}{cccc}
\mathcal{P}^1_2& \mathcal{P}^2_2& \mathcal{P}^3_2&\mathcal{P}^4_2\\
			\begin{tikzpicture}[scale=.5,baseline=7ex]
	   			\draw[thick, red, -] (1,1)--(2,1)--(3,1); 
       				\draw[thick, red, -] (1,2)--(2,2)--(3,2); 
       				\gridpoints{2}{3};
				\node (l1) at (.3,2) {$1$};
          			\node (l2) at (.3,1) {$2$};
			\end{tikzpicture}
&
			 \begin{tikzpicture}[scale=.5,baseline=7ex]
      				\draw[thick, blue, -] (1,1)--(1, 2);
      				\draw[thick, blue, -] (2,1)--(2, 2);
      				\draw[thick, blue, -] (3,1)--(3, 2);
 	 			\gridpoints{2}{3};  
				\node (l1) at (1,2.6) {$2$};
           			\node (l2) at (2,2.6) {$3$};
           			\node (l3) at (3,2.6) {$1$};  
			\end{tikzpicture}
&
			\begin{tikzpicture}[scale=.5,baseline=7ex]
	   			\draw[thick, red, -] (1,1)--(2,1)--(3,2); 
       				\draw[thick, red, -] (1,2)--(2,2)--(3,1); 
       				\gridpoints{2}{3};
				\node (l1) at (.3,2) {$1$};
          			\node (l2) at (.3,1) {$2$};
			\end{tikzpicture}        
&
 			 \begin{tikzpicture}[scale=.5,baseline=7ex]
      				\draw[thick, blue, -] (1,1)--(2, 2);
      				\draw[thick, blue, -] (1,2)--(2, 1);
      				\draw[thick, blue, -] (3,1)--(3, 2);
 	 			\gridpoints{2}{3};  
				\node (l1) at (1,2.6) {$2$};
           			\node (l2) at (2,2.6) {$3$};
           			\node (l3) at (3,2.6) {$1$};  
			\end{tikzpicture}  
\\
			\begin{tikzpicture}[scale=.5,baseline=7ex]
	   			\draw[thick, red, -] (1,1)--(2,1)--(3,1); 
       				\draw[thick, red, -] (1,2)--(2,2)--(3,2); 
       				\gridpoints{2}{3};
				\node (l1) at (.3,2) {$1$};
          			\node (l2) at (.3,1) {$2$};
			\end{tikzpicture}
&
			 \begin{tikzpicture}[scale=.5,baseline=7ex]
      				\draw[thick, blue, -] (1,1)--(1, 2);
      				\draw[thick, blue, -] (2,1)--(2, 2);
      				\draw[thick, blue, -] (3,1)--(3, 2);
 	 			\gridpoints{2}{3};  
				\node (l1) at (1,2.6) {$1$};
           			\node (l2) at (2,2.6) {$2$};
           			\node (l3) at (3,2.6) {$3$};  
			\end{tikzpicture} 
&
			\begin{tikzpicture}[scale=.5,baseline=7ex]
	   			\draw[thick, red, -] (1,1)--(2,2)--(3,2); 
       				\draw[thick, red, -] (1,2)--(2,1)--(3,1); 
       				\gridpoints{2}{3};
				\node (l1) at (.3,2) {$2$};
          			\node (l2) at (.3,1) {$1$};
			\end{tikzpicture}
&
			\begin{tikzpicture}[scale=.5,baseline=7ex]
      				\draw[thick, blue, -] (3,1)--(2, 2);
      				\draw[thick, blue, -] (3,2)--(2, 1);
      				\draw[thick, blue, -] (1,1)--(1, 2);
 	 			\gridpoints{2}{3};  
				\node (l1) at (1,2.6) {$1$};
           			\node (l2) at (2,2.6) {$2$};
           			\node (l3) at (3,2.6) {$3$};  
			\end{tikzpicture}   
\end{array} 
&
    \begin{tikzpicture}[node distance=1.5cm, baseline={([yshift=-.8ex]current bounding box.center)}]
    \node[state] (q1) {$1$};
    \node[state, below of=q1] (q2) {$2$};
    
    \draw (q1) edge[right, bend left] node{$\begin{array}{c}1\,|\,1\end{array}$} (q2);
    \draw (q1) edge[loop left] node{$2\,|\, 2$} (q1);
    \draw (q1) edge[loop right] node{$3\,|\, 3$}  (q1);
    
    \draw (q2) edge[left, bend left] node (e) {$\begin{array}{c}1\,|\,1\end{array}$} (q1);
    \draw (q2) edge[loop left] node{$3\,|\, 2$} (q2);
    \draw (q2) edge[loop right] node{$2\,|\, 3$}  (q2);

    \node[left=.4cm of e] (A) {$\A_2$};
    \end{tikzpicture}
\end{array}
\]
\end{example}  

\begin{example}
\label{ex:runningexample-isotopism-delta}
Apply the isotopism $\sigma=( \iota_m, \iota_n, (1,2), \iota_n)$ to the objects in Example~\ref{ex:runningexample1}.
Permute the values in the third component of every element of $\mathcal{Y}$ by $\delta=(1,2)$.  
Permute the entries of $C$ by $\delta$.
Permute the indices of the rows in $\mathcal{P}^3$ by $\delta$.
The transitions of the automaton are $i \xrightarrow{j|\lambda(i,j)} \delta(\pi(i,j))$.
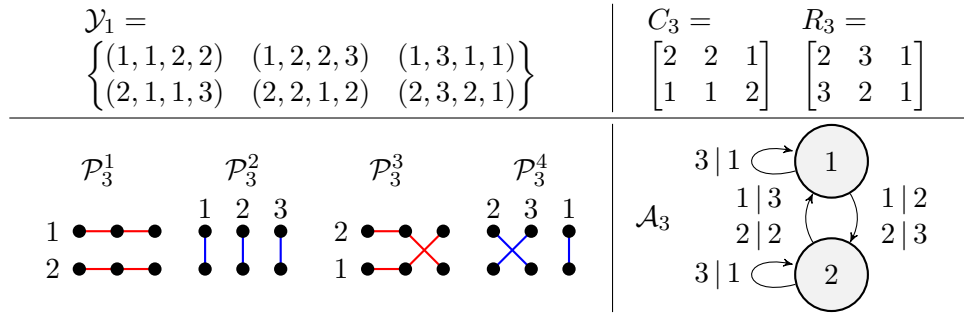
\[
\begin{array}{c|c}
    \begin{array}{l}
      \mathcal{Y}_1=\\
      \left\{\!\!\! \begin{array}{ccc} (1,1,2,2) & (1,2,2,3) & (1,3,1,1)\\
                                       (2,1,1,3) & (2,2,1,2) &  (2,3,2,1)\\ \end{array} \!\!\!\right\}
    \end{array}
  &  
    \begin{array}{ll}
    C_3 = & R_3=\\
    \begin{bmatrix} 2&2&1\\1&1&2\end{bmatrix} & \begin{bmatrix} 2&3&1\\3&2&1\end{bmatrix}
    \end{array}
\\
\noalign{\vskip 2pt}
\hline
\noalign{\vskip 2pt}
\begin{array}{cccc}
\mathcal{P}^1_3& \mathcal{P}^2_3&\mathcal{P}^3_3&\mathcal{P}^4_3\\
			\begin{tikzpicture}[scale=.5,baseline=7ex]
	   			\draw[thick, red, -] (1,1)--(2,1)--(3,1); 
       				\draw[thick, red, -] (1,2)--(2,2)--(3,2); 
       				\gridpoints{2}{3};
				\node (l1) at (.3,2) {$1$};
          			\node (l2) at (.3,1) {$2$};
			\end{tikzpicture}
&
			 \begin{tikzpicture}[scale=.5,baseline=7ex]
      				\draw[thick, blue, -] (1,1)--(1, 2);
      				\draw[thick, blue, -] (2,1)--(2, 2);
      				\draw[thick, blue, -] (3,1)--(3, 2);
 	 			\gridpoints{2}{3};  
				\node (l1) at (1,2.6) {$1$};
           			\node (l2) at (2,2.6) {$2$};
           			\node (l3) at (3,2.6) {$3$};  
			\end{tikzpicture}        
&
			\begin{tikzpicture}[scale=.5,baseline=7ex]
	   			\draw[thick, red, -] (1,1)--(2,1)--(3,2); 
       				\draw[thick, red, -] (1,2)--(2,2)--(3,1); 
       				\gridpoints{2}{3};
				\node (l1) at (.3,2) {$2$};
          			\node (l2) at (.3,1) {$1$};
			\end{tikzpicture}
&
			 \begin{tikzpicture}[scale=.5,baseline=7ex]
      				\draw[thick, blue, -] (1,1)--(2, 2);
      				\draw[thick, blue, -] (1,2)--(2, 1);
      				\draw[thick, blue, -] (3,1)--(3, 2);
 	 			\gridpoints{2}{3};  
				\node (l1) at (1,2.6) {$2$};
           			\node (l2) at (2,2.6) {$3$};
           			\node (l3) at (3,2.6) {$1$};  
			\end{tikzpicture}       
\end{array} 
&
    \begin{tikzpicture}[node distance=1.5cm, baseline={([yshift=-.8ex]current bounding box.center)}]
    \node[state] (q1) {$1$};
    \node[state, below of=q1] (q2) {$2$};
    
    \draw (q1) edge[right, bend left] node{$\begin{array}{c}1\,|\,2\\ 2\,|\,3\end{array}$} (q2);
    \draw (q1) edge[loop left] node{$3\,|\, 1$} (q1);
    
    \draw (q2) edge[left, bend left] node (e) {$\begin{array}{c}1\,|\,3 \\ 2\,|\,2\end{array}$} (q1);
    \draw (q2) edge[loop left] node{$3\,|\, 1$} (q2);

    \node[left=.4cm of e] (A) {$\A_3$};
    \end{tikzpicture}
\end{array}
\]
\end{example}

\begin{example}
\label{ex:runningexample-isotopism-epsilon}
Apply the isotopism $\sigma=(\iota_m, \iota_n, \iota_m, (123))$ to the objects in Example~\ref{ex:runningexample1}.
Permute the values in the fourth component of every element of $\mathcal{Y}$ by $\epsilon=(123)$.  
Permute the entries of $R$ by $\epsilon$.
Permute the indices of the columns in $\mathcal{P}^4$ by $\epsilon$.
The transitions of the automaton are $i \xrightarrow{j|\epsilon(\lambda(i,j))} \pi(i,j)$.
\[\!\!
\begin{array}{c|c}
    \begin{array}{l}
      \mathcal{Y}_4=\\
      \left\{\!\!\! \begin{array}{ccc}(1,1,1,3) & (1,2,1,1) &  (1,3,2,2)\\ 
                                (2,1,2,1) & (2,2,2,3) & (2,3,1,2) \end{array} \!\!\!\right\}
    \end{array}
    &
    \begin{array}{ll}
    C_4 = & R_4=\\
    \begin{bmatrix} 1&1&2\\2&2&1\end{bmatrix} & \begin{bmatrix} 3&1&2\\1&3&2\end{bmatrix}
    \end{array}
\\
\noalign{\vskip 2pt}
\hline
\noalign{\vskip 2pt}
\begin{array}{cccc}
\mathcal{P}^4_1& \mathcal{P}^4_2&\mathcal{P}^4_3& \mathcal{P}^4_4\\
\raisebox{0pt}[0pt][0pt]{
			\begin{tikzpicture}[scale=.5,baseline=8ex]
	   			\draw[thick, red, -] (1,1)--(2,1)--(3,1); 
       				\draw[thick, red, -] (1,2)--(2,2)--(3,2); 
       				\gridpoints{2}{3};
				\node (l1) at (.3,2) {$1$};
          			\node (l2) at (.3,1) {$2$};
			\end{tikzpicture}}
&
\raisebox{0pt}[0pt][0pt]{
			 \begin{tikzpicture}[scale=.5,baseline=8ex]
      				\draw[thick, blue, -] (1,1)--(1, 2);
      				\draw[thick, blue, -] (2,1)--(2, 2);
      				\draw[thick, blue, -] (3,1)--(3, 2);
 	 			\gridpoints{2}{3};  
				\node (l1) at (1,2.6) {$1$};
           			\node (l2) at (2,2.6) {$2$};
           			\node (l3) at (3,2.6) {$3$};  
			\end{tikzpicture}}     
&
\raisebox{0pt}[0pt][0pt]{
			\begin{tikzpicture}[scale=.5,baseline=8ex]
	   			\draw[thick, red, -] (1,1)--(2,1)--(3,2); 
       				\draw[thick, red, -] (1,2)--(2,2)--(3,1); 
       				\gridpoints{2}{3};
				\node (l1) at (.3,2) {$1$};
          			\node (l2) at (.3,1) {$2$};
			\end{tikzpicture}}
&
\raisebox{0pt}[0pt][0pt]{
			 \begin{tikzpicture}[scale=.5,baseline=8ex]
      				\draw[thick, blue, -] (1,1)--(2, 2);
      				\draw[thick, blue, -] (1,2)--(2, 1);
      				\draw[thick, blue, -] (3,1)--(3, 2);
 	 			\gridpoints{2}{3};  
				\node (l1) at (1,2.6) {$3$};
           			\node (l2) at (2,2.6) {$1$};
           			\node (l3) at (3,2.6) {$2$};  
			\end{tikzpicture}} 
\end{array}     
&
 \begin{tikzpicture}[node distance=1.5cm, baseline={([yshift=.2ex]current bounding box.center)}]
    \node[state] (q1) {$1$};
    \node[state, below of=q1] (q2) {$2$};
    
    \draw (q1) edge[right, bend left] node{$3\,|\,2$} (q2);
    \draw (q1) edge[loop left] node{$1\,|\, 3$} (q1);
    \draw (q1) edge[loop right] node{$2\,|\, 1$}  (q1);
    \draw (q2) edge[left, bend left] node (e) {$3\,|\,2$} (q1);
    \draw (q2) edge[loop left] node{$1\,|\, 1$} (q2);
    \draw (q2) edge[loop right] node{$2\,|\, 3$}  (q2);

    \node[left=.4cm of e] (A) {$\A'$};
    \end{tikzpicture}
\end{array}    
\]
\end{example}

Write $\mathfrak{I}=\Sym(m)\times \Sym(n)\times \Sym(m)\times \Sym(n)$ viewed as the group of isotopisms of $[m]\times[n]\times[m]\times[n]$.
Write $\mathfrak{P}$ to denote the set of parastrophisms of  $[m]\times[n]\times[m]\times[n]$.
Call a parastrophism of $[m]\times[n]\times[m]\times[n]$ \emph{shape preserving} when its image is in this set,
and let $\mathfrak{P}^\ast$ be the group of shape-preserving parastrophisms.  
Note that
$\mathfrak{P}^\ast=\{\iota, \sharp, \dagger\sharp\dagger, \sharp\dagger\sharp\dagger\}$ 
when $m\not=n$,  and 
$\mathfrak{P}^\ast=\mathfrak{P}$
when $m=n$.

\begin{proposition}
Let ${\mathfrak{S}}$ be the group generated by $\mathfrak{I}$ and ${\mathfrak{P}}^\ast$ in the group of symmetries of $[m]\times[n]\times[m]\times[n]$.  Then $\mathfrak{S}=\mathfrak{I}\rtimes {\mathfrak{P}}^\ast$. 
\end{proposition}

\begin{proof}
Observe that $\mathfrak{I}\cap \mathfrak{P}^\ast= \{\iota\}$ (the identity subgroup) 
since they move different parts of $[m]\times[n]\times[m]\times[n]$.
To show that  $\mathfrak{I}\lhd {\mathfrak{S}}$, 
it is enough to show that $p^{-1} s p\in \mathfrak{I}$ for all $p\in  {\mathfrak{P}}^\ast$ and $s\in\mathfrak{I}$.  
Since $p$ is shape preserving, the permutations of values within components are defined after applying $p$.  
Now permuting components, applying permutations to values within the components, and 
then undoing the permutation of components  has the effect of applying permutations to values within the components 
(the net effect is that the permutations of the isotopism are applied in different components). 
Thus, the result is an isotopism.  
Hence, $p^{-1} s p\in \mathfrak{I}$, so  $\mathfrak{I}\lhd {\mathfrak{S}}$.  
Since ${\mathfrak{S}}$ is generated by $\mathfrak{I}$ and $\mathfrak{P}^\ast$, 
every element of $\mathfrak{S}$ is a word in elements of $\mathfrak{I}$ and $\mathfrak{P}^\ast$.  
The normality of $\mathfrak{I}$ in $\mathfrak{S}$ allows the interchange of a product $ps$ 
($s\in\mathfrak{I}$, $p\in\mathfrak{P}^\ast$) for a product $s'p$ for some $s'\in\mathfrak{I}$. 
In particular, $\mathfrak{S}=\mathfrak{I}\mathfrak{P}^\ast$.  
Hence, $\mathfrak{S}=\mathfrak{I}\rtimes \mathfrak{P}^\ast$.
\end{proof}

\subsection{Isotopisms and automata}
\label{subsec:isotopism-automata}

Let $\A=(Q, X, E)$ be a letter transducer, and fix orderings $(q_1, \ldots, q_m)$ of $Q$ and $(x_1, \ldots, x_n)$ of $X$.
An isotopism $\sigma=(\rho,\kappa,\delta, \epsilon)$ of $\mathsf{IA}(m,n)$ induces an action of $\sigma$ on $\A$ by
$\sigma(\A) = (Q, X, \sigma(E))$, where 
    $\sigma(E) = \{(q_{\rho(i)}, x_{\kappa(j)}, r_{\delta(k)}, y_{\epsilon(\ell)})  : (q_i,x_j,r_k,y_\ell)\in E\}$.

\begin{theorem}
\label{thm:isotopism_and_automata}
Let $\A=(Q, X, E)$ be a letter transducer, and fix orderings $(q_1, \ldots, q_m)$ of $Q$ and $(x_1, \ldots, x_n)$ of $X$.
\begin{enumerate}
\item Let $\rho\in\Sym(m)$, and let $\sigma=(\rho, \iota_n, \rho, \iota_n)$.
Then $\sigma(\A)$ is related to $\A$ by permuting the set of states $Q$ of $\A$ with $\rho$. 

\item Let $\kappa\in\Sym(n)$, and let $\sigma=(\iota_m, \kappa,  \iota_m, \kappa)$.
Then $\sigma(\A)$ is related to $\A$ by permuting the alphabet $X$ of $\A$ with $\kappa$. 
\end{enumerate}
\end{theorem}

\begin{proof}
In (i), each transition      $(q_i,x_j,r_k,y_\ell)$
      is transformed into    $(q_{\rho(i)}, x_{j}, r_{\rho(k)}, y_{\ell})$.
In (ii), each transition     $(q_i,x_j,r_k,y_\ell)$ 
      is transformed into    $(q_i,x_{\kappa(j)},r_k,y_{\kappa(\ell)})$.         
\end{proof}

\begin{corollary}
Permuting the states  or letters of the alphabet  of a letter transducer corresponds to an isotopism of the associated interleaved array.
\end{corollary}

It follows from the above corollary that the isotopisms  $(\rho, \iota_n, \rho, \iota_n)$ and $(\iota_m, \kappa,  \iota_m, \kappa)$ do not change the (semi)group structure of the associated automaton (semi)group. These isotopisms were referred to as \emph{automata symmetries} in the paper~\cite{bondarenko_gkmnss:full_clas32_short} studying and partially classifying groups generated by an invertible Mealy (3,2)-automata. 

In general, isotopisms do not respect automaton structure. The automata in Examples~\ref{ex:runningexample-isotopism-rho},~\ref{ex:runningexample-isotopism-kappa},~\ref{ex:runningexample-isotopism-delta} are not isomorphic to the automaton in Example~\ref{ex:runningexample1}.  It is merely coincidence that of Example~\ref{ex:runningexample-isotopism-epsilon} is; however, this does illustrate that this possibility can arise.  For enumerative purposes, it is helpful to find isotopy class representatives and compute the stabilizer of the representative in the isotopy group. 
The isotopisms $((1,2), \iota_n, (1,2), (2,3))$ and $(\iota_m, (1,2), \iota_m, (2,3))$ stabilize the objects in Example~\ref{ex:runningexample1}.

It is well-known that a Latin square can be viewed as the Cayley table of a quasigroup~\cite[Section 1.3]{Denes+Keedwell-LSa_SecondEdition_2015}.
In general, an isotopism does not preserve the quasigroup structure.  
While arising from a very different construction, automaton groups associated with isotopic automata are generally not isomorphic, as shown in the example below. 
For more on quasigroups, see, for example,~\cite{SmithRomanowsak:PMA,Smith:IQGR}.

\begin{example}
\label{ex:bellaterra_aleshin}
Let $\A$ be the Aleshin bireversible (3,2)-automaton generating the free group $F_3$ of rank $3$ (see~\cite{vorobets:aleshin}). Then applying the isotopism $\bigl((1,2),(1,2),\iota_3,\iota_2\bigr)$ to the interleaved array corresponding to $\A$ produces the interleaved array corresponding to another bireversible automaton, called the Bellaterra automaton $\mathcal B_3$, which generates the free product $(\Z/2\Z)*(\Z/2\Z)*(\Z/2\Z)$ of three groups of order 2 (as proved by Muntyan and the second author,~\cite{nekrash:self-similar}). See Figure~\ref{fig:aleshin_bellaterra} for details.
\end{example}

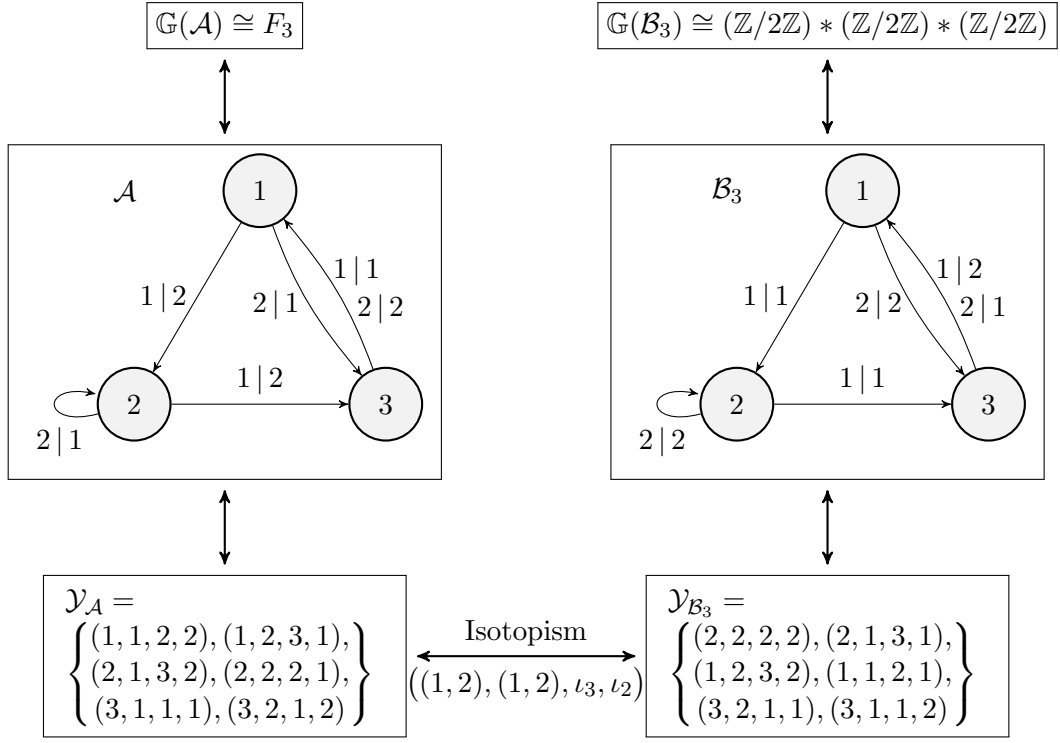
\begin{figure}[ht]
    \centering

\begin{tikzpicture}
  \matrix[ampersand replacement=\&, row sep=1cm, column sep=0.6cm] {
     \node[blocknf] (groupA) {\fbox{$\mathbb{G}(\A)\cong F_3$}}; \& 
 \& \& 
 %
    \node[blocknf] (groupB) {\fbox{$\mathbb{G}(\mathcal B_3)\cong (\Z/2\Z)*(\Z/2\Z)*(\Z/2\Z)$}};
\\
  %
    \node[blocknf] (autA) { \fbox{
    \begin{tikzpicture}[node distance=3cm, baseline={([yshift=-.8ex]current bounding box.center)}]
    \node[state] (q1) {$1$};
    \node[state, below left of=q1,xshift=13pt,yshift=-20pt] (q2) {$2$};
    \node[state, below right of=q1,xshift=-13pt,yshift=-20pt] (q3) {$3$};

    \draw (q1) edge [out=290,in=130]  node[left]{$2\,|\, 1$} (q3);
    \draw (q1) edge node[left]{$1\,|\, 2$} (q2);
    
    \draw (q2) edge node[above,yshift=-2pt]{$1\,|\,2$} (q3);
    \draw (q2) edge [loop left] node[below,xshift=2pt,yshift=-3pt]{$2\,|\, 1$}  (q2);

    \draw (q3) edge [out=110,in=310] node[right,xshift=-5pt,yshift=7pt]{$1\,|\,1$} node[right,xshift=4pt,yshift=-7pt]{$2\,|\,2$} (q1);
    
    \node[left=.9cm of q1] (A) {$\A$};
    
    \end{tikzpicture}
     }}; \& 
 \& \&
    \node[blocknf] (autB) {  \fbox{
    \begin{tikzpicture}[node distance=3cm, baseline={([yshift=-.8ex]current bounding box.center)}]
    \node[state] (q1) {$1$};
    \node[state, below left of=q1,xshift=13pt,yshift=-20pt] (q2) {$2$};
    \node[state, below right of=q1,xshift=-13pt,yshift=-20pt] (q3) {$3$};

    \draw (q1) edge [out=290,in=130]  node[left]{$2\,|\, 2$} (q3);
    \draw (q1) edge node[left]{$1\,|\, 1$} (q2);
    
    \draw (q2) edge node[above,yshift=-2pt]{$1\,|\,1$} (q3);
    \draw (q2) edge [loop left] node[below,xshift=2pt,yshift=-3pt]{$2\,|\, 2$}  (q2);

    \draw (q3) edge [out=110,in=310] node[right,xshift=-5pt,yshift=7pt]{$1\,|\,2$} node[right,xshift=4pt,yshift=-7pt]{$2\,|\,1$} (q1);
    
    \node[left=.9cm of q1] (A) {$\mathcal B_3$};
    
    \end{tikzpicture}
     }};\\
 %
     \node[blocknf] (arrA) {\fbox{$ \begin{array}{l}
      \mathcal{Y}_\A=\\
      \left\{\!\!\!\begin{array}{c} 
      ( 1, 1, 2, 2 ), ( 1, 2, 3, 1 ),\\ 
      (2, 1, 3, 2 ), ( 2, 2, 2, 1 ),\\ 
      ( 3, 1, 1, 1 ), ( 3, 2, 1, 2 )\end{array} \!\!\!\right\}                            
    \end{array}$ }}; \& 
 \& \& 
 %
    \node[blocknf] (arrB) {\fbox{$ \begin{array}{l}
      \mathcal{Y}_{\mathcal B_3}=\\
      \left\{\!\!\!\begin{array}{c} 
      ( 2, 2, 2, 2 ), ( 2, 1, 3, 1 ),\\ 
      ( 1, 2, 3, 2 ), ( 1, 1, 2, 1 ),\\
      ( 3, 2, 1, 1 ), ( 3, 1, 1, 2 )\end{array} \!\!\!\right\}                            
    \end{array}$ }};
\\
  };

\draw[connector,<->] (groupA) --  (autA);

\draw[connector,<->] (groupB) --  (autB);

\draw[connector,<->] (autA) --  (arrA);

\draw[connector,<->] (arrA) -- node[midway, above]{Isotopism} node[midway, below]{$\bigl((1,2),(1,2),\iota_3,\iota_2\bigr)$}(arrB); 

\draw[connector,<->] (arrB) --  (autB);

\end{tikzpicture}
    \caption{Isotopism between the Aleshin automaton $\A$ and the Bellaterra automaton $\mathcal B_3$.}
    \label{fig:aleshin_bellaterra}
\end{figure}

\addcontentsline{toc}{section}{References}
\bibliographystyle{alpha}
\bibliography{mylib}


\newpage
\section*{Notation}
\addcontentsline{toc}{section}{Notation}

We adopt some general notational conventions involving font choice:
\begin{itemize}
\setlength{\itemsep}{0pt}
  \setlength{\parskip}{0pt}
\item $\mathsf{Math\ Sans\ Serif}$ for families of objects
\item $\mathcal{MATH\ CALIGRAPHIC}$ for a single such family (and a few groups of permutations)
\item $MATH\ UPPER\ CASE$ for sub-objects when they are sets or matrices 
\item $math\ lower\ case$ for numbers/tuples of numbers/elements of sets
\item Greek letters are generally used for functions
\end{itemize}

\smallskip

\begin{tabular}{|l|l|ll}
\hline
\bf \large General & \\ \hline
$[d]$& the set $\{1,2,\ldots, d\}$\\
$m$, $n$ & fixed positive integers\\
$s$, $t$ & elements of $[4]$\\
$i$, $k$ & elements of $[m]$\\
$j$, $\ell$ & elements of $[n]$\\
\hline
\hline
\bf \large Functions & \\ \hline
$\mathsf{F}_d=\{f\colon[d]\rightarrow[d]\}$& monoid of functions from $[d]$ to $[d]$ with composition\\
$\mathsf{F}_d^e$ &  direct product of $\mathsf{F}_d$ with itself $e$ times \\
$\Sym(d)$ & symmetric group on $[d]$\\
$\Sym(d)^e$ & direct product of symmetric group on $[d]$ $e$-many times\\

\hline
\hline
\bf \large Automata & \\ \hline
$\mathsf{MA}(m,n)$ & set of $(m,n)$-automata (states $[m]$, letters $[n]$)\\
\quad $\A=(Q, X, \pi,\lambda)$ & Mealy automaton \\
        & \qquad
			\begin{tabular}{lll} 
               $Q$       & states, fixed order  & $\{q_1, \ldots q_m\}$ \\
			   $X$       & letters, fixed order & $\{x_1, \ldots x_n\}$ \\
			   $\pi$     & transition function  & $\pi\colon Q\times X\rightarrow Q$\\
			   $\lambda$ & output function      & $\lambda\colon Q\times X\rightarrow X$
            \end{tabular}\\
$\iz\A$ & Inverse of $\A$\\
$\dz\A$ & Dual of $\A$\\
\hline
\hline
\bf \large Arrays & \\ \hline
$\mathsf{IA}(m,n)$  &  $m\times n$ interleaved arrays \\ 
                    &  \quad subset of $[m]\times[n]\times[m]\times[n]$\\
$\mathsf{GA}(m,n)$  &  $m\times n$  grid arrays \\ 
                    &  \quad interleaved arrays of size $mn$ with property $\mathsf{U}(12)$.\\
\quad $\mathcal{Y} = \{ (i,j,k,\ell)\}$ & interleaved array (``Y'' is the last letter of array)\\
\quad $y_s$         & entry in position $s$ of $y\in \mathcal{Y}$\\
\quad $\mathcal{Y}_\take{st}$ & Restriction of elements of $\mathcal{Y}$ to entries $s$, $t$ \\
$\un{st}$ & Property of interleaved arrays: \\
   & pairs of entries in components $s$, $t$ cover $[m]\times[n]$ or $[n]\times[m]$.\\
\hline
\end{tabular}

\begin{tabular}{|l|l|ll}

\hline
\bf \large Partitions & \\ \hline
$\mathsf{P}(E,d)$  & Partitions  of $E$ into $d$-many cells, some possibly empty \\
$\mathsf{R}(S,m,n)$ & Rete \\
\hline
\bf \large Rete & \\ \hline
$\mathcal{P}$ & Rete\\
			&$\begin{array}{c|c|c|c} \mathcal{P}^1 & \mathcal{P}^2 & \mathcal{P}^3 &  \mathcal{P}^4 \\
			            \{P^1_i\}_{i\in[m]} &\{P^2_j\}_{j\in[n]} &\{P^3_k\}_{k\in[m]} & \{P^4_\ell\}_{\ell\in[n]}
               \end{array}$\\
\quad $\perp$ & orthogonal in the sense of collections of subsets     \\         
\hline
\hline
\bf \large Matrices & \\ \hline
$\mathsf{M}_{m\times n}(E)$ & $m\times n$, entries from $E$\\
$\mathsf{CL}(m,n)$  & Column-Latin matrices \\
$\mathsf{RL}(m,n)$  & Row-Latin matrices \\
$\mathsf{MM}(m,n)$  & Matched matrices $\mathsf{M}_{m\times n}([m])\times\mathsf{M}_{m\times n}([n])$\\
$\mathsf{CP}(m,n)$  & Cooperative pairs \\
\quad $(C,R)$  & matched pair of matrices\\
\quad $C^{(j)}$ & function defined by column $j$ of $C$\\
\quad $R_{(i)}$ & function defined by row $i$ of $R$\\

\hline
\bf \large Matrix Operations & \\ \hline
  $ M^T$ & transpose of $M$\\
  $\colmult$ & column composition product of $\mathsf{M}_{m\times n}([m])$\\
  \quad $\ic$ & inverse for $\colmult$\\
  \quad $\Idc$ & Identity for $\colmult$\\
  $\rowmult$ & row  composition product of $\mathsf{M}_{m\times n}([n])$\\
  \quad $\ir$ & inverse for $\rowmult$\\
  \quad $\Idc$ & Identity for $\rowmult$\\
  $\mmmult$ & product of $\mathsf{MM}(m,n)$:
$(C_1,R_1) \mmmult  (C_2,R_2)=(C_1 \colmult  C_2,  R_1\rowmult R_2)$\\
  $\ltacti$ & left action: columns  of $\mathsf{M}_{m\times n}([m])$ on $\mathsf{M}_{m\times n}([n])$\\
  $\rtacti$ & right action: rows of $\mathsf{M}_{m\times n}([n])$ on $\mathsf{M}_{m\times n}([m])$\\
  $\ltact$   &  left action: columns of $\mathsf{CL}(m,n)$ on $\mathsf{M}_{m\times n}([n])$\\
  $\rtact$   & right action: rows of $\mathsf{RL}(m.n)$ on $\mathsf{M}_{m\times n}([m])$\\

\hline

\end{tabular}

\begin{tabular}{|l|l|ll}
\hline

\bf \large Parastrophy & \\ \hline
$\mathfrak{P}$ & set of parastrophisms\\
$\mathfrak{P}^\ast$ & set of order preseverving parastrophies\\
\quad $\sharp$ & permute components $(2,4)$\\
\quad $\dag$    & permute components $(1,2)(3,4)$\\
\quad $\mathcal{Y}^\flat$ & apply parastrophism $\flat$ to every element of $\mathcal{Y}$\\
\hline
\hline
\bf \large Isotopy & \\ \hline
$\mathfrak{I}$ & set of isotopisms ($m$, $n$ implicit)\\
\quad $\sigma=(\rho,\kappa, \delta,\epsilon)$ & isotopism\\
 & \begin{tabular}{ll}
     $\rho$ & permutation of $[m]$/first component of $y$  \\
     $\kappa$ &  permutation of $[n]$/second component of $y$ \\
     $\delta$ &  permutation of $[m]$/third component of $y$ \\
     $\epsilon$ &  permutation of $[n]$/fourth component of $y$ 
 \end{tabular}\\
\hline 
\end{tabular}

\end{document}